\documentclass[aps,pra,groupedaddress,superscriptaddress,notitlepage,twocolumn,longbibliography,nofootinbib]{revtex4-2}

\usepackage{amsmath,amssymb,amsthm,bm}
\usepackage{graphicx}
\usepackage[colorinlistoftodos]{todonotes}
\usepackage[inline]{enumitem}
\definecolor{darkblue}{rgb}{0.1,0.2,0.6}
\definecolor{darkred}{rgb}{0.8,0.1,0.2}
\definecolor{crimson}{RGB}{164,16,52}
\definecolor{darkgreen}{rgb}{0.31,0.62,0.24}
\usepackage[colorlinks=true, allcolors=crimson]{hyperref}
\usepackage{diagbox}
\usepackage{epigraph}
\usepackage{float}
\usepackage{multirow}
\usepackage{tikz}
\usepackage{mathtools}
\usepackage{braket}
\usepackage{soul} 
\usepackage[all]{xy}
\usepackage{lipsum}

\usepackage{bbm}

\usepackage{tabularx}
\newcommand{\rarrow}{\quad\Longrightarrow\quad}
\newcommand{\abs}[1]{\left\lvert#1\right\rvert}
\newcommand{\norm}[1]{\left\lVert#1\right\rVert}
\newcommand{\tr}[1]{\mathrm{Tr}\left(#1\right)}
\newcommand{\cM}{\mathcal{M}}
\newcommand{\cB}{\mathcal{B}}
\newcommand{\cS}{\mathcal{S}}
\newcommand{\cR}{\mathcal{R}}
\newcommand{\tB}{\tilde{\mathcal{B}}}
\newcommand{\tps}{\widetilde{\psi}}
\newcolumntype{Y}{>{\centering\arraybackslash}X}
\newcolumntype{M}[1]{>{\centering\arraybackslash}m{#1}}
\newcommand{\upward}[2]{\multirow{1}{*}[#1 em]{#2}}
\newtheorem{asmp}{Assumption}
\newcommand{\chao}[1]{{\color{magenta} Chao:~#1}}

\newcommand{\Rom}[1]{\uppercase\expandafter{\romannumeral#1}}

\newcommand{\ex}[1]{\left\langle #1 \right\rangle}

\newcommand{\proj}[1]{|#1\rangle \langle #1|}

\newcommand{\mc}{\mathcal}

\DeclareMathOperator{\Tr}{Tr}

\makeatletter
\renewcommand*\env@matrix[1][*\c@MaxMatrixCols c]{%
	\hskip -\arraycolsep
	\let\@ifnextchar\new@ifnextchar
	\array{#1}}
\makeatother

\newtheorem{theorem}{Theorem}
\newtheorem{lemma}{Lemma}

\newtheorem*{claim*}{Claim}

\usepackage[normalem]{ulem}

\newcommand{\comment}[1]{}

\begin{document}
\title{Mixed-State Entanglement Measures in Topological Order}
\author{Chao Yin}
\affiliation{Department of Physics and Center for Theory of Quantum Matter, University of Colorado, Boulder, Colorado 80309, USA}
\author{Shang Liu}
\email{sliu.phys@gmail.com}
\affiliation{Kavli Institute for Theoretical Physics, University of California, Santa Barbara, California 93106, USA}

\begin{abstract}
    Quantum entanglement is a particularly useful characterization of topological orders which lack conventional order parameters. 
    In this work, we study the entanglement in topologically ordered states between two arbitrary spatial regions, using two distinct mixed-state entanglement measures: the so-called “computable cross-norm or realignment” (CCNR) negativity, and the more well-known partial-transpose (PT) negativity. 
    We first generally compute the entanglement measures: We obtain general expressions both in (2+1)D Chern-Simons field theories under certain simplifying conditions, and in the Pauli stabilizer formalism that applies to lattice models in all dimensions. 
    While the field-theoretic results are expected to be topological and universal, the lattice results contain nontopological/nonuniversal terms as well. This raises the important problem of continuum-lattice comparison which is crucial for practical applications. 
    When the two spatial regions and the remaining subsystem do not have triple intersection, we solve the problem by proposing a general strategy for extracting the topological and universal terms in both entanglement measures. Examples in the (2+1)D $\mathbb{Z}_2$ toric code model are also presented. 
    In the presence of trisection points, however, our result suggests that the subleading piece in the PT negativity is not topological and depends on the local geometry of the trisections, which is in harmonics with a technical subtlety in the field-theoretic calculation.
\end{abstract}

\maketitle


\section{Introduction}

Gapped phases of matter with topological order have gained recent experimental advances \cite{topo_SC,topo_Rydberg,Google2022TCDefect,Dreyer2023TC,quantinuum_nonAbel23}, and are useful for many quantum information tasks due to the long-range entanglement in the quantum many-body state \cite{QI_meet_QM}.
It is then crucial to quantify such entanglement in topologically ordered states. As a breakthrough discovery, in 2D the subleading term in the entanglement entropy (EE) of a subregion contains universal information \cite{Hamma2005TCEE1,Hamma2005TCEE2,Kitaev2006TEE,Levin2006TEE,Fradkin2008TEE,Vishwanath2012TEE,Vishwanath2013TEEReview}. This so-called topological EE can be used to diagnose the underlying topological phase that lacks conventional order parameters, especially in recent quantum simulation experiments \cite{topo_SC,Dreyer2023TC}. Moreover, entanglement properties can even be used for classifying topological orders \cite{Kim2020EntBootstrapFusion,Kim2021EntBootstrapDW}. 

\comment{
There has been a tremendous amount of studies on the entanglement properties of topological orders since the discovery of topological entanglement entropy (EE) \cite{Hamma2005TCEE1,Hamma2005TCEE2,Kitaev2006TEE,Levin2006TEE}. In a 2D topologically ordered ground state, the subleading term in the EE of a subregion contains universal information \cite{Kitaev2006TEE,Levin2006TEE,Fradkin2008TEE,Vishwanath2012TEE,Vishwanath2013TEEReview} and can be used to diagnose the underlying topological phase that lacks conventional order parameters. 
}

As natural generalizations of EE, two spatial regions in a tripartite pure state share entanglement that is often quantified by the partial transpose (PT) negativity $\mathcal{E}_{\mathrm{PT}}$ \cite{Peres1996PT,Vidal2002Negativity}. Such mixed state entanglement measure is also studied in topological orders \cite{Vidal2013TCNeg,Castelnovo2013TCNeg,Wen2016NegBdry,Wen2016NegSurgery,Mulligan2021Disentangling,Ryu2022VertexStates,LiuCC2022IQHNeg}. However, these case-by-case studies work with either specific lattice models, or a continuum field theory that is applied to specific tripartitions and quantum states. A general prescription for comparing the continuum and lattice results is also lacking. Moreover, the continuum approach becomes tricky if the tripartition contains a trisection point \cite{Ryu2022VertexStates}, i.e., a point where the three parties meet. There, the subleading term of $\mc E_{\rm PT}$ may depend on the local trisection geometry, as computed for integer quantum Hall states \cite{LiuCC2022IQHNeg} following the studies of corner contributions to EE \cite{Sierra2010IQHCornerEE,Sirois2021IQHCornerEE}.


In this work, we tackle these problems for $\mathcal{E}_{\mathrm{PT}}$, together with another entanglement measure called the \emph{computable cross norm or realignment} (CCNR) negativity $\mathcal{E}_{\mathrm{CCNR}}$ \cite{Rudolph2005CCNR,Chen2002CCNR,Rudolph2003CCNRProperties}. These two measures quantify different kinds of mixed-state entanglement in general, and are easy to compute from the reduced density matrices, unlike many other entanglement measures \cite{rmp_entan09}. Although $\mathcal{E}_{\mathrm{PT}}$ is widely used in literature, the CCNR negativity has gained recent interest in quantum many-body systems \cite{Aubrun2012RandStateCCNR,Collins2016RMT,Yin2022CFTCCNR,Liu2022CCNR,holo2022CCNR} due to its nice properties. For example, $\mathcal{E}_{\mathrm{CCNR}}$ is related to entanglement quantities with a nice holographic dual \cite{Dutta2021ReflectedEntropy}, and has a simpler topological structure \cite{Yin2022CFTCCNR} than $\mathcal{E}_{\mathrm{PT}}$ in (1+1)D. Intriguingly, we will also find a difference between the two measures in topological orders, namely $\mathcal{E}_{\mathrm{CCNR}}$ is ``more topological'' than $\mathcal{E}_{\mathrm{PT}}$ in certain cases.

To calculate the entanglement measures, we take both a continuum approach using (2+1)D Chern-Simons (CS) theories \cite{Witten1989CSTheory} and a lattice approach using the stabilizer formalism \cite{NielsenChuangBook}. 
These two methods are complementary to each other: the former can deal with much more general topologically ordered states in (2+1)D, while the latter applies to all stabilizer states beyond 2D topological orders and to trisection points. 

In the continuum, after presenting several concrete examples, we give general formulas for the two negativities for a large class of tripartitions and Wilson line (WL) configurations, including most of the cases in Refs.\,\onlinecite{Wen2016NegBdry} and \onlinecite{Wen2016NegSurgery} as specific examples. Our derivation utilizes the locality properties of the surgery method \cite{Witten1989CSTheory}, and reproduces expressions of EE as well. On lattices, we derive general formulas for the two entanglement measures in the stabilizer formalism, and apply them to the $\mathbb{Z}_2$ toric code (TC) model. Although the field-theoretic results are expected to be topological and universal, the lattice results contain nonuniversal terms, so we proceed to investigate the crucial problem of continuum-lattice comparison. In the absence of trisection points, we explain how to extract the topological and universal terms in both negativities on lattices. For trisection points, the subleading piece in the lattice PT negativity seems to be not topological and depends on the local geometry of trisection. This is in harmonics with a technical problem encountered in the continuum approach. The CCNR negativity, however, is topological at least in the lattice example we consider and computable in CS theories. 

The rest of this paper is organized as follows. We review some preliminaries in Section \ref{sec:Preliminaries}, including the various entanglement/correlation measures, and the CS theory. In Sections \ref{sec:CSEnt} and \ref{sec:StabilizerEnt}, we compute the mixed-state entanglement measures in CS theories and in stabilizer states, respectively. In Section \ref{sec:LatticeExamples}, we explain how to compare lattice and continuum results, and examine a few examples in the TC model. We conclude and comment on possible future directions in Section \ref{sec:Discussions}. Some technical calculations in CS theories are elaborated in the Appendix.

\section{Preliminaries}\label{sec:Preliminaries}

\subsection{Entanglement and correlation measures }

Given a pure state $\ket{\psi}$ shared by two parties $A$ and $B$, we define the reduced density matrix $\rho_A=\Tr_B(\ket{\psi}\bra{\psi})$, and the Renyi EE \begin{align}
    S^{(n)}_A=(1-n)^{-1}\log\Tr(\rho_A^n),
\end{align}
which characterizes the entanglement between the two parties. The base of the logarithm is an arbitrary positive number and is fixed once for all. $n$ can be any nonnegative real number, and typically one calculates for integer $n$ and then analytic continue to generic $n$. Setting $n\rightarrow 1$ yields the von Neumann entanglement entropy $S_A=-\Tr(\rho_A\log\rho_A)$. 

We further divide $A$ into $A_1$ and $A_2$, and question about the entanglement between them. Let $\{ \ket{c} \}$ and $\{ \ket{d} \}$ be orthonormal bases of $A_1$ and $A_2$, respectively. We define the partial transposed density matrix $\rho_A^{T_2}$, and a realignment matrix $R_\rho$ according to the following equation. 
\begin{equation}
	\raisebox{-5 pt}{\includegraphics{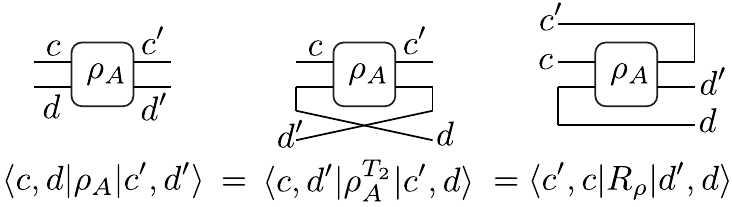}}
\end{equation}
The PT and CCNR negativity are then defined as \begin{subequations} \begin{align}
    \mc E_{\rm PT}& =\log \norm{\rho_A^{T_2}}_1, \\
    \mc E_{\rm CCNR}& =\log\norm{R_\rho}_1,
\end{align}
\end{subequations}
respectively, where the $1$-norm is given by $\norm{R_\rho}_1:=\Tr\sqrt{R_\rho^\dagger R_\rho}$. 
When $\rho_A$ is separable/unentangled, $\mc E_{\rm PT}=0$ and $\mc E_{\rm CCNR}\leq 0$. Therefore, either $\mc E_{\rm PT}>0$ or $\mc E_{\rm CCNR}> 0$ implies $\rho_A$ to be entangled, and these two quantities can be regarded as mixed state entanglement measures. For example, $\mc E_{\rm PT}$ upper bounds how many Bell pairs can be distilled from the state \cite{Vidal2002Negativity}. It is known that these two separability criteria are not comparable \cite{Chen2002CCNR,Rudolph2003CCNRProperties} -- neither one is stronger than the other. More precisely, there are entangled states with $\mc E_{\rm PT}>0$ but $\mc E_{\rm CCNR}\leq 0$, and vice versa. 

Both $\mc E_{\rm PT}$ and $\mc E_{\rm CCNR}$ can be computed from replica tricks: \begin{subequations} \begin{align}
    \mc E_{\rm PT} &=\lim_{\text{even } n\rightarrow 1}\log\Tr[(\rho_A^{T_2})^{n}], \label{eq:PT_repli} \\
    \mc E_{\rm CCNR} &=\lim_{\text{even } n\rightarrow 1}\log\Tr[(R_\rho^\dagger R_\rho)^{n/2}], \label{eq:CCNR_repli}
\end{align}
\end{subequations}
where one analytically continues from \emph{even} integer $n$. Note that the two equations are consistent because $\left(\rho_A^{T_2}\right)^\dagger=\rho_A^{T_2}$.
The Renyi mutual information $I^{(n)}(A_1,A_2)=S^{(n)}_{A_1}+S^{(n)}_{A_2}-S^{(n)}_{A}$ is another useful characterization of bipartite correlations. It contains both classical and quantum correlations and hence does not lead to a separability criterion. Nevertheless, the mutual information can be easily computed using EE results, which we will also derive as byproducts.

\subsection{Chern-Simons theory and surgery method}\label{sec:CS_intro}

Here we briefly review Chern-Simons Theory, which is the most important class of topological quantum field theory that describes topological orders. We refer to the classic paper \cite{Witten1989CSTheory} for details. A CS theory living in a 3D manifold $M$, is given by a partition function $Z(M)=\int [\mathcal{D} A] \mathrm{e}^{\mathrm{i} S} $ with action \begin{equation}
    S = \frac{k}{4\pi} \int_M \tr{A\wedge \mathrm{d} A + \frac{2}{3} A\wedge A\wedge A },
\end{equation}
where $A$ is a gauge field of some gauge group $G$, $k$ is an integer called level. This theory is topological because the action does not depend on the metric. As a gauge theory, operators that are gauge invariant are defined by Wilson lines (WLs) \begin{equation}
    W_R^C(A) = \mathrm{Tr}_R\left(\mathcal{P} \exp \int_C A \right),
\end{equation}
where $C$ is an oriented closed curve in $M$, $\mathcal{P}$ stands for path-ordering, and $R$ is an irreducible representation of group $G$. Inserting WLs in the path integral yields the physically interesting correlation function of them \begin{equation}\label{eq:ZMRR}
    Z\left(M;W_{R_1}^{C_1},W_{R_2}^{C_2},\cdots\right)=\int [\mathcal{D} A] \mathrm{e}^{\mathrm{i} S} W_{R_1}^{C_1}W_{R_2}^{C_2}\cdots.
\end{equation}
Different WLs do not intersect, but can form links and knots.

The 2D boundary $\partial M$ of $M$ hosts a Hilbert space, where each quantum state is defined by the path integral in the interior of $M$ with a particular WL configuration. For example, consider $M=D^3$ is a solid ball, with a sphere boundary $\partial M=S^2$. If there are no WLs inside $M$, the corresponding 2D state is the vacuum of the theory. If there is one WL of representation $R_a$ that punctures the boundary at two points, as shown in Fig.~\ref{fig:SpherePartitions}(a), the sphere is in the unique quantum state consisting a pair of anyons $a,\bar a$ that are dual to each other. Setting $a$ to be the trivial representation (which we denote by $0$) corresponds to the previous vacuum case of no anyons. Inserting more WLs create more anyons, which can fuse and braid with each other.

\begin{figure}
    \centering
    \includegraphics[width=0.48\textwidth]{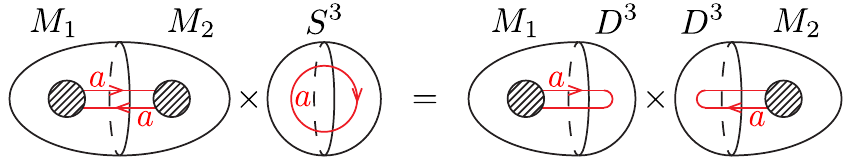}
    \caption{In the surgery method, the manifold $M=M_1\cup M_2$ can be cut along an $S^2$ cross section traversed by two WLs of representation $a$, using a $S^3$ with a Wilson loop of the same representation $a$ in it. The result is two disconnected manifolds $M_L=M_1\cup D^3$ and $M_R=D^3\cup M_2$ with WLs reconnected. The case where no WL traverses the cross section is just setting $a=0$.}
    \label{fig:surgery}
\end{figure}

Partition function \eqref{eq:ZMRR} for any manifold with complicated WL insertions is in principle calculable by the surgery method. The idea is basically to cut the manifold into smaller and simpler pieces. For the simplest case, a single closed WL in 3D sphere $M=S^3$ yields \begin{equation}\label{eq:Za=S0a}
    Z_a := Z(S^3;R_a) = \mathcal{S}_{0a},
\end{equation}
where $\mathcal{S}_{ab}$ is the modular $\mathcal{S}$ matrix of the theory, a unitary matrix indexed by anyon types $\{0,a,b,\cdots\}$ \footnote{$\mc S$ matrix is defined both in the theory of anyons \cite{Kitaev2006Honeycomb,BondersonThesis} and in rational conformal field theories \cite{CFTBook}. CS theories connect these two contexts together. }. If $M$ has a cross section $S^2$ traversed by a single pair of WLs (on the left of Fig.~\ref{fig:surgery}), it can be cut by surgery as shown in Fig.~\ref{fig:surgery}. More precisely, the process introduces three additional manifolds: an $S^3$ with a WL of representation $R_a$ given by \eqref{eq:Za=S0a}, and $M_L,M_R$ that are the two manifolds after cutting the $S^2$ cross section (and reconnecting the cutted WL in each manifold). The corresponding partition functions satisfy \begin{equation}\label{eq:ZZ=ZZ}
    Z(M)\cdot Z_a = Z(M_L) \cdot Z(M_R).
\end{equation}
$M_L,M_R$ are potentially simpler manifolds than $M$, and can be further simplified by surgery \eqref{eq:ZZ=ZZ} on each of them. This process may ultimately reduce the original $M$ problem to a bunch of $S^3$s with known partition function \eqref{eq:Za=S0a}.


The intuition behind \eqref{eq:ZZ=ZZ} is as follows. The partition function can be viewed as an inner product \begin{equation}
    Z(M) = \braket{\psi_L|\psi_R},
\end{equation}
where $\ket{\psi_L},\ket{\psi_R}$ are the left/right states on the cut interface $S^2$. Since the Hilbert space of $S^2$ with one pair of anyons is one-dimensional, the two states $\ket{\psi_L},\ket{\psi_R}$ simply differ by a complex number prefactor (not just a phase because they are not normalized).  On the other hand, a $S^3$ with a Wilson loop $a$ is also an inner product $Z_a = \braket{\psi_a|\psi_a}$, where the state $\ket{\psi_a}$ lives in the same Hilbert space as $\ket{\psi_L},\ket{\psi_R}$. Using this reference state, \eqref{eq:ZZ=ZZ} simply comes from the identity \begin{equation}
    \braket{\psi_L|\psi_R} \braket{\psi_a|\psi_a} = \braket{\psi_L|\psi_a} \braket{\psi_a|\psi_R},
\end{equation}
for any three states in an one-dimensional Hilbert space.

We will calculate entanglement measures in CS theory using this surgery method. The results will be expressed by the quantum dimensions $d_a$ of anyon type $a$ defined by \begin{equation}
    d_a = \mathcal{S}_{0a}/ \mathcal{S}_{00},
\end{equation}
and the total quantum dimension \begin{equation}
    \mathcal{D} = \sqrt{\sum_a d_a^2},
\end{equation}
which equals $(\mc S_{00})^{-1}$ due to unitarity of $\mathcal{S}_{ab}$.
Intuitively, $d_a\geq 1$ may be understood as the Hilbert space dimension shared by each anyon in the limit of many anyons: In the presence of $M\gg 1$ number of type-$a$ anyons, the dimension of the low-lying degenerate subspace scales as $d_a^M$ \cite{Nayak2008RMP}. A topological order is called abelian if $d_a=1$ for all $a$, and is called nonabelian if $d_a>1$ for some $a$. 

\section{Entanglement Measures in Chern-Simons Theories}\label{sec:CSEnt}

In this section we compute the negativities using the surgery method in (2+1)D CS theories. As discussed in Section \ref{sec:CS_intro}, the quantum state lives on the surface of a $3$D manifold, and is defined by path integral in the interior. The surface is tripartitioned to three parties $A_1,A_2,B$, and we compute the PT and CCNR negativities $\mathcal{E}^{\rm top}_{\rm PT},\mathcal{E}^{\rm top}_{\rm CCNR}$ between $A_1$ and $A_2$. Here we use superscript ``$\rm top$'' because CS theory is topological, and to differentiate with latter results that also contain nontopological contributions.

\subsection{Results for simple cases}


\begin{figure}
	\centering
	\includegraphics{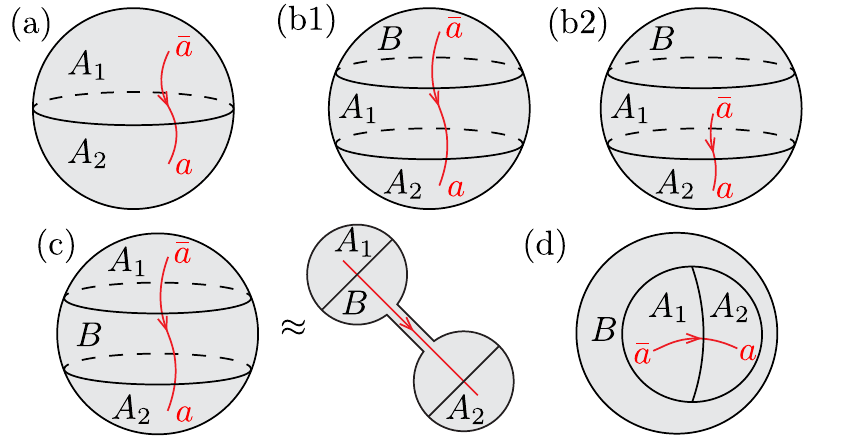}
	\caption{Wave function $\ket{\psi}$ on a sphere with different partition and WL configurations. }
	\label{fig:SpherePartitions}
\end{figure}

We first present results for several examples, before sketching the derivation and general formulas. 
Let us first consider the state on a sphere $S^2$ consisting of a pair of dual anyons $a,\bar a$. The state is uniquely given by the partition function in the solid ball $D^3$ with a WL of representation $a$ puncturing the surface. Setting $a$ to be trivial corresponds to the case of no anyons.  We consider several different partition scenarios as shown in Fig.\,\ref{fig:SpherePartitions}, and the negativities between $A_1$ and $A_2$ are given in Table~\ref{Table:SpherePartitions}, where the PT negativity results are mostly directly from Ref.\,\onlinecite{Wen2016NegSurgery}. We observe that the CCNR negativity depends on all boundary sections of $A_1$ and $A_2$, not just their interface, as opposed to the PT negativity. For example, cases b1 and b2 have the same interface between $A_1$ and $A_2$ (one $D^2$ traversed by one WL $a$), while the interface between $B$ and $A_1$ is different: b1 is traversed by WL $a$ while b2 is not. As a result, the two cases have the same PT negativity, but different CCNR negativity. The PT negativity result for the trisection scenario in Fig.\,\ref{fig:SpherePartitions}(d) is absent due to a subtle technical problem: With the replica approach, one encounters topological spaces that are not manifolds -- suspensions of tori with many handles. Thus the surgery method does not apply here. This casts doubt on the topological nature of $\mc E_{\rm PT}$, which we will return to in later sections using lattice methods.

\begin{table}
	\centering
	\begin{tabular}{|c|c|c|}
		\hline
		{\bf Case} & $\boldsymbol{\mc E_{\rm CCNR}^{\rm top}}$ & $\boldsymbol{\mc E_{\rm PT}^{\rm top}}$ \\
		\hline
		
		{a} & $\log d_a-\log\mc D$ & $\log d_a-\log\mc D$\\
		
		{b1} & $(\log d_a-\log\mc D)/2$ & $\log d_a-\log\mc D$ \\
		
		{b2} & $\log d_a-(1/2)\log\mc D$ & $\log d_a-\log\mc D$ \\
		
		{c} & $\log\mc D-\log d_a$ & $0$ \\
		
		{d} & $\log d_a$ & $*$\\
		\hline
		
	\end{tabular}
	\caption{Topological CCNR and PT negativities between $A_1$ and $A_2$ for the sphere states in Fig.\,\ref{fig:SpherePartitions}, computed by the surgery method of CS theory. Most PT negativity results are directly from Ref.\,\onlinecite{Wen2016NegSurgery}.  }
	\label{Table:SpherePartitions}
\end{table}

We analogously compute the negativities for torus states, where topological ground state degeneracy exists: The Hilbert space on a torus $T^2=S^1\times S^1$ without anyon excitations is multidimensional. An orthonormal basis of this Hilbert space, denoted as $\{ \ket{R_a} \}$, 
corresponds to a finite set of irreducible representations of the gauge group. The state $\ket{R_a}$ can be prepared by performing path integral on the solid torus (bagel) $D^2\times S^1$ inserted by a noncontractible Wilson loop carrying the corresponding representation $R_a$; see Fig.\,\ref{fig:TorusPartitions}. It is an eigenstate of Wilson loop operators along the perpendicular direction \cite{Vishwanath2012TEE}.
We consider a general state $\ket{\psi}=\sum_a\psi_a\ket{R_a}$ with the normalization condition $\braket{\psi|\psi}=\sum_a|\psi_a|^2=1$. 
In Table~\ref{Table:TorusPartitions}, we present results on both the CCNR and PT negativities for the tripartitions shown in Fig.\,\ref{fig:TorusPartitions}(a)-(d). The $\psi_a$ dependence is proposed to distinguish abelian and nonabelian topological orders \cite{Wen2016NegBdry,Wen2016NegSurgery}.

\begin{figure}
	\centering
	\includegraphics{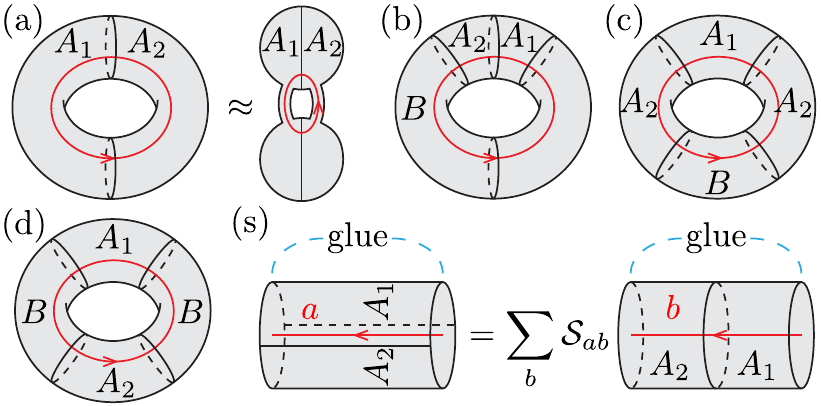}
	\caption{(a)-(d) Wave function $\ket{\psi}=\sum_a\psi_a\ket{R_a}$ on a torus with different partitions. (s) Effect of the $\mc S$ transformation. }
	\label{fig:TorusPartitions}
\end{figure}

\begin{table}
	\centering
	\begin{tabular}{|c|c|c|}
		\hline
		{\bf Case} & $\boldsymbol{\mc E_{\rm CCNR}^{\rm top}}$ & $\boldsymbol{\mc E_{\rm PT}^{\rm top}}$ \\
		\hline
		
		{a} & $2\log(\sum_a|\psi_a|d_a)-2\log\mc D$ &  $2\log(\sum_a|\psi_a|d_a)-2\log\mc D$ \\
		
		{b} & $0$ & $\log(\sum_a|\psi_a|^2d_a)-\log\mc D$ \\
		
		{c} & $\log(\sum_a|\psi_a|^2d_a)-\log\mc D$ & $\log(\sum_a|\psi_a|^2d_a^2)-2\log\mc D$\\
		
		{d} & $\log(\sum_a|\psi_a|^2/d_a^2)+2\log\mc D$ & $0$\\
		\hline
		
	\end{tabular}
	\caption{Topological CCNR and PT negativities between $A_1$ and $A_2$ for the torus states in Fig.\,\ref{fig:TorusPartitions}, computed by the surgery method of CS theory. The PT negativity results are directly from Ref.\,\onlinecite{Wen2016NegSurgery}. }
	\label{Table:TorusPartitions}
\end{table}

\subsection{Surgery calculation: An example}
We calculate one example above in full detail, and the other cases follow similarly. Moreover, this calculation will motivate more general results. The first step is to deform the 3D manifold to a bunch of balls $D^3$ connected by tubes, see examples Fig.~\ref{fig:SpherePartitions}(c) and Fig.~\ref{fig:TorusPartitions}(a), such that each edge between different parties are only supported on the balls, and each ball hosts at most one edge. A ball is called an \emph{edge ball} if it contains one edge, and a \emph{party ball} if it contains no edge and belongs to one party. Each tube connecting the balls has cross section $D^2$ (so we will call them $D^2$-tubes), and belongs to a single party.

\begin{figure}[h]
    \centering
    \includegraphics[width=0.48\textwidth]{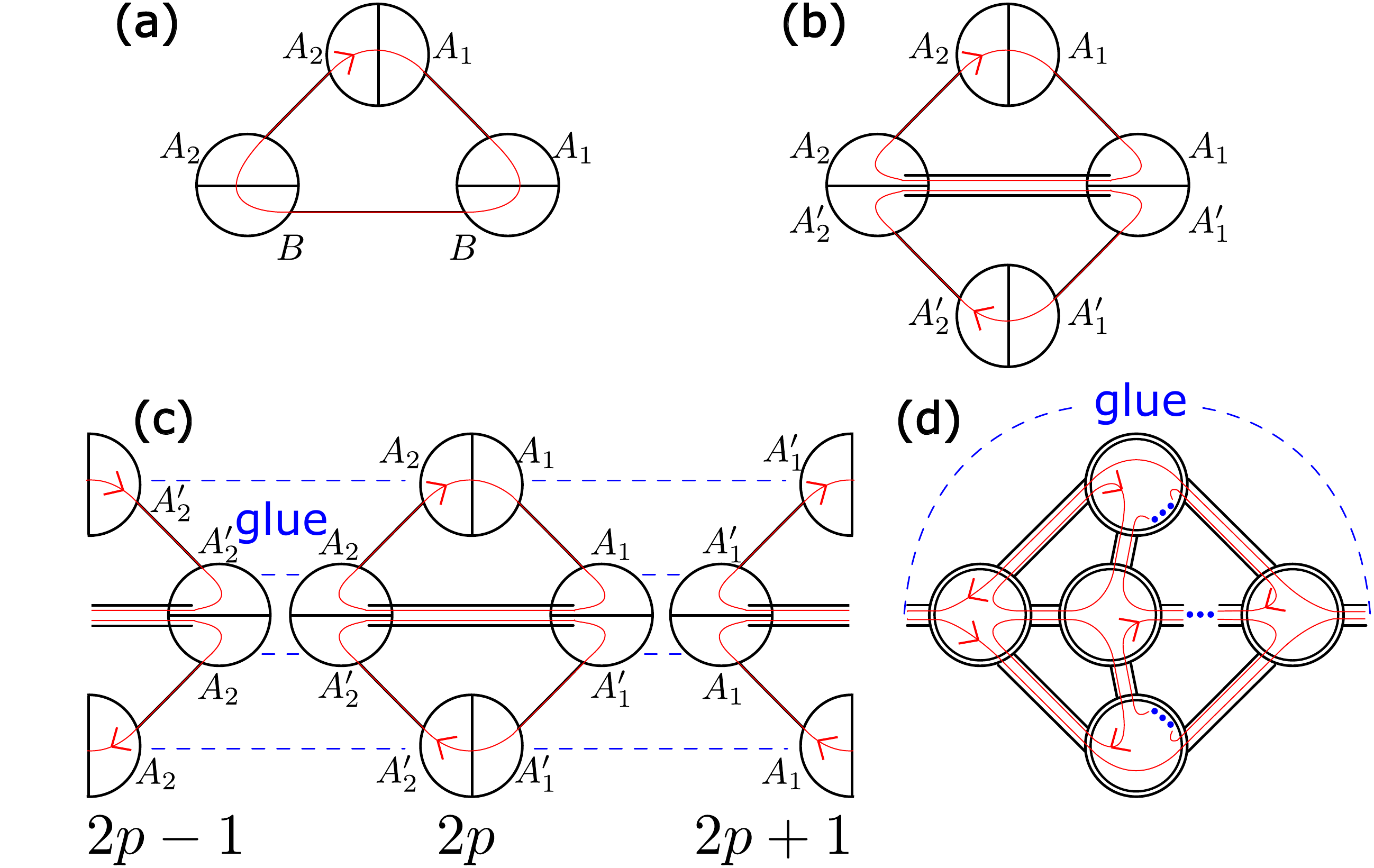}
    \caption{Calculation of CCNR negativity for Fig.~\ref{fig:TorusPartitions}(b). (a) The ball-tube representation of $\ket{\psi}$.  Here a circle represents a ball $D^3$, with an equator if it is an edge ball. A $D^2$-tube is shown by a single black straight line. (b) The path integral representation of $\rho_A$, where two parallel black straight lines represent a $S^2$-tube, which comes from gluing two $D^2$-tubes.  (c) The way to glue $n$ copies of $\rho_A$ together to get $\Tr[(R_\rho^\dagger R_\rho)^{n/2}]$, with the result in (d), where the middle row contains $n$ $S^3$s. Here a double-line circle represents a $S^3$ that comes from gluing two balls $D^3$. }
    \label{fig:surgery_exp}
\end{figure}

Now focus on calculating CCNR for Fig.~\ref{fig:TorusPartitions}(b) with fixed WL representation $\ket{\psi}=\ket{R_a}$. The torus under tripartition is deformed into the ball-tube system symbolized in Fig.~\ref{fig:surgery_exp}(a), where three edge balls are connected by $D^2$-tubes in a circle. To trace over $B$ and get $\rho_A$, we take the torus with its conjugated replica (with orientation and WL direction reversed), and glue the $B$ regions together in Fig.~\ref{fig:surgery_exp}(b). In the ball-tube representation for $\ket{\psi}$, this can be done \emph{separately} for each ball: The $A_1A_2$ ball just doubles with another $A'_1A'_2$, while the $A_1B$ ($A_2B$ similar) ball glues with a $BA'_1$ to produce a $A_1A'_1$ ball. After gluing the balls, it remains to glue the two $D^2$-tubes belonging to $B$: the result is one tube with cross section $S^2$ (two $D^2$s glued together) that connects the interiors of ball $A_1A'_1$ and $A_2A'_2$. 

Next, we need to take $n$ (even) copies of $\rho_A$ and glue them such that the partition function yields $\Tr[(R_\rho^\dagger R_\rho)^{n/2}]$ in the replica trick. In this case, it is convenient to align the $n$ copies in a horizontal row (with periodic boundaries due to the trace), with the odd ones rotated in the paper by $180$ degree. Then we just glue each copy to the part of its two neighbors that face towards it, as shown in Fig.~\ref{fig:surgery_exp}(c). Still, we glue balls first before tubes. For all $p=1,\cdots,n/2$, the $A_1A'_1$ ($A_2A'_2$) ball of the $(2p)$-th copy glues with that of the $(2p+1)$-th ($(2p-1)$-th) copy, which produces a $S^3$. The $n$ balls shared by the two parties (either $A_1A_2$ or $A'_1A'_2$), however, are glued in two disconnected groups where each group produces a $S^3$. Finally, $D^2$-tubes are glued in pairs to $S^2$-tubes, while the $S^2$-tubes connecting $A_1A'_1$ with $A_2A'_2$ remain, and we arrive at a complicated topology Fig.~\ref{fig:surgery_exp}(d) with WLs threading inside. 

We then cut each $S^2$-tube using surgery \eqref{eq:ZZ=ZZ}, and the price is just invoking an extra factor $Z_a^{-1}$. 
Furthermore, after the cut, there is one Wilson loop of representation $a$ left in each $S^3$. Since there are $3n$ $S^2$-tubes and $n+2$ $S^3$s, and the partition function of disconnected manifolds factorize, we have \begin{equation}
    \Tr[(R_\rho^\dagger R_\rho)^{n/2}] = Z_a^{-3n} Z_a^{n+2} = Z_a^{2-2n}\rightarrow 1,
\end{equation}
in the limit $n\rightarrow 1$, yielding $\mc E^{\rm top}_{\rm CCNR}=0$ as given in Table~\ref{Table:TorusPartitions}.

We make two observations in the calculation above. First, each ball in $\ket{\psi}$ is glued with its replicas in a \emph{local} way independent of the other balls/tubes. Second, we merely count the number of $S^2$-tubes and $S^3$s in the end, and it does not matter which $S^3$ is connected to which in the \emph{global} topology of Fig.~\ref{fig:surgery_exp}(d). These properties also hold when the WL representation is not fixed, and lead to the general result below.

\subsection{General result}


The above observations enable calculation for a general tripartite state $\ket{\psi}$ satisfying the following assumptions:
\begin{itemize}
	\item The spacetime manifold $M$ has the form of solid balls connected by solid tubes ($D^2\times [0, 1]$), with at most one circular interface on each ball, and no interface on any tube. 
	\item  The WLs can be deformed, without any touching among themselves or between WL endpoints and the interface, to a configuration that all WLs are contained in the 2D surface $\partial M$, and each tube is traversed by the WLs at most once. 
\end{itemize}
The second assumption avoids braiding of WLs during surgery. The first assumption implies that the tripartition interface is a disjoint union of circles that are contractible in the 3D spacetime, and thus there is no trisection points as in Fig.\,\ref{fig:SpherePartitions}(d). All other examples above (Fig.\,\ref{fig:SpherePartitions}(a)-(c) and Fig.\,\ref{fig:TorusPartitions}(a)-(d)) satisfy the two assumptions; in particular, we illustrate the ball-tube systems in Fig.\,\ref{fig:SpherePartitions}(c) and Fig.\,\ref{fig:TorusPartitions}(a). Note that it is also possible to consider interfaces made of noncontractible loops, by applying a diffeomorphism on the torus surface, which acts on the Hilbert space as a linear transformation \cite{Witten1989CSTheory}. More concretely, we illustrate in Fig.\,\ref{fig:TorusPartitions}(s) a case where the assumptions are satisfied after the modular $\mc S$ transformation. 
We establish the following result. 
\begin{theorem}\label{thm:EntMeasureRelationsCS}
	Consider a CS theory under the above two assumptions with WLs (either closed or open) labeled by $w$. If each WL carries a definite representation, then the topological entanglement measures satisfy (superscripts ``top'' omitted)
	\begin{align}
        & S_P^{(n)} = -E_P \log \mathcal{D} + \textstyle{\sum_w} K_P(w) \log d_{a(w)}, \label{eq:EEFormula}\\
		&\mc E_{\rm CCNR}(A_1,A_2)=(S^{(n)}_{A_1}+S^{(n)}_{A_2})/2-S^{(n)}_A, \label{eq:CCNREntropyRelation}\\
		&\mc E_{\rm PT}(A_1,A_2)= I^{(n)}(A_1,A_2)/2, \label{eq:PTEntropyRelation}
	\end{align}
	for arbitrary $n$. Here in the first line, $P\in\{A_1,A_2,B\}$ is the party, $E_P$ is the number of interfaces shared by $P$, $K_P(w)$ is the number of interfaces shared by $P$ that WL $w$ traverses, and $a(w)$ is the representation of $w$.
\end{theorem}
\noindent 
See Appendix for more general results of indefinite representations and the proof. The underlying idea, already illustrated in the previous example, is that gluing replicas and cutting tubes in surgery are both \emph{local} operations: In the ball-tube system, each
ball (with its connected tubes) is operated separately to
yield a local factor, and the partition function of
the replicated manifold is just the product of all local
factors, independent of the \emph{global} topology on which ball
is connected to which \footnote{Strictly speaking, one needs to deal with an extra summation over representations in the general case; see Appendix~\ref{sec:WL}.}. We expect that our method can be further generalized, for example, to calculate CCNR negativity of tripartitions (which does not satisfy the aforementioned assumptions), since we observe that Eq.\,\ref{eq:CCNREntropyRelation} is also satisfied for the tripartition in Fig.\,\ref{fig:SpherePartitions}(d).


\subsection{In what sense are the negativities topological?}

In a generic (2+1)D field theory, similar to EE, we expect either negativity $\mc E$ for a smooth partition contains ``area-law'' terms proportional to the lengths of different interface sections \cite{Kitaev2006TEE,Levin2006TEE,Wen2016NegBdry}, together with a subleading piece $\mc E^{\rm top}$ that is topological (insensitive to deformations of the partition interface) and universal (insensitive to deformations of the action). 
More precisely, we propose
\begin{align}
    \mc E=\sum_i\beta_i l_i+\mc E^{\rm top}+\cdots
    \label{eq:NegGenericExpression}
\end{align}
where the index $i$ labels different interface sections, $\beta_i$ are nonuniversal coefficients, $l_i$ is the length of the $i$-th section, and ``$\cdots$'' contains terms of higher order in $l_i$. When the partition interface is not smooth, e.g. containing corners, there may be additional nontopological subleading terms in the above summation.

As shown above, negativities in CS theories are either not calculable by surgery for certain non-smooth partitions (e.g. PT negativity for the trisection Fig.~\ref{fig:SpherePartitions}(d)), or given by some pure topological result $\mathcal{E}^{\rm top}$ without the area-law terms. The latter is reasonable because as topological field theories, CS theories have no dependence on metrics. We \emph{conjecture} that the first case (failure of surgery calculation) implies the presence of nontopological terms, while the second case corresponds to the absence of such terms, justifying our notation $\mc E^{\rm top}$ for results computed in CS theory.
In this work, we will not compute the negativities for generic field theories, but we will verify this conjecture by comparing to certain lattice models. 
We will later explain how to define $\mc E^{\rm top}$ on lattice systems, and comment on a caveat about the so-called spurious long-range entanglement. 



\section{Entanglement Measures for Stabilizer States }\label{sec:StabilizerEnt}

As a complementary approach, we also compute the CCNR and PT negativities in a general Pauli stabilizer state. The results are useful for studying stabilizer lattice models, especially with trisection points. The base of logarithm will be set to $2$ for convenience. We consider a state $\ket{\psi}$ uniquely determined by a stabilizer group $G$. Here, $G$ is an abelian group of (multi-site) Pauli operators with $\pm$ signs such that $-1\notin G$, and $\ket\psi$ is the unique simultaneous eigenstate with eigenvalue $+1$ for all $g\in G$. Let $N$ be the total number of qubits in the Hilbert space. We necessarily have $|G|=2^N$, and we say the rank, or the number of generators, of $G$ is $N$. 

The EE in a stabilizer state can be conveniently computed \cite{Chuang2004StabilizerEE}. Let $A$ be a subsystem consisting $N_A$ number of qubits, and $G_A\subset G$ be the subgroup of stabilizers inside $A$. The full and reduced density matrices $\rho$ and $\rho_A$ are respectively given by
\begin{align}
	\rho:=\proj{\psi}=\frac{1}{2^N}\sum_{g\in G}g,\quad \rho_A
	=\frac{1}{2^{N_A}}\sum_{g\in G_A}g. 
\end{align}
Let $G_A$ have $k$ generators, or equivalently $|G_A|=2^k$, then the entropy of subsystem $A$ for any Renyi index is given by $S^{(n)}_A=N_A-k$ \cite{Chuang2004StabilizerEE}. 

Now we divide subsystem $A$ further into $A_1$ and $A_2$. We find the following result about the CCNR negativity that agrees with Eq.~\ref{eq:CCNREntropyRelation}. 
\begin{theorem}\label{thm:CCNRNegStabilizerState}
	For a stabilizer state, $\mc E_{\rm CCNR}(A_1,A_2)=(S^{(m)}_{A_1}+S^{(m)}_{A_2})/2-S^{(m)}_A$ for any $m$. 
\end{theorem}
\begin{proof}
For each $g\in G_A$, we can write $g=O^{A_1}\otimes O^{A_2}$ such that $O^{A_i}$ acts on subsystem $A_i$ and is a Pauli operator up to a sign. We then have
\begin{align}
\Tr\left[ (R_\rho^\dagger R_\rho)^{\frac{n}{2}} \right]=&\frac{1}{2^{n N_A}}\sum_{g_1,\cdots,g_{n}\in G_A}\nonumber\\
&\prod_{i=1}^{n/2}\Tr(O^{A_1}_{2i-1}O^{A_1}_{2i})\Tr(O^{A_2}_{2i}O^{A_2}_{2i+1})  
\label{eq:StabilizerCCNRReplica}
\end{align}
where $g_i=O^{A_1}_i\otimes O^{A_2}_i$. The summand above is nonzero if and only if $O^{A_1}_{2i-1}\propto O^{A_1}_{2i}$ and $O^{A_2}_{2i}\propto O^{A_2}_{2i+1}$ for all $i$ (identifying $n+1$ and $1$). We then only need to consider the case where $y_{2i-1,2i}:=g_{2i-1}g_{2i}\in G_{A_2}$ and $x_{2i,2i+1}:=g_{2i}g_{2i+1}\in G_{A_1}$. The summation over $g_1,\cdots,g_{n}$ in Eq.\,\ref{eq:StabilizerCCNRReplica} can then be replaced with a summation over $g_1$, $x_{2i,2i+1}$, and $y_{2i-1,2i}$. Note that the $x$'s and $y$'s are not all independent: Since $g_2=g_1y_{12},g_3=g_2x_{23},\cdots$, we have $g_1=g_1 y_{12}x_{23}y_{34}\cdots y_{n-1,n}x_{n,1}$, and thus 
\begin{align}
y_{12}y_{34}\cdots y_{n-1,n}=1,\quad x_{23}x_{45}\cdots x_{n,1}=1. 
\end{align}
Let $\sum'$ be the sum over the independent variables $g_1\in G_A$, $x_{23},\cdots,x_{n-2,n-1} \in G_{A_1}$, and $y_{12},\cdots,y_{n-3,n-2} \in G_{A_2}$. We have
\begin{align}
\Tr\left[ (R_\rho^\dagger R_\rho)^{\frac{n}{2}} \right]&=\frac{1}{2^{n N_A}}\sum\nolimits' \Tr_{A_1}(1)^{\frac{n}{2}}\Tr_{A_2}(1)^{\frac{n}{2}}\nonumber\\
&=2^{k+(\frac{n}{2}-1)(k_1+k_2)-\frac{n}{2}N_A},  
\end{align}
where $k_1$ and $k_2$ are the number of generators of $G_{A_1}$ and $G_{A_2}$, respectively. Taking the limit $n\rightarrow 1$, we obtain the CCNR negativity: 
\begin{align}
\mc E_{\rm CCNR}&=k-\frac{1}{2}(k_1+k_2)-\frac{1}{2}N_A\nonumber\\
&=\frac{1}{2}(S_{A_1}^{(m)}+S_{A_2}^{(m)})-S_A^{(m)}. 
\end{align}
\end{proof}

To compute the PT negativity, we need more structural knowledge about the stabilizer group. Let $k_i$ be the rank of $G_{A_i}$. We divide $G_A$ into a product $G_A=G_{A_1}G_{A_2}G_{12}$, where $G_{12}$ has $k-k_1-k_2$ number of generators independent from those of $G_{A_1}$ and $G_{A_2}$ \footnote{The choice of $G_{12}$ is in general not unique. }. Denote by $p_{A_2}(\cdot)$ the restriction of a stabilizer to subsystem $A_2$, which is well defined up to a sign if we require the corresponding restriction to be Hermitian. The following result about $G_{12}$ has been proved in Ref.\,\onlinecite{Chuang2004StabilizerEE} (in particular their Lemma 2), and we also spell out the proof for completeness. 
\begin{lemma}
We have the following canonical choice of the generators of $G_{12}$: 
\begin{align}
G_{12}=\ex{z_1,\cdots,z_r,w_1,\cdots,w_s,\bar w_1,\cdots,\bar w_s}, 
\label{eq:G12Structure}
\end{align}
such that $p_{A_2}(z_i)$ commutes with all generators of $G_{12}$ (therefore commutes with $G_A$), and $p_{A_2}(w_i)$ commutes with all generators of $G_{12}$ except for $\bar w_i$ (therefore the same holds if $w_i$ and $\bar w_i$ are exchanged). 
\end{lemma}
\begin{proof}
Let $C$ be a maximal subgroup of $G_{12}$ such that elements of $p_{A_2}(C)$ all commute. $C$ will eventually be identified with the subgroup generated by $z_i$'s and $w_j$'s.  
Denote by $\{c_j\}$ a set of generators of $C$. We can write $G_{12}=C\bar C$ for another subgroup $\bar C$ generated by $\{\bar c_k\}$ which are independent from $\{c_j\}$. Each $p_{A_2}(\bar c_k)$ must anticommute with some $c_j$, otherwise $\bar c_k$ can be added to $C$, violating the assumption that $C$ is maximal. Moreover, $p_{A_2}(\bar c_{k_1})$ and $p_{A_2}(\bar c_{k_2})$ with $k_1\neq k_2$ cannot anticommute with the same $c_j$, otherwise $\bar c_{k_1}\bar c_{k_2}$ can be added to $C$. As a consequence, the rank of $C$ is no less than that of $\bar C$. We will denote ${\rm rank}(\bar C)=s$ and ${\rm rank}(C)=s+r$. 
Now, we can reorganize the generators of $C$, by multiplying several $c_j$ together or reordering them, such that $p_{A_2}(\bar c_k)$ anticommutes only with $c_k$ and commutes with all other $c_{j\neq k}$. To satisfy the claim in Eq.\,\ref{eq:G12Structure}, we also want elements of $p_{A_2}(\bar C)$ to all commute. This can be achieved by recursively redefine the generators $\bar c_k$ as follows: Suppose $p_{A_2}(\bar c_k)$ for $k=1,2,\cdots,\kappa$ all commute with each other, but $p_{A_2}(\bar c_{\kappa+1})$ anticommutes with $\bar c_{k_1},\cdots,\bar c_{k_l}$ with $k_1,\cdots,k_l\leq \kappa$. We can redefine $\bar c_{\kappa+1}$ by multiplying it with $c_{k_1}c_{k_2}\cdots c_{k_l}$. The new $\bar c_{\kappa+1}$ then satisfies the property that $p_{A_2}(\bar c_{\kappa+1})$ commutes with all $\bar c_k$ with $k\leq \kappa$. The commuting/anticommuting properties between $p_{A_2}(\bar c_{\kappa+1})$ and $c_j$ are unaltered by this redefinition because $p_{A_2}(c_j)$ all commute. Therefore, we can repeat this process until elements of $p_{A_2}(\bar C)$ all commute. 
We have proved the claim, with the identifications $z_i=c_{s+i}$, $w_j=c_j$, and $\bar w_k=\bar c_k$. 
\end{proof}
With Eq.\,\ref{eq:G12Structure} in mind, the PT negativity of a stabilizer state can be computed. 
\begin{theorem}\label{thm:PTNegStabilizerState}
	For a stabilizer state, $\mc E_{\rm PT}(A_1,A_2)=s=[I^{(m)}(A_1,A_2)-r]/2$ for any $m$. 
\end{theorem} 
\begin{proof}
We can write the reduced density matrix $\rho_A$ as 
\begin{align}
\rho_A
=\frac{1}{2^{N_A}}\sum_{g\in G_{A_1}}\sum_{g'\in G_{A_2}}\sum_{h\in G_{12}}gg'h. 
\end{align}
Let $U$ and $V$ be two stabilizers, we generally have $(UV)^{T_2}=\pm U^{T_2}V^{T_2}$ when $p_{A_2}(U)$ and $p_{A_2}(V)$ commute (anticommute). We therefore have 
\begin{align}
&\Tr\left[ (\rho_A^{T_2})^{n} \right]=\frac{1}{2^{n N_A}}\sum_{\{g_i,g_i',h_i\}}\nonumber\\
&\quad\quad\Tr\left[ (g_1\cdots g_n)(g_1'\cdots g_n')^{T_2}(h_1^{T_2}\cdots h_n^{T_2}) \right], 
\end{align} 
where $g_i\in G_{A_1}$, $g'_i\in G_{A_2}$, and $h_i\in G_{12}$. 
The sums over $g$ and $g'$ can now be performed with the constraints $g_1\cdots g_n=g_1'\cdots g_n'=1$, and we obtain
\begin{align}
\Tr\left[ (\rho_A^{T_2})^{n} \right]=\frac{2^{(n-1)(k_1+k_2)}}{2^{n N_A}}\sum_{ \{  h_i\} }\Tr(h_1^{T_2}\cdots h_n^{T_2}). 
\end{align}
Since the effect of partial transpose on a stabilizer is just a $\pm$ sign, for the trace on the right-hand side to be nonzero, we must have $h_n=h_1\cdots h_{n-1}$. The difficult part is to figure out the sign of this trace. Each $h_i\in G_{12}$ can be expanded as a product of the canonical generators, and is thus represented by a collection of $\mathbb{Z}_2$ indices $(\lambda^i_1,\cdots,\lambda^i_r,\mu^i_1,\cdots,\mu^i_s,\nu^i_1,\cdots,\nu^i_s)$, corresponding to the powers of $z_i$, $w_i$, and $\bar w_i$, respectively. We define an operator $\tilde h_i$, by replacing each canonical generator in the expansion of $h_i$ by its partial transpose. For example, if $h_i=z_a w_b \bar w_c$, then $\tilde h_i=z_a^{T_2}w_b^{T_2}\bar w_c^{T_2}$. Notice that $h_1h_2\cdots h_n=1$ implies $\tilde h_1\tilde h_2\cdots \tilde h_n=1$ as well. 
One can see that $h_i^{T_2}=\tilde h_i\exp(i\pi\sum_l \mu^i_l\nu^i_l)$. The sign of $\Tr(h_1^{T_2}\cdots h_n^{T_2})$ is then given by
\begin{align}
\prod_{l=1}^s\exp\left\{ i\pi\left[ -\sum_{i=1}^{n-1}\mu^i_l\nu^i_l + \left(\sum_{i=1}^{n-1} \mu^i_l\right)\left(\sum_{j=1}^{n-1}\nu^j_l\right) \right] \right\}. 
\end{align}
The summation over $\nu^i_l$ gives nonzero result only if any $(n-2)$ elements of $\{ \mu^1_l,\cdots,\mu^{n-1}_l \}$ sum up to an even integer. For the case of even $n$, this implies that $\mu^1_l,\cdots,\mu^{n-1}_l$ are either all even or all odd. The independent variables to sum over are therefore $\mu^1_l$ for all $l$, and $\nu^i_l$ for $i=1,\cdots,n-1$ and all $l$. Also remembering that $\Tr_A(1)=2^{N_A}$, we finally arrive at the result
\begin{align}
\sum_{ \{ h_i \} }\Tr(h_1^{T_2}\cdots h_n^{T_2})=2^{ns+N_A}\quad (\text{$n$ even}). 
\end{align}
Taking the limit $n\rightarrow 1$, we obtain 
\begin{align}
\mc E_{\rm PT}&=s=\frac{1}{2}(k-k_1-k_2-r)\\
&=\frac{1}{2}I^{(m)}(A_1,A_2)-\frac{1}{2}r, 
\end{align}
where $I^{(m)}(A_1,A_2)$ is the Renyi mutual information, and $r$ is defined in Eq.\,\ref{eq:G12Structure}. 
\end{proof}
\noindent As we will see, with trisection points, the dependence of $r$ on local geometry is exactly the reason why the subleading term in $\mc E_{\rm PT}$ is nontopological. 

\section{Lattice Examples}\label{sec:LatticeExamples}
\subsection{Lattice-continuum comparison protocol}
In this section, we would like to demonstrate our results in lattice examples. In order to compare with the continuum results, we need to first precisely define the subleading term $\mc E^{\rm top}$ of each negativity $\mc E$ on lattices. We wish the term $\mc E^{\rm top}$ to be topological (insensitive to smooth deformations of the subregions) and universal (insensitive to perturbations to the Hamiltonian). 

Recall that the negativity $\mc E(A_1,A_2)$ contains area-law terms and other nontopological terms due to sharp features like corners. We need to find a way to cancel them all. To this end, consider the following linear combination. 
\begin{align}
\Delta &:=\mc E(A_1,A_2)-\mc E(A,\varnothing)\nonumber\\
&-\left[\mc E(A_1,\bar A_1)+\mc E(A_2,\bar A_2)-\mc E(A,B)\right]/2,
\label{eq:DeltaDef}
\end{align}
where $\bar{\cdot}$ means complement. 
Suppose all nontopological/nonuniversal terms in $\mc E$ are made of \emph{local} contributions (insensitive to changes far away) near the interfaces. We claim that in the \emph{absence} of trisection points, all these terms have been cancelled out in $\Delta$. In other words, the quantity $\Delta$ is topological and universal. 
To see why this is true, let us examine a concrete example shown in the top left panel of Fig.\,\ref{fig:CancellationExample}. In this particular partition of the space, $A_1$ has the topology of an annulus, and $A_2$ has the topology of two disjoint disks. We have illustrated the nontopological/nonuniversal terms by colors. Two colors are used because $A_1A_2$ interfaces and $A_iB$ interfaces can contribute in different ways to $\mc E(A_1,A_2)$. Thick lines represent area-law terms, while big dots represent corner contributions. The nonuniversal terms of the five terms in Eq.\,\ref{eq:DeltaDef} have all been illustrated in this figure, and one can see that they do cancel with each other. Such cancellation happens in general for any tripartition without trisections. However, in the presence of trisection points, the interfaces between different pairs of subregions meet together, and there may be uncanceled local contributions near the trisections. We need to study concrete examples to see whether $\Delta$ is robust in this situation, and the result turns out to be negative. 

\begin{figure}
    \centering
    \includegraphics[width=\linewidth]{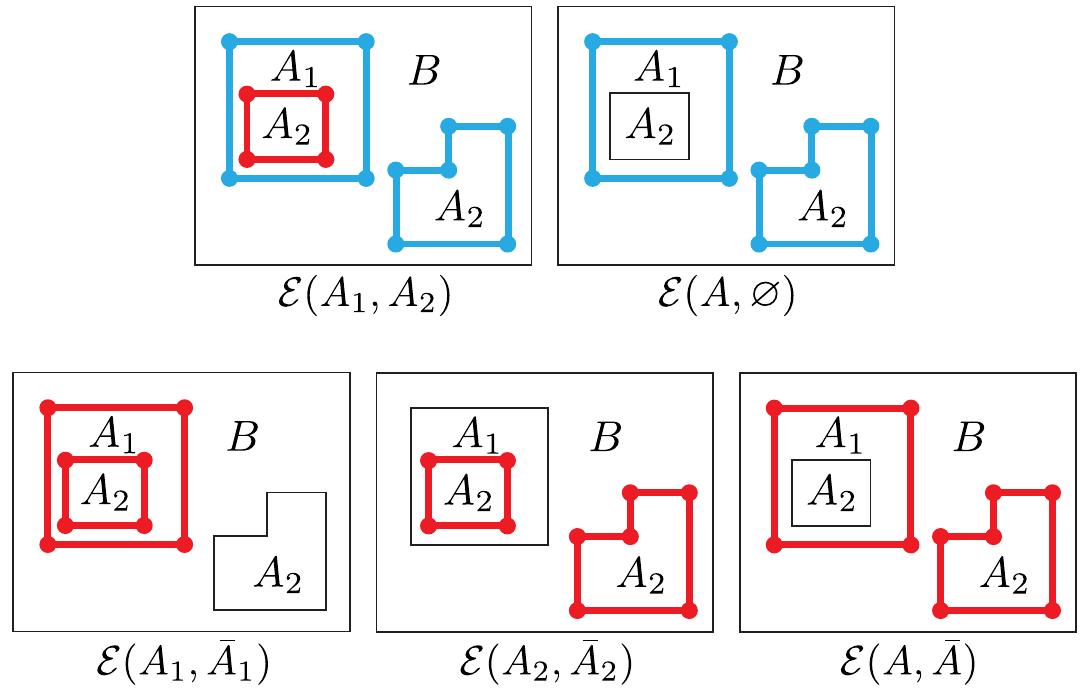}
    \caption{Nontopological/nonuniversal terms within the five terms in the expression for $\Delta$ --- a concrete example. The two colors represent nontopological/nonuniversal terms in $\mc E(A_1,A_2)$ localized at $A_1A_2$ interfaces (red) and $A_iB$ interfaces (blue). For example, the second row for $\mc E(P,\bar{P})$ do not have blue interfaces because there is no third party (or the corresponding $B=\varnothing$). }
    \label{fig:CancellationExample}
\end{figure}

We note that the above argument relies on the locality assumption about nontopological/nonuniversal terms. This assumption does not hold in certain systems with the so-called \emph{spurious} long-range entanglement \cite{spurious_cylinder16,spurious_subsystem19,spurious_nonSPT20,spurious_bound23}. We assume there are meaningful classes of systems without spurious entanglement, and in this work we only focus on such nice systems, and talk about, e.g., universality in the corresponding possibly constrained phase space. 

When the quantity $\Delta$ is robust, it can be used to give a definition of $\mc E^{\rm top}$ on lattices. We note that on the right-hand side of Eq.\,\ref{eq:DeltaDef}, all terms except for the first one are actually EEs: $\mc E_{\rm CCNR}(A,\varnothing)=-S^{(2)}_A/2$, $\mc E_{\rm PT}(A,\varnothing)=0$, and $\mc E(P,\bar P)=S^{(1/2)}_P$. Hence, given the lattice result of $\Delta$ and the (continuum or lattice) results of topological EEs, we can add a superscript ``top'' to each term on the right-hand side of Eq.\,\ref{eq:DeltaDef}, and then solve $\mc E^{\rm top}(A_1,A_2)$. More explicitly, we define $\mc E^{\rm top}$ on lattices by
\begin{align}
	&\mc E_{\rm CCNR}^{\rm top}(A_1,A_2):=\Delta_{\rm CCNR}\nonumber\\
    &\quad+\frac{1}{2}\left\{\left[S^{(1/2)}_{A_1}+S^{(1/2)}_{A_2}-S^{(1/2)}_{A}\right]-S^{(2)}_A\right\}^{\rm top}, \label{eq:CCNRLatticeDef}\\
	&\mc E_{\rm PT}^{\rm top}(A_1,A_2):=\Delta_{\rm PT}\nonumber\\
    &\quad+\frac{1}{2}\left[S^{(1/2)}_{A_1}+S^{(1/2)}_{A_2}-S^{(1/2)}_{A}\right]^{\rm top}, \label{eq:PTLatticeDef}
\end{align}
where $\Delta$ should be computed on lattices according to Eq.\,\ref{eq:DeltaDef}. In this work, for simplicity, we will just use the continuum results of $S^{\rm top}$ in the above equations. This means we are effectively comparing the continuum and lattice results for $\Delta$. The readers may instead wish to use lattice definitions of topological EEs. For this purpose, we note that there already exist lattice definitions of topological EEs for disk and annulus regions \cite{Kitaev2006TEE,Levin2006TEE,Brown2013TEETwist}, and then it should not be hard to find working definitions for more complicated regions, because each region with a complicated topology can always be decomposed into simpler ones. We will not go into more details in this direction.

\subsection{A sphere example without trisection}
Let us consider the $\mathbb{Z}_2$ TC model \cite{Kitaev2003ToricCode}, the simplest lattice model with a topological order. The model can be defined on an arbitrary 2D closed space manifold discretized into a lattice with qubits living on links. The Hamiltonian takes the form 
\begin{align}
H_{\rm TC}=-\sum_s \mathbf{Z}_s-\sum_p \mathbf{X}_p, 
\end{align}
where $\mathbf{Z}_s$ is the product of Pauli $Z$ operators on all links containing the site $s$, and $\mathbf{X}_p$ is the product of Pauli $X$ operators on all links surrounding the plaquette $p$; see Fig.\,\ref{fig:ToricCodeExample}(a) for the example of a square lattice. For concreteness, we put the model on a cube surface (topologically a sphere) as shown in Fig.\,\ref{fig:ToricCodeExample}(b), and focus on its unique ground state. 

\begin{figure}
	\centering
	\includegraphics{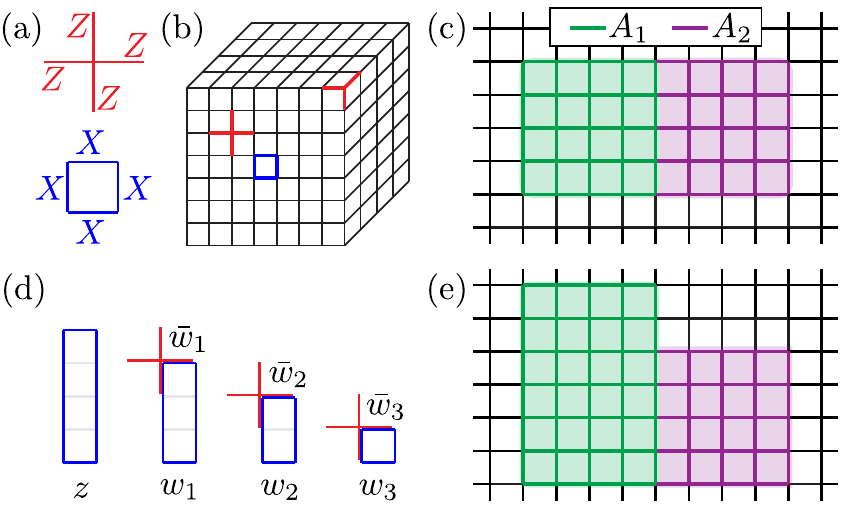}
	\caption{(a) Star and plaquette terms of the toric code model in the case of square lattice. (b) Toric code on a cube surface which is topologically a sphere. (c) Subsystems $A_1$ (green, left) and $A_2$ (purple, right). (d) Canonical generators of $G_{12}$. (e) A deformation of the tripartition in panel c. }
	\label{fig:ToricCodeExample}
\end{figure}

As the first example, consider the tripartition in Fig.\,\ref{fig:SpherePartitions}(b). We can prove $r=0$ in this case: If $r>0$, $p_{A_2}(z_1)$ commutes with $G_A$. Since $p_{A_2}(z_1)$ lies in $A_2$ and is far away from the $AB$ interface, it should then commute with all star and plaquette terms. Now suppose we combine $B$ and $A_1$ into a single region and call it the new $A_1$, the \emph{old} $z_1$ still satisfies $z_1\notin G_{A_1}G_{A_2}$ and $p_{A_2}(z_1)$ commuting with $G_A$, thus $r>0$ still holds as one can prove by contradiction using Eq.\,\ref{eq:G12Structure}, but this is not possible because $\mc E_{\rm PT}(A_1,A_2)=I^{(1/2)}(A_1,A_2)/2$ for a bipartite pure state. Hence, by comparing Theorem~\ref{thm:EntMeasureRelationsCS}, \ref{thm:CCNRNegStabilizerState}, and \ref{thm:PTNegStabilizerState}, we see that $\mc E_{\rm PT}^{\rm top}$ and $\mc E_{\rm CCNR}^{\rm top}$ computed on the lattice indeed match with the continuum results. 

\subsection{Trisection points}
Next, we still consider the ground state on the cube surface, but take $A_1$ and $A_2$ to be adjacent rectangular regions as shown in Fig.\,\ref{fig:ToricCodeExample}(c); this corresponds to the tricky scenario in Fig.\,\ref{fig:SpherePartitions}(d). Any stabilizer of the system, to commute with the star and plaquette terms, must consists of $X$ loops along links, and $Z$ loops along dual lattice links. One can then check that a set of independent generators for $G_A$ is simply all star and plaquette terms inside, similarly for $G_{A_1}$ and $G_{A_2}$. Let $L_{A_1}$, $L_{A_2}$, and $L_A$ be the lattice perimeters of $A_1$, $A_2$, and $A$, respectively. A simple counting gives $S^{(n)}_A=N_A-k=L_A-1$, similarly for $A_1$ and $A_2$. Note that the $-1$ is nothing but the topological EE of a contractible region with a smooth boundary for the TC model \cite{Kitaev2006TEE,Levin2006TEE}. 
The rank of $G_{12}$ is given by $k-k_1-k_2=2L_{12}-1=I^{(n)}(A_1,A_2)$ where $L_{12}$ is the lattice length of the interface between $A_1$ and $A_2$. 
We can choose the generators of $G_{12}$ to be the star and plaquette terms crossing the interface, but they do not satisfy the property of canonical generators introduced previously. We show in Fig.\,\ref{fig:ToricCodeExample}(d) how these generators can be reorganized into a canonical set, from which we see $r=1$. According to Theorem \ref{thm:PTNegStabilizerState}, $\mc E_{\rm PT}(A_1,A_2)=[I^{(n)}(A_1,A_2)-1]/2=L_{12}-1$. The CCNR negativity is found to be $\mc E_{\rm CCNR}(A_1,A_2)=L_{12}-L_A/2$. From our previous lattice definition of $\mc E^{\rm top}$, we obtain $\mc E_{\rm PT}^{\rm top}=-1$ and $\mc E_{\rm CCNR}^{\rm top}=0$, where the latter is consistent with the field-theoretic calculation. 

In view of Theorem~\ref{thm:CCNRNegStabilizerState}, the above result of $\mc E_{\rm CCNR}^{\rm top}$ should at least be topological \emph{for this particular model}, but $\mc E_{\rm PT}^{\rm top}$ is more subtle. Consider the deformed tripartition in Fig.\,\ref{fig:ToricCodeExample}(e). 
With one more star term included in $G_{12}$, we see $r=0$, and therefore $\mc E_{\rm PT}=I^{(n)}(A_1,A_2)/2$. It follows that $\mc E_{\rm PT}^{\rm top}$ now becomes $-1/2$, different from the previous $-1$ result. Hence, $\mc E_{\rm PT}^{\rm top}$ defined according to Eq.\,\ref{eq:DeltaDef} is actually not topological. 

\subsection{A torus example}
Finally, let us consider the $\mathbb{Z}_2$ TC model on a torus -- square lattice with periodic boundary condition, and the tripartition in Fig.\,\ref{fig:TorusPartitions}b: $A_1$ and $A_2$ are adjacent, and all interface circles are vertical. The TC Hamiltonian has four degenerate ground states, so we need to specify two more stabilizers besides the star and plaquette terms. The TC model has two independent types of string operators: the $e$ string consisting of Pauli $X$ operators along links, and the $m$ string consisting of Pauli $Z$ operators along dual lattice links. We consider the interesting situation where the additional stabilizers are \emph{horizontal} $e$ and $m$ loops that wind around the torus. Each circular interface is crossed by the $e$ or $m$ loop once. 

In a general CS theory, the state $\ket{R_0}$, defined by path integral inside the torus with no WL insertion, is an eigenstate with eigenvalue $d_a$ of the vertical string/WL operator in the representation $R_a$, as one can verify using surgery. Hence, in the TC model, $\ket{R_0}$ is the eigenstate with eigenvalue $+1$ for both the \emph{vertical} $e$ and $m$ loops. The state $\ket{\psi}$ we consider is then related to $\ket{R_0}$ by a modular $\mc S$ transformation, or more precisely $\ket{\psi}=\sum_j\psi_j\ket{R_j}$, with $\psi_j=\mc S_{0j}=1/2$ $(j=0,1,2,3)$, where we used $\mc S_{0j}=d_j/\mc D$ and $d_j=1$, $\mc D=2$. 
From our continuum results, 
\begin{align}
	\text{(continuum)}\quad\mc E_{\rm CCNR}^{\rm top}=0,\quad\mc E_{\rm PT}^{\rm top}=-1, 
\end{align}
where the base of logarithm has been set to $2$. The topological EE of $A_1$, $A_2$, or $A$ is simply zero; this is the maximally entangled state. 

Now let us get onto the lattice. From Theorems~\ref{thm:CCNRNegStabilizerState} and \ref{thm:PTNegStabilizerState}, and Eq.\,\ref{eq:DeltaDef}, we have $\Delta_{\rm CCNR}=0$, and $\Delta_{\rm PT}=-r/2$. Since all topological EEs are zero for the state being considered, we see that the lattice result of $\mc E_{\rm CCNR}^{\rm top}$ matches with the continuum one. The lattice result of $\mc E_{\rm PT}^{\rm top}$ is $-r/2$. It is not hard to see $r=2$ in this case: We can take $z_1$ ($z_2$) to be the product of two vertical $e$ ($m$) loops inside $A_1$ and $A_2$, respectively. A pair of vertical $e$ ($m$) loops can be generated by the plaquette (star) terms in between and is indeed a valid stabilizer, but a single vertical $e$ ($m$) loop anticommutes with the horizontal $m$ ($e$) loop and does not belong to the stabilizer group. We have thus also verified $\mc E_{\rm PT}^{\rm top}=-1$ on the lattice. 

\section{Discussions}\label{sec:Discussions}
In this work, we have studied bipartite mixed-state quantum entanglement in topologically ordered states, quantified by PT and CCNR negativities. General results are obtained both in the continuum (CS theory) and in lattice models.
On lattices, We have examined whether the negativities are topological and universal. In particular, for the tricky scenario with trisection points, the subleading piece of the PT negativity is not topological, in contrast to CCNR. 

There are a number of possible future directions: (1) It remains to be checked whether our result of $\mc E^{\rm top}_{\rm CCNR}$ with trisection points has any universal meaning on lattices. 
(2) Our general approaches both in continuum and on lattices may be extended to study other interesting entanglement quantities, such as the reflected entropy \cite{Dutta2021ReflectedEntropy,Chen2021ReflectedEntropy,Zaletel2022MarkovGap,Ryu2022VertexStates}. (3) The negativity formulas in the stabilizer formalism can be readily applied to higher-dimensional topological orders and to other quantum systems, such as random quantum circuits. (4) It is also interesting to consider more general string-net models. 

\textit{Acknowledgments. }---
We are grateful to Yuhan Liu for a stimulating discussion, and to Jonah Kudler-Flam, Shinsei Ryu, and Ashvin Vishwanath for valuable feedbacks on our manuscript. S. L. is supported by the Gordon and Betty Moore Foundation under Grant No. GBMF8690 and the National Science Foundation under Grant No. NSF PHY-1748958. 

\textit{Note added. }---
Upon completion of this manuscript, we became aware of a few related works: The use of PT negativity for diagnosing error-corrupted topological orders is discussed in Ref.\,\onlinecite{Fan2023TODecoherence}. Refs.\,\onlinecite{Ryu2023Diagrammatic} and \onlinecite{Ryu2023Trisection} both have continuum computations of PT negativity with trisection points. The former takes an anyon diagram approach with a particular regularization of the trisection points (wormholes), and the latter takes a boundary/vertex state approach. 

\appendix*

\section{Entanglement Measures 
in Chern-Simons Theory}

\subsection{Overview}
Let $\cM$ be an orientable closed 2D manifold with genus $g\ge 0$.
Consider a Chern-Simons (CS) theory on $\cM$ and its specific state $\ket{\psi}$, defined by the path integral in the 3D interior $\cM_3\supset \cM$ of $\cM$ with Wilson lines (WLs) inserted. The goal of this Appendix is to calculate entanglement measures of $\ket{\psi}$ with respect to a partition of $\cM$ under fairly general assumptions. Comparing to the main text, this Appendix contains a few different notations, and can be read independently. In this overview subsection we set up the problem, present our main results, and finally review the surgery method that we will use.

\subsubsection{Tripartition: Ball-tube system}\label{sec:ball-tube}

Consider a tripartition $\cM = A_1\cup A_2\cup B$ that separates $\cM$ into regions $\{r_j:j = 1,\cdots, N\}$, where each region $r_j$ belongs to either party $A_1,A_2$ or $B$. This reduces to a bipartition $\cM = A\cup B$ when combining $A = A_1\cup A_2$. We make an assumption on the tripartition:
\begin{asmp}\label{asmp}
    The edges of the regions $\{r_j\}$ are disjoint circles that are contractable in $\cM_3$.
\end{asmp}
In particular, this implies that no three different regions $i,j,k$ share common points: \begin{equation}\label{eq:3no_touch}
    r_i\cap r_j \cap r_k = \varnothing.
\end{equation}
Any pair of regions $(r_i,r_j)$ either is not connected, or shares $\ge 1$ disconnected edges. All edges do not touch each other according to \eqref{eq:3no_touch}. 
Beyond \eqref{eq:3no_touch}, Assumption \ref{asmp} also requires contractability of edges in the bulk. For example, Fig.~\ref{fig:tripart}(a) and any tripartition on the sphere that satisfies \eqref{eq:3no_touch} fulfil Assumption \ref{asmp}, while Fig.~\ref{fig:tripart}(b) does not: the edge circles between $B$ and $A_1$ are non-contractable. Note that tripartitions like Fig.~\ref{fig:tripart}(b) can be transformed using modular transformation to a ``good'' tripartition satisfying Assumption \ref{asmp}, with suitable Wilson loops inserted \cite{Witten1989CSTheory,Wen2016NegBdry}. However, we leave as an open question whether any tripartition satisfying \eqref{eq:3no_touch} can be transformed in this way, and just impose Assumption \ref{asmp} throughout.

Based on Assumption \ref{asmp}, $\cM_3$ can be deformed to a ball-tube system, with Fig.~\ref{fig:tripart}(a) as an example. In the ball-tube system, each region $r_j$ on the surface corresponds to an interior 3D region $\cR_j$, composed of several balls and half-balls connected by $D^2$-tubes. Here a (half-)ball is topologically $D^3$ while a $D^2$-tube is $D^2\times [0, 1]$. To connect regions, for each edge between $r_i$ and $r_j$, the corresponding half-ball of $\cR_i$ combines with its counterpart half-ball of $\cR_j$ to make a ball. We refer to this as an \emph{edge ball} to differentiate with the \emph{party balls} that belong to only one party, and we call the joint surface between the two half-balls an \emph{interface}. The boundary of the interface is then an edge that connects $r_i$ and $r_j$. In this way $\cM_3$ is deformed to a set $\cB$ of balls connected by tubes.

We provide further justification on why the above deformation is possible. Namely, one can start from the unpartitioned $\cM_3$ and add the edges one by one. The initial $\cM_3$ is naturally a ball-tube system, where a central ball is connected to $g$ surrounding balls by two $D^2$-tubes each. When adding the first edge that is contractable in $\cM_3$, one of the balls will be cut by an interface to become an edge ball. We do not need to make the cut on tubes because it can always deform to one ball on its end. Then adding the second edge, there are two possibilities: if the cut is on a party ball then it becomes an edge ball just as step one; otherwise if the cut is on an edge ball, we deform the ball to two edge balls connected by one new tube. The manifold thus remains to be a ball-tube system, so that edges can be added one by one.

There is some freedom in the above procedure, namely for each $\cR_j$, which pairs of (half-)balls are connected by tubes? The entanglement quantities we study should not depend on the choice. Therefore, to express our results, we need to use quantities that do not depend on the connection. Define $E$ to be the total number of edges, which is just the number of edge balls independent of the connection. Moreover, define $E_P$ ($P = A_1,A_2,B$) to be the number of edges with party $P$ on one side. We will define similar quantities for WL in the next subsubsection.

\begin{figure}
    \centering
    \includegraphics[width=0.48\textwidth]{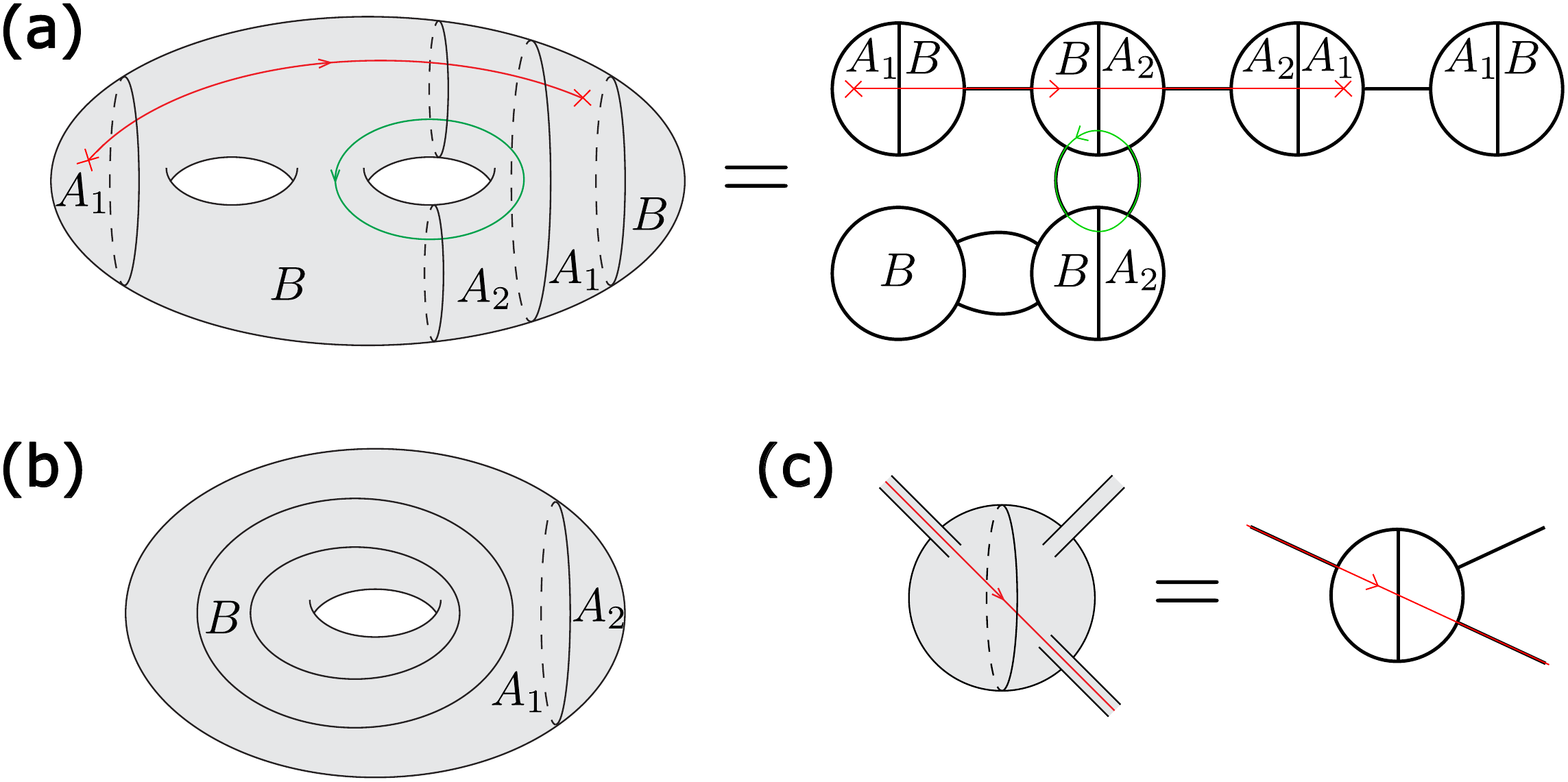}
    \caption{\label{fig:tripart} (a) A tripartition of a genus-$2$ surface $\cM$ that satisfies Assumption \ref{asmp}, which is deformed to the right ball-tube system. The relevant quantities are $E=5, E_B=4, K(1)=3, K_B(1)=2, K_{12}(1)=K_{1B}(1)=K_{2B}(1)=1, K(2)=K_{2B}(2)=2, K_{12}(2)=0$. (b) A tripartition that violates Assumption \ref{asmp}, which cannot be deformed to a ball-tube system. (c) The 3D version of an example element in the sketched ball-tube system (a). In particular, a black line represents a tube with cross section $D^2$, and a red line is a WL. }
\end{figure}

\subsubsection{Assumption on WL}
Let $W$ WLs be embedded in $\cM_3$.
A WL can either be a loop in the interior, or puncture the surface $\cM$ with two anyons. To avoid braiding and knot, we consider the following WL configuration in the ball-tube system $\cM_3$: \begin{asmp}\label{asmpW}
    The WLs can be simultaneously deformed, without touching among themselves or between WL endpoints and edges on the surface $\cM$, to a configuration that all WLs are contained in the surface $\cM$ and do not touch each other. Moreover, each $D^2$-tube contains at most one WL which goes through the tube once.
\end{asmp} 
The deformation does not change topology, and thus does not change the quantities that we will compute. Moreover, contractable (in $\cM_3$) Wilson loops with no flux of other WLs can be annihilated freely, as one can check explicitly using surgery introduced in Section \ref{sec:surgery}. Note that even if a tube contains more than one WLs in one ball-tube system representation of $\cM_3$, one may adjust the connection of balls to get another ball-tube system that satisfies Assumption \ref{asmpW}. From now on we assume such a ball-tube system is chosen. Similar to $E$, the number of interfaces traversed by the WL $w$, $K(w)$, is well-defined. Note that an interface can contribute more than one to $K(w)$, if it is traversed by $w$ back and forth. Similarly, $K_P(w)$ ($K_{PP'}(w)$) counts the number of interfaces shared by party $P$ (and party $P'$) that $w$ traverses.

Each WL is assigned a representation $a$, and we consider the superposition
\begin{widetext}
\begin{equation}\label{eq:psi}
    \ket{\psi} = \sum_{a(1)\cdots a(W)}\psi_{a(1)}(1)\times \psi_{a(2)}(2)\times\cdots \times \psi_{a(W)}(W) \ket{a(1)\cdots a(W)}.
\end{equation}
\end{widetext}
Here $\ket{a(1)\cdots a(W)}$ is the \emph{normalized} state with each WL $w\in \{1,\cdots,W\}$ in the fixed representation $a(w)$. As we will check in Section \ref{sec:norm}, these states are orthogonal to each other. Thus \eqref{eq:psi} can be viewed as a direct product among WLs, and we assume $\sum_a \abs{\psi_a(w)}^2=1$ for any $w$ such that $\braket{\psi|\psi}=1$.

\subsubsection{Entanglement measures and the replica method}
We study three entanglement measures. First, $\ket{\psi}$ is a bipartite pure state shared by $A$ and $B$, so there is entanglement entropy (EE) from the reduced density matrix $\rho_A$ on $A$ \begin{align}
    S_A &= S_B = -\tr{ \rho_A\log \rho_A} = \lim_{n\rightarrow 1} S^{(n)}_A,
    \quad \mathrm{where}
    \nonumber\\ \quad S^{(n)}_A &= (1-n)^{-1}\log \tr{ \rho_A^n}.
\end{align}
Here $S^{(n)}_A$ is the $n$-th R\' enyi EE. In the replica method, we will compute $S^{(n)}_A$ first, and then take the replica limit $n\rightarrow 1$ to get the von Neumann EE. The same method is also applied for the following two entanglement measures.

$\rho_A$ is a bipartite mixed state shard by $A_1$ and $A_2$, and we study two quantities that measure the quantum correlation between the two parties: the
partial transpose negativity (PT) ${\mc E}_{\mathrm{PT}}$, and the computable cross norm negativity (CCNR) ${\mc E}_{\mathrm{CCNR}}$ \cite{Rudolph2005CCNR,Chen2002CCNR}. They are defined by \begin{subequations} \begin{align}
{\mc E}_{\mathrm{PT}} &= \log \norm{\rho_A^{T_2}}_1 = \lim_{\text{even } n\rightarrow 1} {\mc E}_{\mathrm{PT}}^{(n)}, \\
{\mc E}_{\mathrm{CCNR}} &= \log \norm{R_\rho}_1 = \lim_{\text{even } n\rightarrow 1} {\mc E}_{\mathrm{CCNR}}^{(n)},
\end{align}
\end{subequations}
where 
\begin{subequations} \begin{align}
{\mc E}_{\mathrm{PT}}^{(n)} &= \log\tr{\rho_A^{T_2}}^n, \\
{\mc E}_{\mathrm{CCNR}}^{(n)} &= \log \tr{R_\rho^\dagger R_\rho}^{n/2}.
\end{align}
\end{subequations}

Here $T_2$ means partial transpose on $A_2$, and $R_\rho$ is related to $\rho$ by \begin{equation}
    \bra{e_1}\bra{e'_1} R_\rho \ket{e_2}\ket{e'_2} = \bra{e_1}\bra{e_2} \rho \ket{e'_1}\ket{e'_2},
\end{equation}
where $\{\ket{e_1}\}, \{\ket{e_2}\}$ are basis for subsystem $A_1$ and $A_2$ respectively.

\subsubsection{Main result}\label{sec:app_main}
For any state \eqref{eq:psi} with Assumptions \ref{asmp} and \ref{asmpW}, we have
\begin{widetext}
\begin{subequations}\label{eq:main}
    \begin{align}
    S_A &= -E_B\log \mathcal{D} + \sum_{w} \sum_a \abs{\psi_a(w)}^2 \left(K_B(w)\log d_a- \mathbb{I}[K_B(w)>0]\log\abs{\psi_a(w)}^2\right), \label{eq:SA}\\
    {\mc E}_{\mathrm{PT}} &= -(E-E_B)\log \mathcal{D} + \sum_{w} \left\{\begin{array}{cr}
        \log\left(\sum_a \abs{\psi_a(w)}^2d_a^{K(w)-K_B(w)} \right), & \text{if } K_B(w)>0 \\
        2\log\left(\sum_a \abs{\psi_a(w)}d_a^{K(w)/2}\right), & \text{if } K(w)>K_B(w)=0
    \end{array}\right. \label{eq:ppt}\\
    {\mc E}_{\mathrm{CCNR}} &= -(E-\frac{3}{2}E_B)\log \mathcal{D} \nonumber\\
    &\quad + \sum_{w} \left\{\begin{array}{cr}
        \log\left(\sum_a \abs{\psi_a(w)}^2d_a^{K(w)-\frac{3}{2}K_B(w)} \right), & \text{if at least two in } \{K_{12}(w),K_{1B}(w),K_{2B}(w)\} \text{ are nonzero} \\
        \frac{1}{2}\log\left(\sum_a \abs{\psi_a(w)}^4d_a^{-K(w)}\right), & \text{if only one in } \{K_{1B}(w),K_{2B}(w)\} \text{ is nonzero and } K_{12}(w)=0 \\
        2\log\left(\sum_a \abs{\psi_a(w)}d_a^{K(w)/2}\right), & \text{if } K(w)>K_B(w)=0
    \end{array}\right. \label{eq:ccnr}
\end{align}
\end{subequations} 
\end{widetext}
where $\mathcal{D}=1/\cS_{00}$ with $\cS_{ab}$ being the modular $\cS$ matrix of the Chern-Simons theory, and \begin{equation}\label{eq:gamma}
    d_a = \cS_{0a}/\cS_{00},
\end{equation}
is the quantum dimension.

When all representations are fixed $\psi_a(w)=\delta_{a,a(w)}$ by some function $a(w)$, we have simplified expressions for all R\' enyi quantities
\begin{widetext}
\begin{subequations}\label{eq:renyi}
    \begin{align}
    S_P^{(n)} &= S_P = -E_P \log \mathcal{D} + \sum_w K_P(w) \log d_{a(w)}, \label{eq:Sn}\\
    \, {\mc E}_{\mathrm{PT}}^{(n)} &= -[E-E_B+E(1-n)]\log\mathcal{D} + \sum_w [K(w)-K_B(w)+K(w)(1-n)] \log d_{a(w)}, \label{eq:PTn} \\
    \, {\mc E}_{\mathrm{CCNR}}^{(n)} &=-\left[E-\frac{3}{2}E_B+\left(E-\frac{1}{2}E_B\right)(1-n)\right]\log\mathcal{D}\nonumber\\ &\quad + \sum_w\left[K(w)-\frac{3}{2}K_B(w)+\left(K(w)-\frac{1}{2}K_B(w)\right)(1-n)\right] \log d_{a(w)}, \label{eq:CCNRn}
\end{align} 
\end{subequations} 
which reduces to \begin{subequations}\label{eq:PC=S}
    \begin{align}
    2\, {\mc E}_{\mathrm{PT}}^{(n)} &= S_1 + S_2 - S_{12} + (1-n)(S_1 + S_2 + S_{12}) = I_{12}+(1-n)(S_1 + S_2 + S_{12}), \label{eq:PT=I} \\
    2\, {\mc E}_{\mathrm{CCNR}}^{(n)} &= S_1 + S_2 - 2S_{12} + (1-n)(S_1 + S_2) = 2\, {\mc E}_{\mathrm{PT}}^{(n)} - (2-n)S_{12}, \label{eq:CCNR=S}
\end{align} 
\end{subequations} 
\end{widetext}
where $S_{12}=S_A=S_B$, and $I_{12}$ is the mutual information shared between $A_1$ and $A_2$.

\subsubsection{Overview of surgery method}\label{sec:surgery}
Surgery is introduced in \cite{Witten1989CSTheory} to calculate partition functions of Chern-Simons theories in general settings (with braiding among WLs etc). For our purpose, we only need two gadgets: \begin{enumerate}
    \item The partition function of a $S^3$ with a single Wilson loop of representation $a$ in it is \begin{equation}\label{eq:Z=S}
        Z_a = \cS_{0a}.
    \end{equation}
    \item As shown in Fig.~\ref{fig:surgery_app}, an $S^2$-tube traversed by two WLs of representation $a$ can be cut by a $S^3$ with a Wilson loop of the same representation $a$ in it. Formally, let $M$ be the original manifold with the $S^2$-tube connected, and let $M_L,M_R$ be the two manifolds after cutting the tube (and reconnecting the cutted WL in each manifold), then \begin{equation}\label{eq:ZZ=ZZ_app}
        Z(M)\cdot Z_a = Z(M_L) \cdot Z(M_R).
    \end{equation}
\end{enumerate}

As an exercise of these two rules, we have the factorization of a $S^3$ with more than one Wilson loops that do not link with each other, see \eqref{eq:Zaa} as the two-loop case. The reason is that empty $S^2$-tubes can be cut to separate the loops into different $S^3$s. We encourage the reader to revisit the main text for the intuition behind \eqref{eq:ZZ=ZZ_app}, and an example calculation of $\mathcal{E}_{\rm CCNR}$ using the two rules.

\begin{figure}
    \centering
    \includegraphics[width=0.48\textwidth]{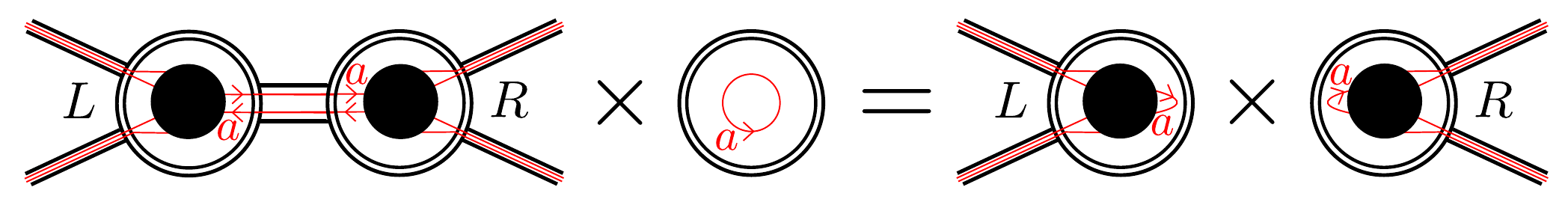}
    \caption{In the surgery method, an $S^2$-tube traversed by two WLs of representation $a$ is cut by a $S^3$ with a Wilson loop of the same representation $a$ in it. In comparison to Fig.~\ref{fig:tripart}, two black concentric circles represent a $S^3$, while two black parallel straight lines represent a $S^2$-tube. The black solid circle means the structure of WLs inside does not matter. $L,R$ means the number of tubes going out of the left/right $S^3$, while only two of the tubes are explicitly shown. This is the ball-tube version of Fig.~\ref{fig:surgery} in the main text. }
    \label{fig:surgery_app}
\end{figure}

\subsection{No Wilson line}\label{sec:noWilson}
In this subsection we consider the case without WLs, and derive the first terms in \eqref{eq:main} and \eqref{eq:renyi}. Treating WLs will be tedious, but shares the same basic idea in this subsection: each entanglement quantity in the replica method is (roughly) \emph{product of local contributions} from individual balls. 

Thanks to the symmetry $A_1\leftrightarrow A_2$ in the problem, we classify the balls in $\cB$ to 4 types: $\cB_A, \cB_B,\cB_{12},\cB_{AB}$ where the subscript indicates what parties the ball belongs to. For example, $\cB_A$ belongs to either $A_1$ or $A_2$, while $\cB_{12}$ is an edge ball shared by $A_1$ and $A_2$. Slightly abusing notation, we use (for example) $\cB_A$ to represent two precise items: 1. any ball of type $\cB_A$, and 2. the set of all such balls. These 4 types are sketched in the first column of Table~\ref{tab:noWilson}, where $L,R\ge 0$ means the number of tubes going out of the left/right half-ball, although for $\cB_A, \cB_B$ only their sum $L+R$ is meaningful. We label the balls by $\beta$, so $L,R$ are actually functions $L(\beta), R(\beta)$.

\begin{table*}
\begin{tabularx}{\textwidth}{M{0.7cm}|Y|Y|Y|Y} 
        \hline (0,0) & $\ket{\psi}$ & $\rho_A$ & $\tr{\rho_A^n}$ or $\tr{\rho_A^{T_2}}^n$ & $\tr{R_\rho^\dagger R_\rho}^{n/2}$ \\\hline
        \upward{3}{$\cB_{AB}$} & \upward{4.5}{\includegraphics[scale=0.6]{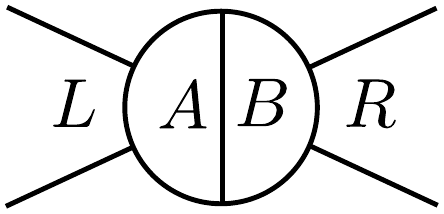}} & \upward{4.5}{\includegraphics[scale=0.6]{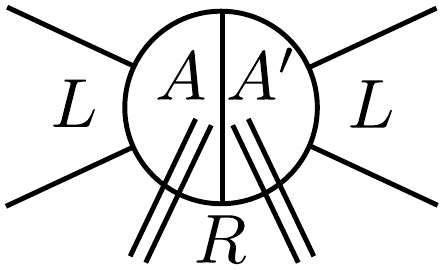}} & \includegraphics[scale=0.6]{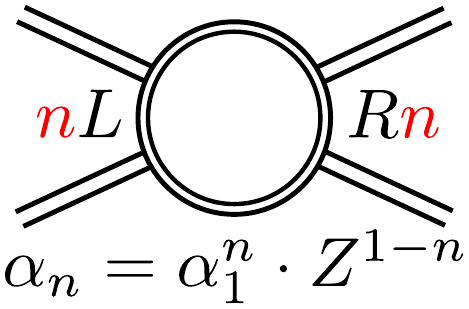} & \includegraphics[scale=0.6]{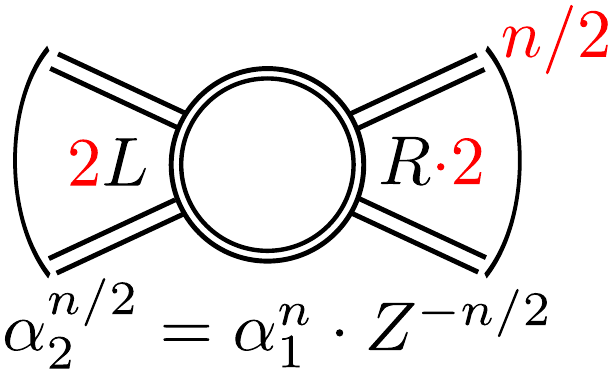}  \\ \hline
        \upward{3.5}{$\cB_{12}$} & \upward{5}{\includegraphics[scale=0.6]{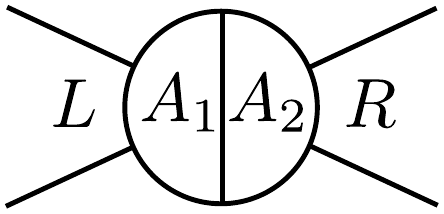}} & \includegraphics[scale=0.6]{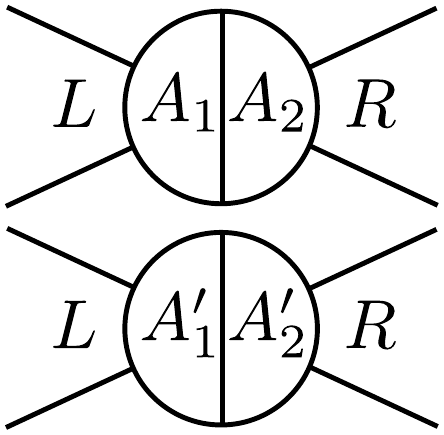} & \multicolumn{2}{c}{\includegraphics[scale=0.6]{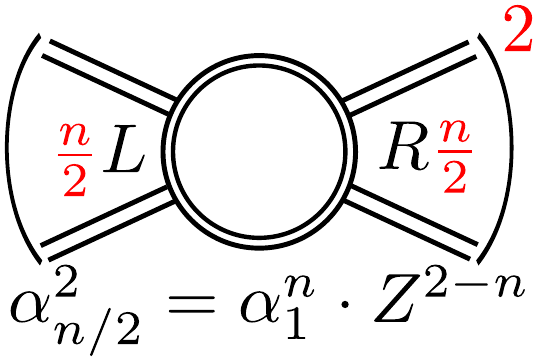} }  \\ \hline
        \upward{3.5}{$\cB_{A}$} & \upward{5}{\includegraphics[scale=0.6]{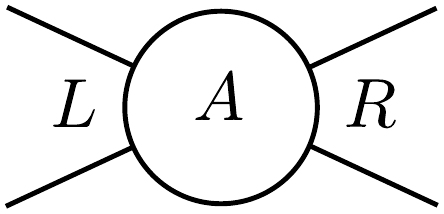}} & \includegraphics[scale=0.6]{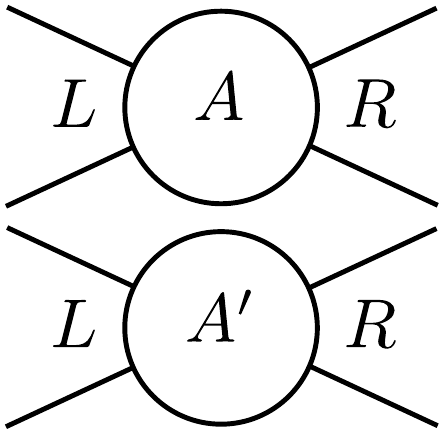} & 
        \multicolumn{2}{c}{\multirow{2}{*}[4em]{\includegraphics[scale=0.6]{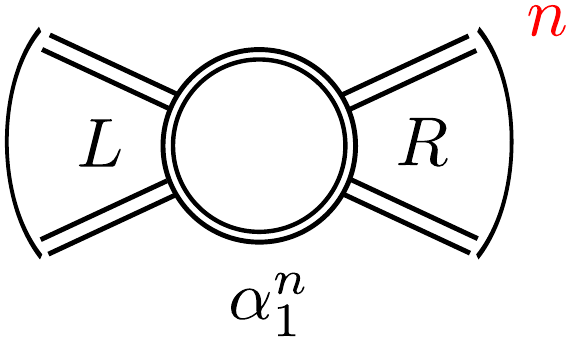}}} \\ \cline{1-3}
        \upward{2}{$\cB_{B}$} & \includegraphics[scale=0.6]{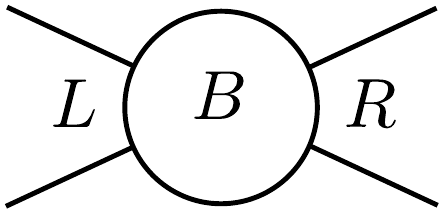} & \includegraphics[scale=0.6]{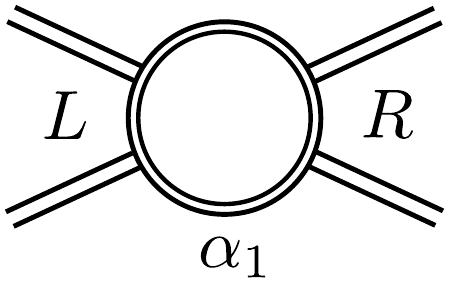} & \multicolumn{2}{c}{} \\ \hline
\end{tabularx}
\caption{\label{tab:noWilson} Local contribution of the four types of balls to various quantities, if there is no WL. We refer to the element in the second row, third column as (2,3) for example. The third column represents both EE and PT in the following way: we only show the PT value $\tr{\rho_A^{T_2}}^n$ with $n$ even for $\cB_{12}$, because for EE it is the same as $\cB_A$. For the other three types, EE and PT are literarily the same, so one plot represents both. For the fourth column $n$ is even. The symbols are defined by Fig.~\ref{fig:tripart}(c) together with the following: $L,R$ means the number of tubes going out of the left/right half-ball, while only two of the tubes are explicitly shown. A double-line means a tube of cross section $S^2$, while a double-lined circle means $S^3$. $\alpha_n$ is defined in \eqref{eq:an}, and the (4,2) element shows that the normalization factor for $n$ copies of $\rho_A$ is $\alpha_1^n$. }
\end{table*}

\subsubsection{Trace over $B$}
For the entanglement measures we consider, the first manipulation on the state is always tracing $B$ to get the density matrix $\rho_A$. Although it is sophisticated to plot the global geometry of the path-integral representation of $\rho_A$, it is clear what happens locally to each ball (together with its attached tubes), as shown in the second column of Table~\ref{tab:noWilson}. A $\cB_A$ or $\cB_{12}$ ball is just duplicated. A $\cB_{AB}$ becomes a ball composed of $A$ and $A'$ half-balls, and the tubes connecting $B$ becomes tubes with cross section $S^2$. A $\cB_B$ becomes a $S^3$, with all its tubes becoming $S^2$-tubes.

This last case actually yields the normalization of the state, where each ball is traced over as a whole to become $S^3$. For one ball, it has $L+R$ $S^2$-tubes. So in the surgery method, we should use $L+R$ balls to cut them, which yields a factor of $Z^{-L-R}$ from \eqref{eq:ZZ=ZZ_app}. Here $Z=\cS_{00}$ is the partition function for a $S^3$ without WLs \eqref{eq:Z=S}. Since each tube is shared by two balls, each ball contributes $Z^{1-(L+R)/2}$ as a whole, where the $Z^1$ factor comes from the ball itself and the rest comes from the tubes attached to it. Then the normalization is \begin{equation}\label{eq:norm}
    \braket{\psi|\psi} = \prod_{\beta\in\cB} \alpha_1(\beta) = Z^{1-g},
\end{equation} 
where we have defined \begin{equation}\label{eq:an}
    \alpha_n(\beta) := Z^{1-[L(\beta)+R(\beta)]n/2},
\end{equation}
for each ball $\beta$ in the representation of $\ket{\psi}$. 
Observe that this result is determined \emph{locally} by looking at the balls individually. It does not depend on which ball is connected to which. This is the key feature that enables all our calculations. Note that in this subsection $\ket{\psi}$ (and $\rho_A$ etc) is the \emph{unnormalized} state represented by path integral with no prefactors, unlike \eqref{eq:psi}.

\subsubsection{Replica method for EE and PT}
Focusing on EE first, in the replica method $n$ copies of $\rho_A$ should connect cyclically. This can be done locally for each ball in $\rho_A$. In Table~\ref{tab:noWilson}(1,2), $\rho_A$ of $\cB_{AB}$ consists $A,A'$ as its half-balls, so the cyclic connection of $n$ balls yields a single $S^3$. The $R$ $S^2$-tubes for each copy of $\rho_A$ are just replicated $n$ times. Among the $2L$ $D^2$-tubes of the $p$-th copy, the $L$ of them connected to $A'$ combine with those connected to $A$ of the $(p+1)$-th copy, making $L$ $S^2$-tubes. Thus there are totally $n(L+R)$ $S^2$-tubes going out of the $S^3$, as shown in Table~\ref{tab:noWilson}(1,3), which contributes a factor $\alpha_n$ \eqref{eq:an} to $\tr{\rho_A^n}$. $\cB_{12}$ is the same as $\cB_A$ for calculating EE, where the $A'$ ball of the $p$-th copy of $\rho_A$ combines with $A$ of the $(p+1)$-th copy to yield a $S^3$. Their tubes combine to $S^2$-tubes accordingly. Thus there will be $n$ $S^3$s, each has $L+R$ $S^2$-tubes, as shown in Table~\ref{tab:noWilson}(3,3). Moreover, the contribution of $\cB_A$ is the same as its contribution $\alpha_1^n$ to normalization. $\cB_B$ is not involved in the contraction process, so it is just replicated $n$ times, giving an identical result to $\cB_A$. This is expected because of the symmetry $S_A = S_B$. In conclusion, the unnormalized $\tr{\rho_A^n}$ is \begin{equation}
    \tr{\rho_A^n} = \left(\prod_{\beta\in\cB_{AB} } \alpha_n(\beta)\right) \left(\prod_{\beta\in \cB-\cB_{AB} } \alpha_1^n(\beta)\right),
\end{equation}
and when canceling with the normalization \eqref{eq:norm}, \begin{equation}
    \frac{\tr{\rho_A^n}}{\braket{\psi|\psi}^n} = \prod_{\beta\in\cB_{AB} } \frac{\alpha_n(\beta)}{\alpha_1^n(\beta)} = \prod_{\beta\in\cB_{AB} } Z^{1-n} = Z^{E_B(1-n)}.
\end{equation}
This leads to the first term in \eqref{eq:SA} by taking the replica limit.

PT is almost the same as EE, because the only difference is that $A_1$ and $A_2$ contract in opposite directions in the cycle. Thus $\cB_{AB},\cB_A,\cB_B$ contributes exactly the same local factor as for EE, since they locally cannot tell the relative direction of contraction. That is why we put EE and PT in one column in Table~\ref{tab:noWilson}, where the EE for $\cB_{12}$ contributes the same amount as $\cB_A$, and the (2,3) element is for PT only. To get Table~\ref{tab:noWilson}(2,3), the $A_1A_2$ ball of the $p$-th copy contracts with the $A_1'A_2'$ balls of the $(p-1)$-th and the $(p+1)$-th. Since $n$ is even for PT, the $A_1A_2$ ball of the odd copies combine with the $A_1'A_2'$ ball of the even copies to get one $S^3$, while the $A_1A_2$ ball of the even copies combine with the $A_1'A_2'$ ball of the odd copies to get another $S^3$. Each of the two $S^3$s contributes a factor $\alpha_{n/2}$ because there are $(L+R)n/2$ $S^2$-tubes. Summing up all balls, we have \begin{align}
    \frac{\tr{\rho_A^{T_2}}^n}{\braket{\psi|\psi}^n} &= \left( \prod_{\beta\in\cB_{AB} } \frac{\alpha_n(\beta)}{\alpha_1^n(\beta)} \right) \left( \prod_{\beta\in\cB_{12} } \frac{\alpha_{n/2}^2(\beta)}{\alpha_1^n(\beta)} \right) \nonumber\\ 
    &= \left( \prod_{\beta\in\cB_{AB} } Z^{1-n} \right) \left( \prod_{\beta\in\cB_{12} } Z^{2-n} \right) \nonumber\\
    &= Z^{E_B(1-n)+(E-E_B)(2-n)},
\end{align}
which leads to the first term in \eqref{eq:ppt} by taking $n\rightarrow 1$.

\subsubsection{Replica method for CCNR}
Consider contracting $n$ copies of $\rho_A$ for CCNR. For the $(2p+1)$-th copy, $A_1$ ($A_1'$) regions connect to $A_1'$ ($A_1$) of the $(2p)$-th copy, while $A_2$ ($A_2'$) regions connect to $A_2'$ ($A_2$) of the $(2p+2)$-th copy. This can be done for the four types of balls in the similar way above, so here we just report the results in the fourth column of Table~\ref{tab:noWilson}. Summing up all balls, we have \begin{align}
    \frac{\tr{R_\rho^\dagger R_\rho}^{n/2}}{\braket{\psi|\psi}^n} &= \left( \prod_{\beta\in\cB_{AB} } \frac{\alpha_2^{n/2}(\beta)}{\alpha_1^n(\beta)} \right) \left( \prod_{\beta\in\cB_{12} } \frac{\alpha_{n/2}^2(\beta)}{\alpha_1^n(\beta)} \right) \nonumber\\
    &= \left( \prod_{\beta\in\cB_{AB} } Z^{-n/2} \right) \left( \prod_{\beta\in\cB_{12} } Z^{2-n} \right) \nonumber\\
    &= Z^{E_B(-n/2)+(E-E_B)(2-n)},
\end{align}
which leads to the first term in \eqref{eq:ccnr} by taking $n\rightarrow 1$.

\subsection{Normalization with Wilson line: Extra factor for party balls}\label{sec:norm}
In the previous subsection we omitted WLs (or in other words, setting $\psi_a(w)=\delta_{a0}$). With WL, the local nature still holds, so we only need to recalculate the balls that the WLs traverse, gather all local extra factors (LEFs) from the WL contribution, and multiply them together with the no-WL result. Then after a careful sum over representations, the final result is in the form of a global extra factor (GEF) multiplied by the no-WL result.

Since there are $4$ types of balls and many ways for a WL to traverse a ball, we separate the discussion into two subsections: this and the next. In this subsection we warm up by focusing on the normalization of the state, where all regions $r_j$ can be viewed as one party. The balls in $\cM_3$ are then all party balls which can be viewed as $\cB_B$. We will check the orthogonality of the basis in \eqref{eq:psi}, and derive results that will be useful in Section \ref{sec:WL}.

\subsubsection{Only one WL, which traverses each ball at most once}
We first consider only $W=1$ WL for now. The state is then \begin{equation}\label{eq:psi1}
    \ket{\psi}=\sum_a \psi_a \ket{a} = \sum_a \tps_a \ket{\widetilde{a}}, \quad \mathrm{where} \quad \tps_a = \frac{1}{\sqrt{\braket{\widetilde{a}|\widetilde{a}}} }\psi_a.
\end{equation}
Here, unlike $\ket{a}$, $\ket{\widetilde{a}}$ is the \emph{unnormalized} state defined by the path integral with the $a$-independent prefactor $Z^{(g-1)/2}$ (see \eqref{eq:norm}) that normalizes the no-WL state $\ket{0}=\ket{\widetilde{0}}$.

If $w$ punctures $\cM$, it can shrink along its trajectory to be contained in one ball; otherwise $w$ is a loop, which we assume traverses each ball at most once for now. Thus we have two possibilities $\cB'_1, \cB'_2$ for a ball that contains WL, as shown in the first column of Table \ref{tab:norm}. Note that only tubes containing the WL are shown explicitly, since other tubes without WL do not contribute to LEF. For normalization $\braket{\psi|\psi}$, each ball combines with its replica to produce a $S^3$, and the $D^2$-tubes become $S^2$-tubes. The representation is denoted by $a$ and $a'$ for the two replicas, and there is a \emph{global} summation over them: $\braket{\psi|\psi} = \sum_{aa'} \tps_a\tps^*_{a'}\braket{\tilde{a}'|\tilde{a}}$, where $\braket{\tilde{a}'|\tilde{a}}$ is $\braket{0|0}$ times the product of all LEFs. 

For $\cB'_1$,  each $S^2$-tube traversed by the WL is cut by a $S^3$ with a Wilson loop in it, leaving another loop in the original $S^3$. Since one loop can only have one representation, the surgery yields nonzero result only when $a=a'$, thus a factor $\delta_{aa'}$ should be included. Physically this is because the $S^2$ cross section of the tube hosts a Hilbert space, where excitation can only be created in anyon anti-anyon pairs. Since there are two tubes containing the WL together with the ball, the LEF is \begin{equation}\label{eq:daa'}
    Q^{\mathrm{N}}(\cB'_1) = \frac{Z_a^{1-2/2} } {Z^{1-2/2}}\delta_{aa'} = \delta_{aa'},
\end{equation}
where $Z_a=\cS_{0a}$ is the partition function for a $S^3$ with an $a$-loop \eqref{eq:Z=S}, and $Z=Z_0$ is the $a=0$ result. The superscript $\mathrm{N}$ refers to normalization.

For $\cB'_2$ in the second row of Table \ref{tab:norm}, the two WLs of the two replicas connect to form a loop in $S^3$, which also requires $a=a'$. After cutting the tubes that all have no WL, the LEF is then \begin{equation}\label{eq:gamma_daa}
    Q^{\mathrm{N}}(\cB'_2) = \frac{Z_a}{Z}\delta_{aa'}= d_a\delta_{aa'},
\end{equation}
where $d_a$ is defined in \eqref{eq:gamma}.

To sum up, the GEF of normalization due to WL is \begin{align}\label{eq:norm_1WL}
    \frac{\braket{\psi|\psi}}{\braket{0|0}} =& \sum_{aa'} \tps_a\tps^*_{a'} \delta_{aa'} \nu_a = \sum_{a} \abs{\tps_a}^2\nu_a, \quad\text{where}\quad \\
    \nu_a:=& \left\{\begin{array}{cr}
        1, & \text{if } w \text{ is a loop}  \\
        d_a, & \text{if } w \text{ is not a loop}
    \end{array}\right.
\end{align}
Here ``not a loop'' means $w$ punctures $\cM$ at two points. Thus in order for $\ket{\psi}$ to be normalized, we have \begin{equation}\label{eq:tps}
    \tps_{a}=\psi_a \nu_a^{-1/2}. 
\end{equation}

\begin{table*}
\begin{tabularx}{\textwidth}{M{0.7cm}|Y|Y|Y} 
        \hline (0,0) & $\ket{\psi}$ & $\braket{\psi|\psi}$ & $\tr{\rho_A^n}$ or $\tr{\rho_A^{T_2}}^n$ or $\tr{R_\rho^\dagger R_\rho}^{n/2}$ \\\hline
        \upward{2}{$\cB'_1$} & \includegraphics[scale=0.6]{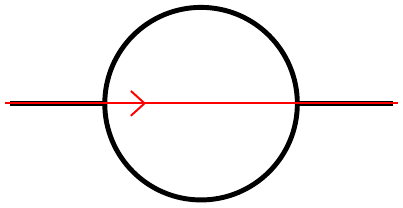} & \includegraphics[scale=0.6]{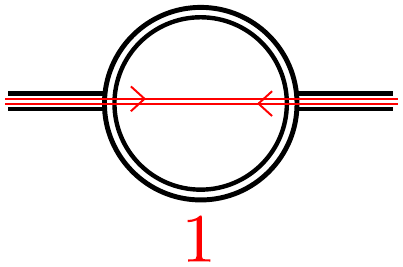} & \includegraphics[scale=0.6]{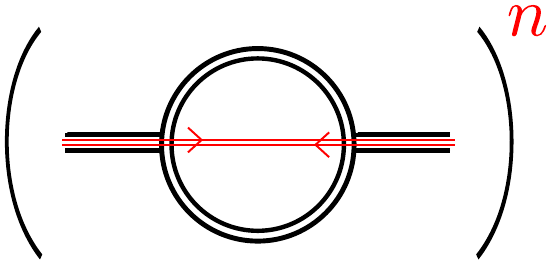} \\ \hline
        \upward{2}{$\cB'_2$} & \includegraphics[scale=0.6]{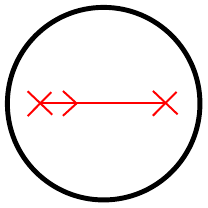} & \includegraphics[scale=0.6]{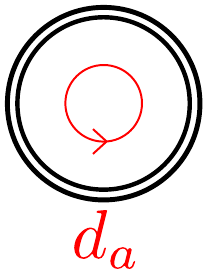} &  \includegraphics[scale=0.6]{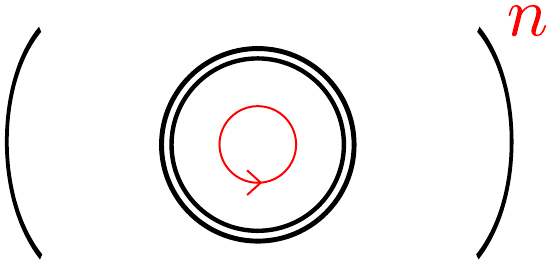} \\ \hline
\end{tabularx}
\caption{\label{tab:norm} LEF of a party ball traversed by one WL. Only the tubes containing the WL are shown. }
\end{table*}

\subsubsection{Multiple WLs, where each WL can traverse a ball multiple times}\label{sec:factor}
Now we consider the general case with $W>1$ WLs, and each WL can traverse a ball more than once. If for each ball there is still at most one WL that passes it, then all balls' contribution to normalization is multiplied so that the total WL contribution factorizes to each individual WL. Formally, the GEF $Q$ is \begin{align}
    Q &= \sum_{a(1)a'(1)a(2)\cdots a'(W)} \left(\prod_{w=1}^W \tps_{a(w)}\tps_{a'(w)}^*\right) \prod_{\beta\in \cB} Q_{\beta} \nonumber\\
    &= \prod_{w=1}^W\left(\sum_{a(w)a'(w)} \tps_{a(w)}\tps_{a'(w)}^* \prod_{\beta\in \cB} Q_{\beta}(w)\right),
\end{align}
as long as the LEF for each ball factorizes \begin{equation}\label{eq:Qnb}
    Q_{\beta} = \prod_{w=1}^W Q_{\beta}(w),
\end{equation}
where $Q_{\beta}(w)$ is the LEF if WLs other than $w$ are all removed. \eqref{eq:Qnb} is trivial if only one WL $w$ passes $\beta$, because then $Q_{\beta}=Q_{\beta}(w)$ and $Q_{\beta}(w') = 1$ for $w'\neq w$. Moreover, if $w$ traverses $\beta$ $M(w,\beta)$ times, then the ball contains $M(w,\beta)$ segments of WL $w$. Intuitively, a ball locally cannot distinguish whether two WL segments in it belong to one WL or not, so we expect a further factorization similar to \eqref{eq:Qnb}: \begin{equation}\label{eq:Qw}
    Q_{\beta}(w) = \prod_{m=1}^{M(w,m)} Q_{\beta}(w,m),
\end{equation}
where $Q_\beta(w,m)$ is the LEF when only the $m$-th segment of $w$ is present.

Before proving \eqref{eq:Qnb} and \eqref{eq:Qw}, we first discuss their implication for normalization. A loop $w$ may traverse a ball multiple times, but since $Q_{\beta}(w,m)$ is always $\delta_{aa'}$ \eqref{eq:daa'} whose power equals itself, the LEF $Q_{\beta}(w)$ is still $\delta_{aa'}$ as long as $w$ traverses $\beta$. On the other hand, if $w$ is not a loop, it is contained in one ball and we do not need \eqref{eq:Qw}. As a result, the GEF is just the product of \eqref{eq:norm_1WL} using \eqref{eq:Qnb}: \begin{equation}
    \frac{\braket{\psi|\psi}}{\braket{0|0}} = \prod_{w=1}^W \sum_{a} \abs{\tps_a(w)}^2\nu_a = \prod_{w=1}^W \sum_{a} \abs{\psi_a(w)}^2.
\end{equation}
This verifies that the basis in \eqref{eq:psi} is orthogonal.

\begin{figure}
    \centering
    \includegraphics[width=0.48\textwidth]{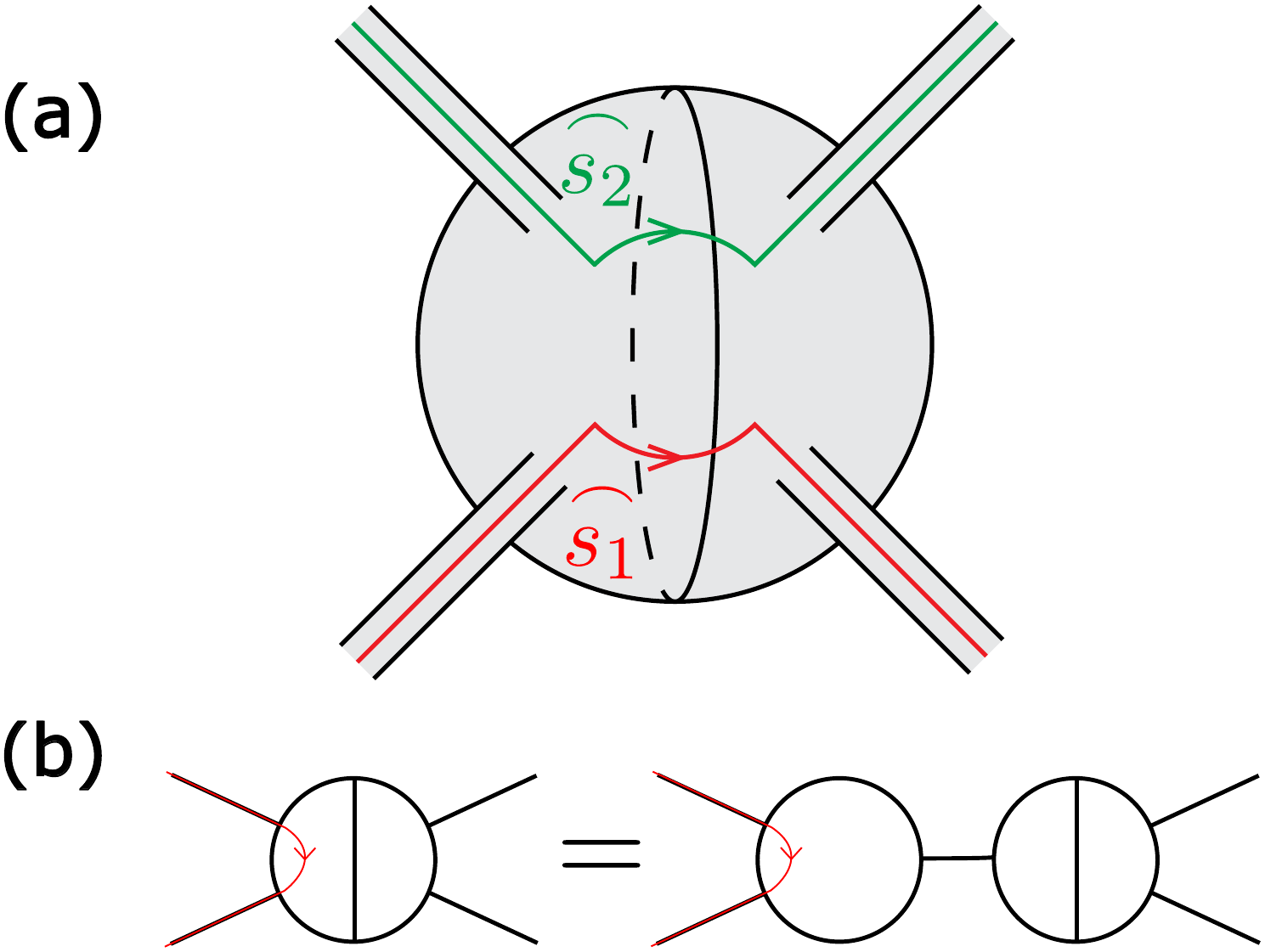}
    \caption{\label{fig:2Wil} (a) Example of two WLs traversing one ball. The WLs lie on the surface and do not intersect due to Assumption \ref{asmpW}. As a consequence, the Wilson loops in the final $S^3$ do not link. (b) A WL in an edge ball that does not traverse the interface, is equivalent to the WL traversing a new party ball, i.e., $\cB'_1$ in Table.~\ref{tab:norm}. }
\end{figure}

It remains to show \eqref{eq:Qnb} and \eqref{eq:Qw}, which combine to just one claim: the LEF for each ball factorizes to the WL segments contained in it. We prove this for two WL segments, since the general case follows similarly.

Suppose WL segments $s_1$ and $s_2$ traverse $\beta$ as shown in Fig.~\ref{fig:2Wil}(a), with representation $a_1$ and $a_2$. They can belong to either one or two WLs globally. For the former case, there is a \emph{global} prefactor $\delta_{a_1a_2}$, but that is irrelevant for LEF here. According to Assumption \ref{asmpW}, the two WL segments live on the surface of the ball, and each $D^2$-tube contains at most one WL segment.
On the ball, the two ``arcs'' $\overset{\frown}{s_1}$ and $\overset{\frown}{s_2}$ do not touch each other. Then it is clear that after gluing replicas and cutting the $S^2$-tubes, the remaining Wilson loops for $s_1$ and $s_2$ in the $S^3$ do not link with each other. Its partition function $Z_{a_1a_2}$ is then factorized due to the surgery operation in Fig.~\ref{fig:surgery_app} with  WL in the tube \begin{equation}\label{eq:Zaa}
    Z_{a_1a_2}\cdot Z = Z_{a_1} Z_{a_2}, \rarrow \frac{Z_{a_1a_2}}{Z} = \frac{Z_{a_1}}{Z} \frac{Z_{a_2}}{Z}.
\end{equation}
On the other hand, LEF due to cutting tubes (which includes factors like $\delta_{a_1a'_1}$) also factorizes because each $S^2$-tube belongs to either $s_1$ or $s_2$, or neither. This establishes the factorization properties \eqref{eq:Qnb} and \eqref{eq:Qw}.

\subsection{Wilson line contribution to entanglement}\label{sec:WL}
In this subsection, we derive the GEF for the three entanglement measures at any replica number $n$, which reduces to our main results \eqref{eq:main} and \eqref{eq:renyi} by combining with the results in Section \ref{sec:noWilson} and taking special limits.
Again, we first consider a state \eqref{eq:psi1} with $W=1$ WL, which is assumed to traverse each ball at most once. We use the factorization property to handle the general case in the final subsubsection.

\subsubsection{WL in an edge ball that traverses the interface}
There are totally seven possibilities for the WL to traverse the interface of an edge ball, as shown by $\tB_1,\cdots,\tB_7$ in the first column of Table \ref{tab:Wilson}. Recall that only tubes containing the WL are shown explicitly. As in Section \ref{sec:noWilson}, we first trace over $B$ to get $\rho_A$, as shown in the second column.

We move on to EE and PT. As discussed in Section \ref{sec:noWilson}, for these two quantities the local contraction structure is the same for $\cB_{AB}$, while for $\cB_{12}$ we only need to consider PT. Thus we report the results in the single third column in Table~\ref{tab:Wilson}.
Taking $\tB_1$ for example, $n$ copies of $\rho_A$ contract to one $S^3$ with $n$ $S^2$-tubes on both sides. Each of the $n$ tubes on the left comes from the $S^2$-tube of one copy. For the right, the $D^2$-tube with $a'_p$-line of the $p$-th copy combine with the $D^2$-tube with $a_{p+1}$-line of the $(p+1)$-th copy to yield one $S^2$-tube, where $p=1,\cdots,n$. In this process all representations are identified to $a$, which yields a factor \begin{equation}\label{eq:11nn}
    [11'\cdots nn'] := \delta_{a_1a'_1}\delta_{a'_1a_2}\delta_{a_2a_2'}\cdots \delta_{a_na_n'}.
\end{equation}
Hereafter we use such shorthand notations to identify all representations in the square brakets.
Furthermore, along the WLs one can traverse all of the $2n$ $S^2$-tubes if hopping between the two lines in one tube is allowed. This means when cutting all the tubes in surgery, there is only one loop of representation $a_1\equiv a$ left in the $S^3$. Since each $S^2$-tube is cut by a $S^3$ with an $a$-loop, the LEF is then \begin{align}
    Q^{\mathrm{EE}}(\tB_1) = Q^{\mathrm{PT}}(\tB_1) &= \frac{Z_a^{1-2n/2}}{Z^{1-2n/2}}[11'\cdots nn'] \nonumber\\ &= d_a^{1-n}[11'\cdots nn'].
\end{align}
The other cases are worked out similarly, for example \begin{equation}\label{eq:12nn}
    Q^{\mathrm{PT}}(\tB_5) = d_{a_1}^{1-\frac{n}{2}} d_{a_2}^{1-\frac{n}{2}}\times [12'\cdots (n-1) n']\times [1'2\cdots (n-1)'n].
\end{equation}

CCNR is similar, for which we report the result for all types of balls in the fourth column of Table~\ref{tab:Wilson}. There are three possibilities locally for which representations are equal: $\cB_{12}$ has \eqref{eq:12nn}, while $\cB_{1B}$ and $\cB_{2B}$ yield \begin{equation}\label{eq:1122}
    [11'22']\times [33'44']\times \cdots \times[(n-1)(n-1)'nn'],
\end{equation}
and \begin{equation}\label{eq:2233}
    [22'33']\times [44'55']\times \cdots \times[nn'11'],
\end{equation}
respectively. For example, \begin{equation}
    Q^{\mathrm{CCNR}}(\tB_1) = \prod_{p=1}^{n/2} d_{a_{2p-1}}^{-1} \times \left(\text{ either \eqref{eq:1122} or \eqref{eq:2233} } \right),
\end{equation}
depending on whether the left half-ball is $A_1$ or $A_2$. Note that PT and CCNR are identical for $\cB_{12}$, which is also manifest in Table~\ref{tab:noWilson}.

\begin{table*}
\begin{tabularx}{\textwidth}{M{0.7cm}|M{2.3cm}|M{2.3cm}|M{1.8cm}|Y|Y|M{2.7cm}} 
    \hline (0,0) & $\ket{\psi}$ & $\rho_A$ & \multicolumn{2}{c|}{ $\tr{\rho_A^n}$ or $\tr{\rho_A^{T_2}}^n$} & \multicolumn{2}{c}{ $\tr{R_\rho^\dagger R_\rho}^{n/2}$} \\\hline
    \upward{3}{$\tB_1$} & \upward{4}{
    \includegraphics[scale=0.55]{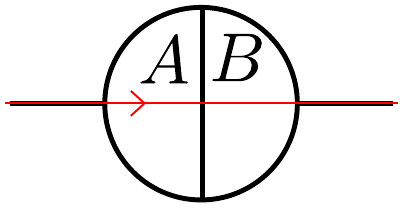}} & 
    \includegraphics[scale=0.55]{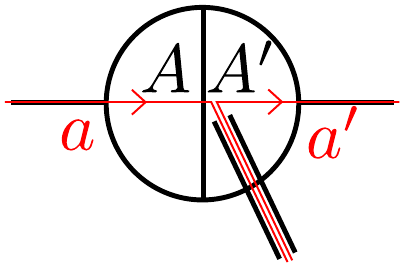} & \multirow[b]{4}{*}{\shortstack{$[11'\cdots nn']$}} &
    \includegraphics[scale=0.55]{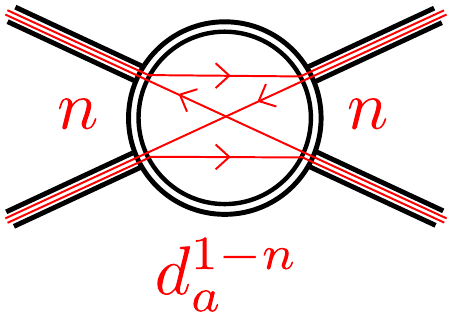} & 
    \includegraphics[scale=0.55]{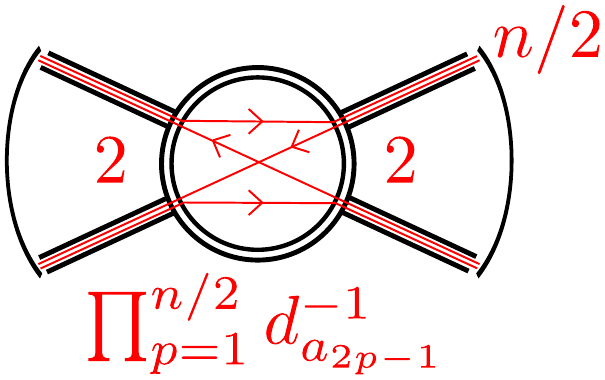} & \multirow[b]{4}{*}{\shortstack{ $[11'22']\times \cdots\times$ \\ $[(n-1)(n-1)'nn']$ \\ or \\ $[22'33']\times \cdots\times$ \\ $[nn'11']$}} \\ \cline{1-3}\cline{5-6}
    \upward{1.5}{$\tB_2$} & 
    \includegraphics[scale=0.55]{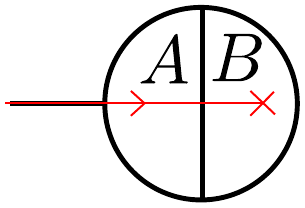} & 
    \includegraphics[scale=0.55]{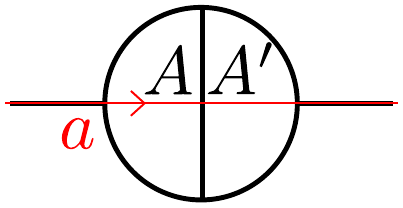} & \multirow{4}{*}{} & \multirow{2}{*}[0.7em]{
    \includegraphics[scale=0.55]{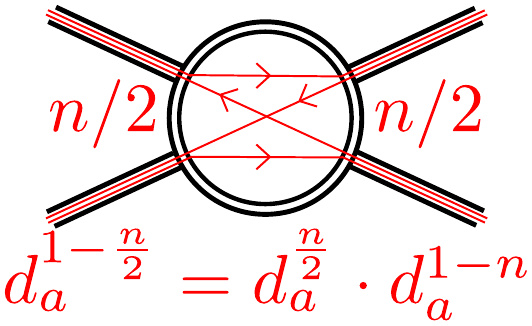}} & \multirow{2}{*}[1.5em]{
    \includegraphics[scale=0.55]{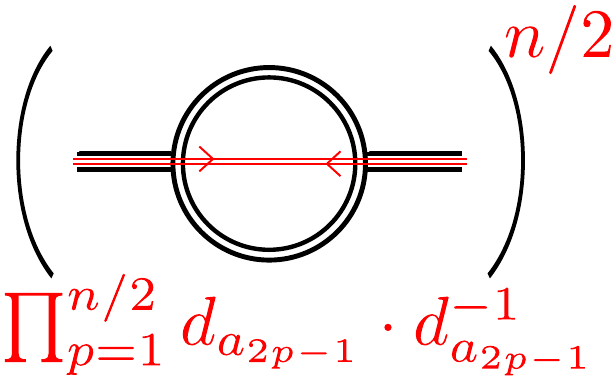} } & \multirow{4}{*}{} \\ \cline{1-3}
    \upward{2}{$\tB_3$} & 
    \includegraphics[scale=0.55]{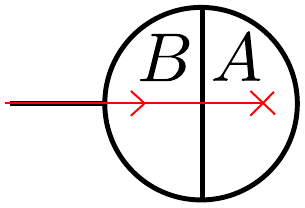} & \includegraphics[scale=0.55]{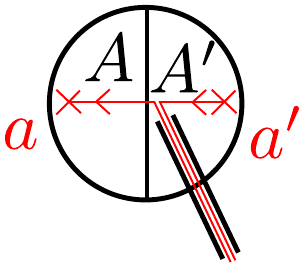} & \multirow{4}{*}{} &
    \multirow{2}{*}{} &  & \multirow{4}{*}{} \\ \cline{1-3}\cline{5-6}
    \upward{3}{$\tB_4$} & \upward{4}{
    \includegraphics[scale=0.55]{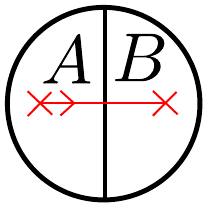}} & \upward{4}{
    \includegraphics[scale=0.55]{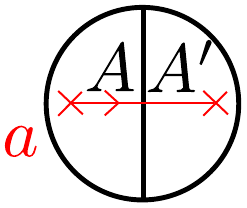}} & \multirow{4}{*}{} &
    \includegraphics[scale=0.55]{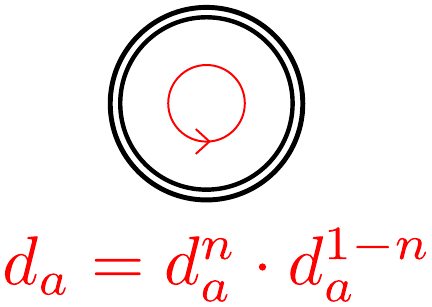} & \includegraphics[scale=0.55]{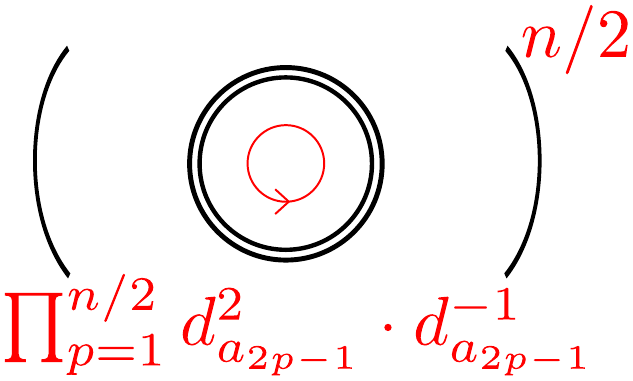} & \multirow{4}{*}{} \\ \hline
    \upward{3}{$\tB_5$} & \upward{4.5}{
    \includegraphics[scale=0.55]{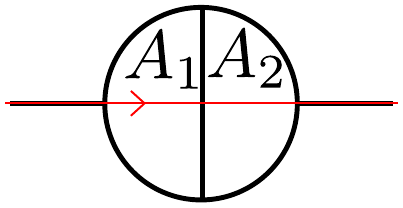}} & \includegraphics[scale=0.55]{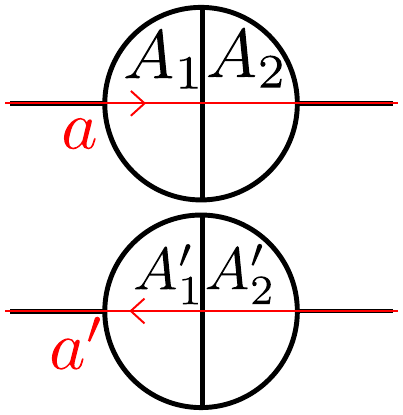} & \multirow[b]{3}{*}{\shortstack{$[12'3\cdots n']$\\ $\times[1'23'\cdots n]$}} & \multicolumn{2}{c|}{
    \includegraphics[scale=0.55]{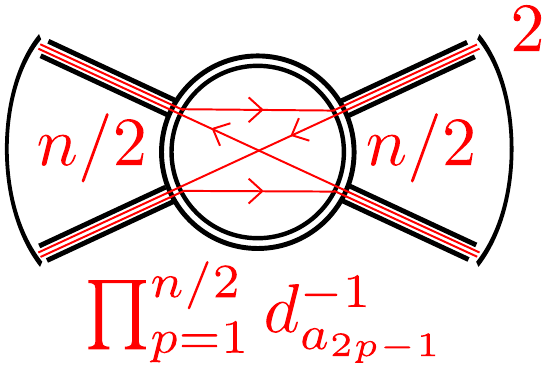}} & \multirow[b]{3}{*}{\shortstack{$[12'3\cdots n']$\\ $\times[1'23'\cdots n]$}} \\ \cline{1-3}\cline{5-6}
    \upward{3}{$\tB_6$} & \upward{4.5}{
    \includegraphics[scale=0.55]{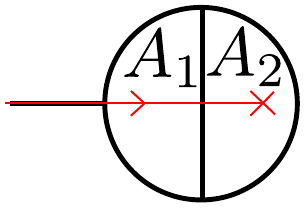}} & \includegraphics[scale=0.55]{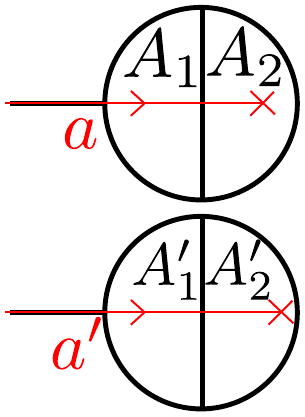} & \multirow{3}{*}{} & \multicolumn{2}{c|}{
    \includegraphics[scale=0.55]{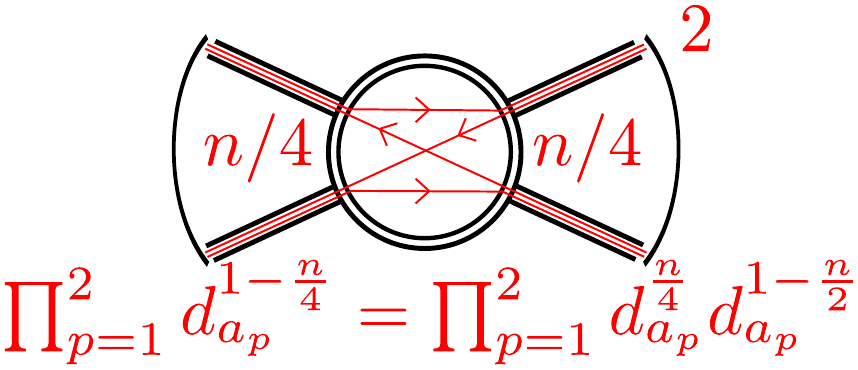}} & \multirow{3}{*}{} \\ \cline{1-3}\cline{5-6}
    \upward{3}{$\tB_7$} & \upward{4.5}{
    \includegraphics[scale=0.55]{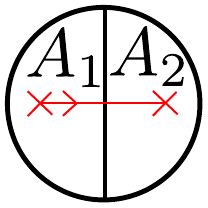}} & \includegraphics[scale=0.55]{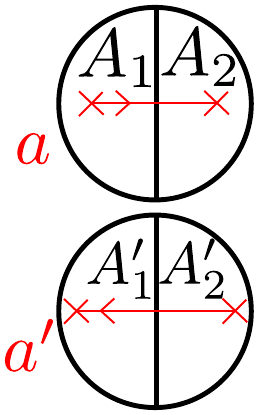} & \multirow{3}{*}{} & \multicolumn{2}{c|}{
    \includegraphics[scale=0.55]{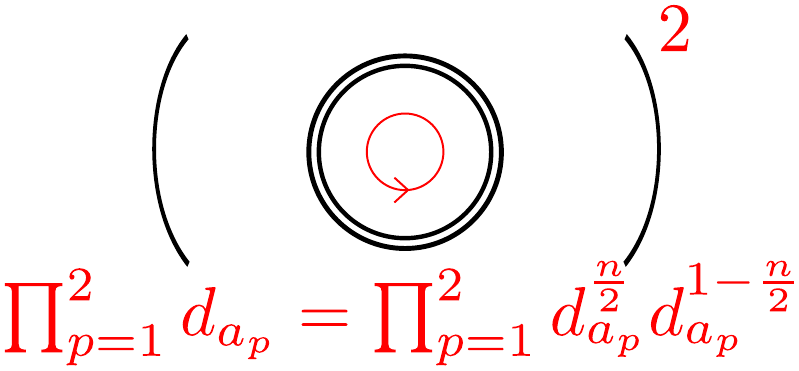}} & \multirow{3}{*}{} \\ \hline
\end{tabularx}
\caption{\label{tab:Wilson} LEF for various quantities for the $7$ types of balls $\tB_1,\cdots,\tB_7$. See Table~\ref{tab:noWilson} on how EE and PT are combined. The LEF depends on both the topology of the final object (the fourth and fifth column), and the way representations are identified, which yields a factor (defined by e.g. \eqref{eq:11nn}) shown in the third and last column. The symbols are the same as Fig.~\ref{fig:tripart}(c) and Table~\ref{tab:noWilson}, and the crossing structure of WLs indicates that each $S^3$ contains exactly one Wilson loop after cutting the $S^2$-tubes. }
\end{table*}

\subsubsection{WL in a party ball}
For a party ball, it can be traversed by the WL in two ways $\cB'_1, \cB'_2$ as shown in Table \ref{tab:norm}. There is no WL segment that goes through one tube and punctures the ball surface once, because it can shrink to disappear in this ball. For each of the two cases at a given replica number $n$, the final topology is the same regardless of the entanglement measure and the party ($A_1, A_2$ or $B$) of the ball, which is already manifest in Table \ref{tab:noWilson}. As shown in the third column of Table \ref{tab:norm}, it is just $n$ disconnected copies of the $n=1$ topology for normalization. The identification of representations, however, depends on the party. We first focus on $B$. Using \eqref{eq:daa'} and \eqref{eq:gamma_daa}, the LEFs are \begin{equation}\label{eq:QB'}
    Q(\cB'_1)= [11']\times [22']\times\cdots\times [nn'],  \quad \text{for } \cB_B
\end{equation}
and \begin{equation}
    Q(\cB'_2)= \prod_{p=1}^nd_{a_p}\times
        [11']\times [22']\times\cdots\times [nn'],  \quad \text{for } \cB_B.
\end{equation}
Since $\cB'_2$ contains the whole WL, the GEF for any entanglement measure is then \begin{equation}\label{eq:sum=1}
    \sum_{a_1a'_1\cdots a'_n} \tps_{a_1}\tps_{a_1'}^*\cdots \tps_{a_n}\tps_{a_n'}^* Q(\cB'_2) = \left(\sum_a \abs{\tps_a}^2 d_a\right)^n =1,
\end{equation}
using \eqref{eq:tps}. This result is intuitive: the WL is created locally in the single region, so it does not affect the entanglement \emph{among} regions. \eqref{eq:sum=1} actually holds for all parties, because the representations are always identified in prime and no-prime pairs.  

On the other hand, the party also does not matter for $\cB'_1$, because we can always set $Q(\cB'_1)=1$ without changing the product of all LEFs (which yields GEF after summing over representations). The reason is that for a WL in $\cB'_1$, it must traverse the interface of another edge ball, whose LEF already includes the representation-identification factor of the party ball. For example, if the party ball is $\cB_B$, then the WL must traverse the interface of another $\cB_{AB}$: otherwise it can shrink to be contained in the party ball. Then the LEF for $\cB_{AB}$ includes a factor either \eqref{eq:11nn}, \eqref{eq:1122} or \eqref{eq:2233}, which all has \eqref{eq:QB'} as a factor.

\subsubsection{WL in an edge ball that does not traverse the interface}
The only case missing in the previous analysis is that the WL traverses an edge ball, but not through its interface: it goes in and out of the same half-ball. As shown in Fig.~\ref{fig:2Wil}(b), this case actually reduces to $\cB'_1$ when separating the ball into two balls: an edge ball with no WL and a party ball that the WL traverses. Thus according to the previous subsubsection, the LEF can be set to $1$.

\subsubsection{GEF for only one WL, which traverses each ball at most once}
The previous three subsubsections discuss all cases assuming there is only one WL, and it traverses each ball at most once. We see that as long as the WL is not contained in a single region where the GEF is trivial \eqref{eq:sum=1}, we only need to consider the cases in Table.~\ref{tab:Wilson} by setting all other LEFs to $1$.

Summing up the $\cB_{AB}$ cases in the third line of Table~\ref{tab:Wilson}, we get the GEF for EE 
\begin{widetext}
\begin{align}\label{eq:EEW}
    \frac{\tr{\rho_A^n}}{\tr{\rho_A^n}_0} &= \sum_{a_1a_1'\cdots a_na_n'} \tps_{a_1}\tps_{a_1'}^*\cdots \tps_{a_n}\tps_{a_n'}^* [11'\cdots nn'] \prod_{\beta\in \tB_1} d_a^{1-n}\prod_{\beta\in \tB_2\cup \tB_3} d_a^{1-n/2}\prod_{\beta\in \tB_4} d_a \nonumber\\ &= \sum_a \abs{\psi_a}^{2n} \nu_a^{-n} \prod_{\beta\in \tB_1} d_a^{1-n}\prod_{\beta\in \tB_2\cup \tB_3} d_a^{1-n/2}\prod_{\beta\in \tB_4} d_a = \sum_a \abs{\psi_a}^{2n} d_a^{(1-n)K_B},
\end{align}
\end{widetext} 
regardless of whether $w$ is a loop or not. Here the subscript $0$ means the normalized no-WL result, and we have used $[11'\cdots nn']^2=[11'\cdots nn']$ together with \eqref{eq:tps}. The difference among $\tB_1,\cdots,\tB_4$ turns out to cancel exactly with the normalization difference, which also holds for $\tB_5,\tB_6,\tB_7$ and for PT and CCNR as we will see.
\eqref{eq:EEW} leads to the second term in \eqref{eq:SA} by taking the limit \begin{equation}
    \lim_{n\rightarrow 1}\frac{1}{1-n}\log \left(\sum_a x_a^ny_a^{1-n}\right) = \sum_a x_a (\log y_a-\log x_a),
\end{equation}
where $\sum_a x_a = 1$.

For PT, we have two possibilities on whether $w$ traverses a $\cB_{AB}$ or not. If $K_B>0$, all representations $a_1,a'_1,\cdots,a_n,a'_n$ are equal due to the factor \eqref{eq:11nn} from one of the $\cB_{AB}$; otherwise they are equal in two disjoint groups with prefactor \eqref{eq:12nn} due to $\cB_{12}$. The GEF for PT is then 
\begin{widetext}
\begin{align}\label{eq:PTW}
    \frac{\tr{\rho_A^{T_2}}^n}{\tr{\rho_A^{T_2}}^n_0} = \left\{\begin{array}{cc}
        \sum_a \abs{\psi_a}^{2n} d_a^{(1-n)K_B}d_a^{(2-n)(K-K_B)}, & K_B>0 \\
        \left(\sum_a \abs{\psi_a}^{n} d_a^{(1-n/2)(K-K_B)} \right)^2,  & K_B=0
    \end{array}\right.
\end{align}

\end{widetext}

For CCNR, since any two of \eqref{eq:12nn},\eqref{eq:1122} and \eqref{eq:2233} multiply to \eqref{eq:11nn}, there are also three possibilities for the global prefactor: Among the three quantities $K_{12}, K_{1B}, K_{2B}$, if at least two are non-zero then the prefactor is \eqref{eq:11nn}; if only $K_{1B}$ or $K_{2B}$ is non-zero the prefactor is \eqref{eq:1122} (which is equivalent to \eqref{eq:2233}); finally if only $K_{12}>0$ the prefactor is \eqref{eq:12nn}. To summarize, the GEF for CCNR is
\begin{widetext}
\begin{align}\label{eq:CCNRW}
    \frac{\tr{R_\rho^\dagger R_\rho}^{n/2}}{\tr{R_\rho^\dagger R_\rho}^{n/2}_0} = \left\{\begin{array}{cr}
        \sum_a \abs{\psi_a}^{2n} d_a^{-nK_B/2}d_a^{(2-n)(K-K_B)}, & \text{if at least two in } \{K_{12}, K_{1B}, K_{2B}\} \text{ are nonzero}
         \\
        \left(\sum_a \abs{\psi_a}^4 d_a^{-K_B} \right)^{n/2},  &  \text{if only one in } \{K_{1B},K_{2B}\} \text{ is nonzero and } K_{12}=0  \\
        \left(\sum_a \abs{\psi_a}^{n} d_a^{(1-n/2)(K-K_B)} \right)^2,  & K_B=0
    \end{array}\right.
\end{align}
\end{widetext}

\subsubsection{General case: Factorization property}
Although Section \ref{sec:norm} only focuses on normalization, the factorization property of LEF in \eqref{eq:Qnb} and \eqref{eq:Qw} generally holds for all three entanglement measures, at any replica number $n$. The LEFs by cutting $S^2$-tubes are factorized simply due to Assumption \ref{asmpW} where at most one WL is contained in each $D^2$-tube. One the other hand, the LEF of Wilson loops in an $S^3$ generated by gluing $\tilde{n}$ balls is also factorized (Section \ref{sec:norm} being the $\tilde{n}=2$ case), because the loops do not link so that \eqref{eq:Zaa} follows. The absence of linking is guaranteed because WLs do not touch on the surface of the ball, so they are in ``different continents'' in the $S^3$. 

Thanks to such factorization, the GEFs \eqref{eq:EEW},\eqref{eq:PTW} and \eqref{eq:CCNRW} hold even if the WL traverses a ball more than once. This is justified by the following two observations. First, the different situations (lines) in \eqref{eq:PTW} and \eqref{eq:CCNRW} only depend on whether each $K$ quantity ($K,K_P$ or $K_{PP'}$) is zero or not. Thus traversing a ball more than once does not change the zeroness of the $K$s, so the situation (line) is unchanged. Second, the exponents of the power of $d_a$ in \eqref{eq:EEW},\eqref{eq:PTW} and \eqref{eq:CCNRW} are all proportional to $K$s. This agrees with the factorization property, because traversing an interface one more time just adds one to the corresponding $K$. Finally, multiple WLs also factorize due to \eqref{eq:Qnb}, so each WL contributes by \eqref{eq:EEW},\eqref{eq:PTW} and \eqref{eq:CCNRW}, while all WLs are summed over in the end. This completes our derivation for \eqref{eq:main} and \eqref{eq:renyi}.

\bibliography{Bib_Refs.bib}

\begin{thebibliography}{54}%
\makeatletter
\providecommand \@ifxundefined [1]{%
 \@ifx{#1\undefined}
}%
\providecommand \@ifnum [1]{%
 \ifnum #1\expandafter \@firstoftwo
 \else \expandafter \@secondoftwo
 \fi
}%
\providecommand \@ifx [1]{%
 \ifx #1\expandafter \@firstoftwo
 \else \expandafter \@secondoftwo
 \fi
}%
\providecommand \natexlab [1]{#1}%
\providecommand \enquote  [1]{``#1''}%
\providecommand \bibnamefont  [1]{#1}%
\providecommand \bibfnamefont [1]{#1}%
\providecommand \citenamefont [1]{#1}%
\providecommand \href@noop [0]{\@secondoftwo}%
\providecommand \href [0]{\begingroup \@sanitize@url \@href}%
\providecommand \@href[1]{\@@startlink{#1}\@@href}%
\providecommand \@@href[1]{\endgroup#1\@@endlink}%
\providecommand \@sanitize@url [0]{\catcode `\\12\catcode `\$12\catcode
  `\&12\catcode `\#12\catcode `\^12\catcode `\_12\catcode `\%12\relax}%
\providecommand \@@startlink[1]{}%
\providecommand \@@endlink[0]{}%
\providecommand \url  [0]{\begingroup\@sanitize@url \@url }%
\providecommand \@url [1]{\endgroup\@href {#1}{\urlprefix }}%
\providecommand \urlprefix  [0]{URL }%
\providecommand \Eprint [0]{\href }%
\providecommand \doibase [0]{https://doi.org/}%
\providecommand \selectlanguage [0]{\@gobble}%
\providecommand \bibinfo  [0]{\@secondoftwo}%
\providecommand \bibfield  [0]{\@secondoftwo}%
\providecommand \translation [1]{[#1]}%
\providecommand \BibitemOpen [0]{}%
\providecommand \bibitemStop [0]{}%
\providecommand \bibitemNoStop [0]{.\EOS\space}%
\providecommand \EOS [0]{\spacefactor3000\relax}%
\providecommand \BibitemShut  [1]{\csname bibitem#1\endcsname}%
\let\auto@bib@innerbib\@empty
\bibitem [{\citenamefont {Satzinger}\ \emph {et~al.}(2021)\citenamefont
  {Satzinger}, \citenamefont {Liu}, \citenamefont {Smith}, \citenamefont
  {Knapp}, \citenamefont {Newman}, \citenamefont {Jones}, \citenamefont {Chen},
  \citenamefont {Quintana}, \citenamefont {Mi}, \citenamefont {Dunsworth},
  \citenamefont {Gidney}, \citenamefont {Aleiner}, \citenamefont {Arute},
  \citenamefont {Arya}, \citenamefont {Atalaya} \emph {et~al.}}]{topo_SC}%
  \BibitemOpen
  \bibfield  {author} {\bibinfo {author} {\bibfnamefont {K.~J.}\ \bibnamefont
  {Satzinger}}, \bibinfo {author} {\bibfnamefont {Y.-J.}\ \bibnamefont {Liu}},
  \bibinfo {author} {\bibfnamefont {A.}~\bibnamefont {Smith}}, \bibinfo
  {author} {\bibfnamefont {C.}~\bibnamefont {Knapp}}, \bibinfo {author}
  {\bibfnamefont {M.}~\bibnamefont {Newman}}, \bibinfo {author} {\bibfnamefont
  {C.}~\bibnamefont {Jones}}, \bibinfo {author} {\bibfnamefont
  {Z.}~\bibnamefont {Chen}}, \bibinfo {author} {\bibfnamefont {C.}~\bibnamefont
  {Quintana}}, \bibinfo {author} {\bibfnamefont {X.}~\bibnamefont {Mi}},
  \bibinfo {author} {\bibfnamefont {A.}~\bibnamefont {Dunsworth}}, \bibinfo
  {author} {\bibfnamefont {C.}~\bibnamefont {Gidney}}, \bibinfo {author}
  {\bibfnamefont {I.}~\bibnamefont {Aleiner}}, \bibinfo {author} {\bibfnamefont
  {F.}~\bibnamefont {Arute}}, \bibinfo {author} {\bibfnamefont
  {K.}~\bibnamefont {Arya}}, \bibinfo {author} {\bibfnamefont {J.}~\bibnamefont
  {Atalaya}}, \emph {et~al.},\ }\bibfield  {title} {\bibinfo {title} {Realizing
  topologically ordered states on a quantum processor},\ }\href
  {https://doi.org/10.1126/science.abi8378} {\bibfield  {journal} {\bibinfo
  {journal} {Science}\ }\textbf {\bibinfo {volume} {374}},\ \bibinfo {pages}
  {1237} (\bibinfo {year} {2021})}\BibitemShut {NoStop}%
\bibitem [{\citenamefont {Semeghini}\ \emph {et~al.}(2021)\citenamefont
  {Semeghini}, \citenamefont {Levine}, \citenamefont {Keesling}, \citenamefont
  {Ebadi}, \citenamefont {Wang}, \citenamefont {Bluvstein}, \citenamefont
  {Verresen}, \citenamefont {Pichler}, \citenamefont {Kalinowski},
  \citenamefont {Samajdar}, \citenamefont {Omran}, \citenamefont {Sachdev},
  \citenamefont {Vishwanath}, \citenamefont {Greiner}, \citenamefont
  {Vuletić},\ and\ \citenamefont {Lukin}}]{topo_Rydberg}%
  \BibitemOpen
  \bibfield  {author} {\bibinfo {author} {\bibfnamefont {G.}~\bibnamefont
  {Semeghini}}, \bibinfo {author} {\bibfnamefont {H.}~\bibnamefont {Levine}},
  \bibinfo {author} {\bibfnamefont {A.}~\bibnamefont {Keesling}}, \bibinfo
  {author} {\bibfnamefont {S.}~\bibnamefont {Ebadi}}, \bibinfo {author}
  {\bibfnamefont {T.~T.}\ \bibnamefont {Wang}}, \bibinfo {author}
  {\bibfnamefont {D.}~\bibnamefont {Bluvstein}}, \bibinfo {author}
  {\bibfnamefont {R.}~\bibnamefont {Verresen}}, \bibinfo {author}
  {\bibfnamefont {H.}~\bibnamefont {Pichler}}, \bibinfo {author} {\bibfnamefont
  {M.}~\bibnamefont {Kalinowski}}, \bibinfo {author} {\bibfnamefont
  {R.}~\bibnamefont {Samajdar}}, \bibinfo {author} {\bibfnamefont
  {A.}~\bibnamefont {Omran}}, \bibinfo {author} {\bibfnamefont
  {S.}~\bibnamefont {Sachdev}}, \bibinfo {author} {\bibfnamefont
  {A.}~\bibnamefont {Vishwanath}}, \bibinfo {author} {\bibfnamefont
  {M.}~\bibnamefont {Greiner}}, \bibinfo {author} {\bibfnamefont
  {V.}~\bibnamefont {Vuletić}},\ and\ \bibinfo {author} {\bibfnamefont
  {M.~D.}\ \bibnamefont {Lukin}},\ }\bibfield  {title} {\bibinfo {title}
  {Probing topological spin liquids on a programmable quantum simulator},\
  }\href {https://doi.org/10.1126/science.abi8794} {\bibfield  {journal}
  {\bibinfo  {journal} {Science}\ }\textbf {\bibinfo {volume} {374}},\ \bibinfo
  {pages} {1242} (\bibinfo {year} {2021})}\BibitemShut {NoStop}%
\bibitem [{\citenamefont {{Andersen}}\ \emph {et~al.}(2022)\citenamefont
  {{Andersen}}, \citenamefont {{Lensky}}, \citenamefont {{Kechedzhi}},
  \citenamefont {{Drozdov}}, \citenamefont {{Bengtsson}}, \citenamefont
  {{Hong}}, \citenamefont {{Morvan}}, \citenamefont {{Mi}}, \citenamefont
  {{Opremcak}}, \citenamefont {{Acharya}}, \citenamefont {{Allen}},
  \citenamefont {{Ansmann}}, \citenamefont {{Arute}}, \citenamefont {{Arya}},
  \citenamefont {{Asfaw}}, \citenamefont {{Atalaya}} \emph
  {et~al.}}]{Google2022TCDefect}%
  \BibitemOpen
  \bibfield  {author} {\bibinfo {author} {\bibfnamefont {T.~I.}\ \bibnamefont
  {{Andersen}}}, \bibinfo {author} {\bibfnamefont {Y.~D.}\ \bibnamefont
  {{Lensky}}}, \bibinfo {author} {\bibfnamefont {K.}~\bibnamefont
  {{Kechedzhi}}}, \bibinfo {author} {\bibfnamefont {I.}~\bibnamefont
  {{Drozdov}}}, \bibinfo {author} {\bibfnamefont {A.}~\bibnamefont
  {{Bengtsson}}}, \bibinfo {author} {\bibfnamefont {S.}~\bibnamefont {{Hong}}},
  \bibinfo {author} {\bibfnamefont {A.}~\bibnamefont {{Morvan}}}, \bibinfo
  {author} {\bibfnamefont {X.}~\bibnamefont {{Mi}}}, \bibinfo {author}
  {\bibfnamefont {A.}~\bibnamefont {{Opremcak}}}, \bibinfo {author}
  {\bibfnamefont {R.}~\bibnamefont {{Acharya}}}, \bibinfo {author}
  {\bibfnamefont {R.}~\bibnamefont {{Allen}}}, \bibinfo {author} {\bibfnamefont
  {M.}~\bibnamefont {{Ansmann}}}, \bibinfo {author} {\bibfnamefont
  {F.}~\bibnamefont {{Arute}}}, \bibinfo {author} {\bibfnamefont
  {K.}~\bibnamefont {{Arya}}}, \bibinfo {author} {\bibfnamefont
  {A.}~\bibnamefont {{Asfaw}}}, \bibinfo {author} {\bibfnamefont
  {J.}~\bibnamefont {{Atalaya}}}, \emph {et~al.},\ }\bibfield  {title}
  {\bibinfo {title} {{Observation of non-Abelian exchange statistics on a
  superconducting processor}},\ }\href
  {https://doi.org/10.48550/arXiv.2210.10255} {\bibfield  {journal} {\bibinfo
  {journal} {arXiv e-prints}\ ,\ \bibinfo {eid} {arXiv:2210.10255}} (\bibinfo
  {year} {2022})},\ \Eprint {https://arxiv.org/abs/2210.10255}
  {arXiv:2210.10255 [quant-ph]} \BibitemShut {NoStop}%
\bibitem [{\citenamefont {{Iqbal}}\ \emph {et~al.}(2023)\citenamefont
  {{Iqbal}}, \citenamefont {{Tantivasadakarn}}, \citenamefont {{Gatterman}},
  \citenamefont {{Gerber}}, \citenamefont {{Gilmore}}, \citenamefont {{Gresh}},
  \citenamefont {{Hankin}}, \citenamefont {{Hewitt}}, \citenamefont {{Horst}},
  \citenamefont {{Matheny}}, \citenamefont {{Mengle}}, \citenamefont
  {{Neyenhuis}}, \citenamefont {{Vishwanath}}, \citenamefont {{Foss-Feig}},
  \citenamefont {{Verresen}},\ and\ \citenamefont {{Dreyer}}}]{Dreyer2023TC}%
  \BibitemOpen
  \bibfield  {author} {\bibinfo {author} {\bibfnamefont {M.}~\bibnamefont
  {{Iqbal}}}, \bibinfo {author} {\bibfnamefont {N.}~\bibnamefont
  {{Tantivasadakarn}}}, \bibinfo {author} {\bibfnamefont {T.~M.}\ \bibnamefont
  {{Gatterman}}}, \bibinfo {author} {\bibfnamefont {J.~A.}\ \bibnamefont
  {{Gerber}}}, \bibinfo {author} {\bibfnamefont {K.}~\bibnamefont {{Gilmore}}},
  \bibinfo {author} {\bibfnamefont {D.}~\bibnamefont {{Gresh}}}, \bibinfo
  {author} {\bibfnamefont {A.}~\bibnamefont {{Hankin}}}, \bibinfo {author}
  {\bibfnamefont {N.}~\bibnamefont {{Hewitt}}}, \bibinfo {author}
  {\bibfnamefont {C.~V.}\ \bibnamefont {{Horst}}}, \bibinfo {author}
  {\bibfnamefont {M.}~\bibnamefont {{Matheny}}}, \bibinfo {author}
  {\bibfnamefont {T.}~\bibnamefont {{Mengle}}}, \bibinfo {author}
  {\bibfnamefont {B.}~\bibnamefont {{Neyenhuis}}}, \bibinfo {author}
  {\bibfnamefont {A.}~\bibnamefont {{Vishwanath}}}, \bibinfo {author}
  {\bibfnamefont {M.}~\bibnamefont {{Foss-Feig}}}, \bibinfo {author}
  {\bibfnamefont {R.}~\bibnamefont {{Verresen}}},\ and\ \bibinfo {author}
  {\bibfnamefont {H.}~\bibnamefont {{Dreyer}}},\ }\bibfield  {title} {\bibinfo
  {title} {{Topological Order from Measurements and Feed-Forward on a Trapped
  Ion Quantum Computer}},\ }\href {https://doi.org/10.48550/arXiv.2302.01917}
  {\bibfield  {journal} {\bibinfo  {journal} {arXiv e-prints}\ ,\ \bibinfo
  {eid} {arXiv:2302.01917}} (\bibinfo {year} {2023})},\ \Eprint
  {https://arxiv.org/abs/2302.01917} {arXiv:2302.01917 [quant-ph]} \BibitemShut
  {NoStop}%
\bibitem [{\citenamefont {Iqbal}\ \emph {et~al.}(2023)\citenamefont {Iqbal}
  \emph {et~al.}}]{quantinuum_nonAbel23}%
  \BibitemOpen
  \bibfield  {author} {\bibinfo {author} {\bibfnamefont {M.}~\bibnamefont
  {Iqbal}} \emph {et~al.},\ }\bibfield  {title} {\bibinfo {title} {{Creation of
  Non-Abelian Topological Order and Anyons on a Trapped-Ion Processor}},\
  }\href@noop {} {\  (\bibinfo {year} {2023})},\ \Eprint
  {https://arxiv.org/abs/2305.03766} {arXiv:2305.03766 [quant-ph]} \BibitemShut
  {NoStop}%
\bibitem [{\citenamefont {Zeng}\ \emph {et~al.}(2019)\citenamefont {Zeng},
  \citenamefont {Chen}, \citenamefont {Zhou}, \citenamefont {Wen} \emph
  {et~al.}}]{QI_meet_QM}%
  \BibitemOpen
  \bibfield  {author} {\bibinfo {author} {\bibfnamefont {B.}~\bibnamefont
  {Zeng}}, \bibinfo {author} {\bibfnamefont {X.}~\bibnamefont {Chen}}, \bibinfo
  {author} {\bibfnamefont {D.-L.}\ \bibnamefont {Zhou}}, \bibinfo {author}
  {\bibfnamefont {X.-G.}\ \bibnamefont {Wen}}, \emph {et~al.},\ }\href@noop {}
  {\emph {\bibinfo {title} {Quantum information meets quantum matter}}}\
  (\bibinfo  {publisher} {Springer},\ \bibinfo {year} {2019})\BibitemShut
  {NoStop}%
\bibitem [{\citenamefont {{Hamma}}\ \emph
  {et~al.}(2005{\natexlab{a}})\citenamefont {{Hamma}}, \citenamefont
  {{Ionicioiu}},\ and\ \citenamefont {{Zanardi}}}]{Hamma2005TCEE1}%
  \BibitemOpen
  \bibfield  {author} {\bibinfo {author} {\bibfnamefont {A.}~\bibnamefont
  {{Hamma}}}, \bibinfo {author} {\bibfnamefont {R.}~\bibnamefont
  {{Ionicioiu}}},\ and\ \bibinfo {author} {\bibfnamefont {P.}~\bibnamefont
  {{Zanardi}}},\ }\bibfield  {title} {\bibinfo {title} {{Ground state
  entanglement and geometric entropy in the Kitaev model [rapid
  communication]}},\ }\href {https://doi.org/10.1016/j.physleta.2005.01.060}
  {\bibfield  {journal} {\bibinfo  {journal} {Physics Letters A}\ }\textbf
  {\bibinfo {volume} {337}},\ \bibinfo {pages} {22} (\bibinfo {year}
  {2005}{\natexlab{a}})},\ \Eprint {https://arxiv.org/abs/quant-ph/0406202}
  {arXiv:quant-ph/0406202 [quant-ph]} \BibitemShut {NoStop}%
\bibitem [{\citenamefont {{Hamma}}\ \emph
  {et~al.}(2005{\natexlab{b}})\citenamefont {{Hamma}}, \citenamefont
  {{Ionicioiu}},\ and\ \citenamefont {{Zanardi}}}]{Hamma2005TCEE2}%
  \BibitemOpen
  \bibfield  {author} {\bibinfo {author} {\bibfnamefont {A.}~\bibnamefont
  {{Hamma}}}, \bibinfo {author} {\bibfnamefont {R.}~\bibnamefont
  {{Ionicioiu}}},\ and\ \bibinfo {author} {\bibfnamefont {P.}~\bibnamefont
  {{Zanardi}}},\ }\bibfield  {title} {\bibinfo {title} {{Bipartite entanglement
  and entropic boundary law in lattice spin systems}},\ }\href
  {https://doi.org/10.1103/PhysRevA.71.022315} {\bibfield  {journal} {\bibinfo
  {journal} {\pra}\ }\textbf {\bibinfo {volume} {71}},\ \bibinfo {eid} {022315}
  (\bibinfo {year} {2005}{\natexlab{b}})},\ \Eprint
  {https://arxiv.org/abs/quant-ph/0409073} {arXiv:quant-ph/0409073 [quant-ph]}
  \BibitemShut {NoStop}%
\bibitem [{\citenamefont {{Kitaev}}\ and\ \citenamefont
  {{Preskill}}(2006)}]{Kitaev2006TEE}%
  \BibitemOpen
  \bibfield  {author} {\bibinfo {author} {\bibfnamefont {A.}~\bibnamefont
  {{Kitaev}}}\ and\ \bibinfo {author} {\bibfnamefont {J.}~\bibnamefont
  {{Preskill}}},\ }\bibfield  {title} {\bibinfo {title} {{Topological
  Entanglement Entropy}},\ }\href
  {https://doi.org/10.1103/PhysRevLett.96.110404} {\bibfield  {journal}
  {\bibinfo  {journal} {\prl}\ }\textbf {\bibinfo {volume} {96}},\ \bibinfo
  {eid} {110404} (\bibinfo {year} {2006})},\ \Eprint
  {https://arxiv.org/abs/hep-th/0510092} {arXiv:hep-th/0510092 [hep-th]}
  \BibitemShut {NoStop}%
\bibitem [{\citenamefont {{Levin}}\ and\ \citenamefont
  {{Wen}}(2006)}]{Levin2006TEE}%
  \BibitemOpen
  \bibfield  {author} {\bibinfo {author} {\bibfnamefont {M.}~\bibnamefont
  {{Levin}}}\ and\ \bibinfo {author} {\bibfnamefont {X.-G.}\ \bibnamefont
  {{Wen}}},\ }\bibfield  {title} {\bibinfo {title} {{Detecting Topological
  Order in a Ground State Wave Function}},\ }\href
  {https://doi.org/10.1103/PhysRevLett.96.110405} {\bibfield  {journal}
  {\bibinfo  {journal} {\prl}\ }\textbf {\bibinfo {volume} {96}},\ \bibinfo
  {eid} {110405} (\bibinfo {year} {2006})},\ \Eprint
  {https://arxiv.org/abs/cond-mat/0510613} {arXiv:cond-mat/0510613
  [cond-mat.str-el]} \BibitemShut {NoStop}%
\bibitem [{\citenamefont {{Dong}}\ \emph {et~al.}(2008)\citenamefont {{Dong}},
  \citenamefont {{Fradkin}}, \citenamefont {{Leigh}},\ and\ \citenamefont
  {{Nowling}}}]{Fradkin2008TEE}%
  \BibitemOpen
  \bibfield  {author} {\bibinfo {author} {\bibfnamefont {S.}~\bibnamefont
  {{Dong}}}, \bibinfo {author} {\bibfnamefont {E.}~\bibnamefont {{Fradkin}}},
  \bibinfo {author} {\bibfnamefont {R.~G.}\ \bibnamefont {{Leigh}}},\ and\
  \bibinfo {author} {\bibfnamefont {S.}~\bibnamefont {{Nowling}}},\ }\bibfield
  {title} {\bibinfo {title} {{Topological entanglement entropy in Chern-Simons
  theories and quantum Hall fluids}},\ }\href
  {https://doi.org/10.1088/1126-6708/2008/05/016} {\bibfield  {journal}
  {\bibinfo  {journal} {Journal of High Energy Physics}\ }\textbf {\bibinfo
  {volume} {2008}},\ \bibinfo {eid} {016} (\bibinfo {year} {2008})},\ \Eprint
  {https://arxiv.org/abs/0802.3231} {arXiv:0802.3231 [hep-th]} \BibitemShut
  {NoStop}%
\bibitem [{\citenamefont {{Zhang}}\ \emph {et~al.}(2012)\citenamefont
  {{Zhang}}, \citenamefont {{Grover}}, \citenamefont {{Turner}}, \citenamefont
  {{Oshikawa}},\ and\ \citenamefont {{Vishwanath}}}]{Vishwanath2012TEE}%
  \BibitemOpen
  \bibfield  {author} {\bibinfo {author} {\bibfnamefont {Y.}~\bibnamefont
  {{Zhang}}}, \bibinfo {author} {\bibfnamefont {T.}~\bibnamefont {{Grover}}},
  \bibinfo {author} {\bibfnamefont {A.}~\bibnamefont {{Turner}}}, \bibinfo
  {author} {\bibfnamefont {M.}~\bibnamefont {{Oshikawa}}},\ and\ \bibinfo
  {author} {\bibfnamefont {A.}~\bibnamefont {{Vishwanath}}},\ }\bibfield
  {title} {\bibinfo {title} {{Quasiparticle statistics and braiding from
  ground-state entanglement}},\ }\href
  {https://doi.org/10.1103/PhysRevB.85.235151} {\bibfield  {journal} {\bibinfo
  {journal} {\prb}\ }\textbf {\bibinfo {volume} {85}},\ \bibinfo {eid} {235151}
  (\bibinfo {year} {2012})},\ \Eprint {https://arxiv.org/abs/1111.2342}
  {arXiv:1111.2342 [cond-mat.str-el]} \BibitemShut {NoStop}%
\bibitem [{\citenamefont {{Grover}}\ \emph {et~al.}(2013)\citenamefont
  {{Grover}}, \citenamefont {{Zhang}},\ and\ \citenamefont
  {{Vishwanath}}}]{Vishwanath2013TEEReview}%
  \BibitemOpen
  \bibfield  {author} {\bibinfo {author} {\bibfnamefont {T.}~\bibnamefont
  {{Grover}}}, \bibinfo {author} {\bibfnamefont {Y.}~\bibnamefont {{Zhang}}},\
  and\ \bibinfo {author} {\bibfnamefont {A.}~\bibnamefont {{Vishwanath}}},\
  }\bibfield  {title} {\bibinfo {title} {{Entanglement entropy as a portal to
  the physics of quantum spin liquids}},\ }\href
  {https://doi.org/10.1088/1367-2630/15/2/025002} {\bibfield  {journal}
  {\bibinfo  {journal} {New Journal of Physics}\ }\textbf {\bibinfo {volume}
  {15}},\ \bibinfo {eid} {025002} (\bibinfo {year} {2013})},\ \Eprint
  {https://arxiv.org/abs/1302.0899} {arXiv:1302.0899 [cond-mat.str-el]}
  \BibitemShut {NoStop}%
\bibitem [{\citenamefont {{Shi}}\ \emph {et~al.}(2020)\citenamefont {{Shi}},
  \citenamefont {{Kato}},\ and\ \citenamefont
  {{Kim}}}]{Kim2020EntBootstrapFusion}%
  \BibitemOpen
  \bibfield  {author} {\bibinfo {author} {\bibfnamefont {B.}~\bibnamefont
  {{Shi}}}, \bibinfo {author} {\bibfnamefont {K.}~\bibnamefont {{Kato}}},\ and\
  \bibinfo {author} {\bibfnamefont {I.~H.}\ \bibnamefont {{Kim}}},\ }\bibfield
  {title} {\bibinfo {title} {{Fusion rules from entanglement}},\ }\href
  {https://doi.org/10.1016/j.aop.2020.168164} {\bibfield  {journal} {\bibinfo
  {journal} {Annals of Physics}\ }\textbf {\bibinfo {volume} {418}},\ \bibinfo
  {eid} {168164} (\bibinfo {year} {2020})},\ \Eprint
  {https://arxiv.org/abs/1906.09376} {arXiv:1906.09376 [cond-mat.str-el]}
  \BibitemShut {NoStop}%
\bibitem [{\citenamefont {{Shi}}\ and\ \citenamefont
  {{Kim}}(2021)}]{Kim2021EntBootstrapDW}%
  \BibitemOpen
  \bibfield  {author} {\bibinfo {author} {\bibfnamefont {B.}~\bibnamefont
  {{Shi}}}\ and\ \bibinfo {author} {\bibfnamefont {I.~H.}\ \bibnamefont
  {{Kim}}},\ }\bibfield  {title} {\bibinfo {title} {{Entanglement bootstrap
  approach for gapped domain walls}},\ }\href
  {https://doi.org/10.1103/PhysRevB.103.115150} {\bibfield  {journal} {\bibinfo
   {journal} {\prb}\ }\textbf {\bibinfo {volume} {103}},\ \bibinfo {eid}
  {115150} (\bibinfo {year} {2021})},\ \Eprint
  {https://arxiv.org/abs/2008.11793} {arXiv:2008.11793 [cond-mat.str-el]}
  \BibitemShut {NoStop}%
\bibitem [{\citenamefont {{Peres}}(1996)}]{Peres1996PT}%
  \BibitemOpen
  \bibfield  {author} {\bibinfo {author} {\bibfnamefont {A.}~\bibnamefont
  {{Peres}}},\ }\bibfield  {title} {\bibinfo {title} {{Separability Criterion
  for Density Matrices}},\ }\href {https://doi.org/10.1103/PhysRevLett.77.1413}
  {\bibfield  {journal} {\bibinfo  {journal} {\prl}\ }\textbf {\bibinfo
  {volume} {77}},\ \bibinfo {pages} {1413} (\bibinfo {year} {1996})},\ \Eprint
  {https://arxiv.org/abs/quant-ph/9604005} {arXiv:quant-ph/9604005 [quant-ph]}
  \BibitemShut {NoStop}%
\bibitem [{\citenamefont {{Vidal}}\ and\ \citenamefont
  {{Werner}}(2002)}]{Vidal2002Negativity}%
  \BibitemOpen
  \bibfield  {author} {\bibinfo {author} {\bibfnamefont {G.}~\bibnamefont
  {{Vidal}}}\ and\ \bibinfo {author} {\bibfnamefont {R.~F.}\ \bibnamefont
  {{Werner}}},\ }\bibfield  {title} {\bibinfo {title} {{Computable measure of
  entanglement}},\ }\href {https://doi.org/10.1103/PhysRevA.65.032314}
  {\bibfield  {journal} {\bibinfo  {journal} {\pra}\ }\textbf {\bibinfo
  {volume} {65}},\ \bibinfo {eid} {032314} (\bibinfo {year} {2002})},\ \Eprint
  {https://arxiv.org/abs/quant-ph/0102117} {arXiv:quant-ph/0102117 [quant-ph]}
  \BibitemShut {NoStop}%
\bibitem [{\citenamefont {{Lee}}\ and\ \citenamefont
  {{Vidal}}(2013)}]{Vidal2013TCNeg}%
  \BibitemOpen
  \bibfield  {author} {\bibinfo {author} {\bibfnamefont {Y.~A.}\ \bibnamefont
  {{Lee}}}\ and\ \bibinfo {author} {\bibfnamefont {G.}~\bibnamefont
  {{Vidal}}},\ }\bibfield  {title} {\bibinfo {title} {{Entanglement negativity
  and topological order}},\ }\href {https://doi.org/10.1103/PhysRevA.88.042318}
  {\bibfield  {journal} {\bibinfo  {journal} {\pra}\ }\textbf {\bibinfo
  {volume} {88}},\ \bibinfo {eid} {042318} (\bibinfo {year} {2013})},\ \Eprint
  {https://arxiv.org/abs/1306.5711} {arXiv:1306.5711 [quant-ph]} \BibitemShut
  {NoStop}%
\bibitem [{\citenamefont {{Castelnovo}}(2013)}]{Castelnovo2013TCNeg}%
  \BibitemOpen
  \bibfield  {author} {\bibinfo {author} {\bibfnamefont {C.}~\bibnamefont
  {{Castelnovo}}},\ }\bibfield  {title} {\bibinfo {title} {{Negativity and
  topological order in the toric code}},\ }\href
  {https://doi.org/10.1103/PhysRevA.88.042319} {\bibfield  {journal} {\bibinfo
  {journal} {\pra}\ }\textbf {\bibinfo {volume} {88}},\ \bibinfo {eid} {042319}
  (\bibinfo {year} {2013})},\ \Eprint {https://arxiv.org/abs/1306.4990}
  {arXiv:1306.4990 [cond-mat.str-el]} \BibitemShut {NoStop}%
\bibitem [{\citenamefont {{Wen}}\ \emph
  {et~al.}(2016{\natexlab{a}})\citenamefont {{Wen}}, \citenamefont
  {{Matsuura}},\ and\ \citenamefont {{Ryu}}}]{Wen2016NegBdry}%
  \BibitemOpen
  \bibfield  {author} {\bibinfo {author} {\bibfnamefont {X.}~\bibnamefont
  {{Wen}}}, \bibinfo {author} {\bibfnamefont {S.}~\bibnamefont {{Matsuura}}},\
  and\ \bibinfo {author} {\bibfnamefont {S.}~\bibnamefont {{Ryu}}},\ }\bibfield
   {title} {\bibinfo {title} {{Edge theory approach to topological entanglement
  entropy, mutual information, and entanglement negativity in Chern-Simons
  theories}},\ }\href {https://doi.org/10.1103/PhysRevB.93.245140} {\bibfield
  {journal} {\bibinfo  {journal} {\prb}\ }\textbf {\bibinfo {volume} {93}},\
  \bibinfo {eid} {245140} (\bibinfo {year} {2016}{\natexlab{a}})},\ \Eprint
  {https://arxiv.org/abs/1603.08534} {arXiv:1603.08534 [cond-mat.mes-hall]}
  \BibitemShut {NoStop}%
\bibitem [{\citenamefont {{Wen}}\ \emph
  {et~al.}(2016{\natexlab{b}})\citenamefont {{Wen}}, \citenamefont {{Chang}},\
  and\ \citenamefont {{Ryu}}}]{Wen2016NegSurgery}%
  \BibitemOpen
  \bibfield  {author} {\bibinfo {author} {\bibfnamefont {X.}~\bibnamefont
  {{Wen}}}, \bibinfo {author} {\bibfnamefont {P.-Y.}\ \bibnamefont {{Chang}}},\
  and\ \bibinfo {author} {\bibfnamefont {S.}~\bibnamefont {{Ryu}}},\ }\bibfield
   {title} {\bibinfo {title} {{Topological entanglement negativity in
  Chern-Simons theories}},\ }\href {https://doi.org/10.1007/JHEP09(2016)012}
  {\bibfield  {journal} {\bibinfo  {journal} {Journal of High Energy Physics}\
  }\textbf {\bibinfo {volume} {2016}},\ \bibinfo {eid} {12} (\bibinfo {year}
  {2016}{\natexlab{b}})},\ \Eprint {https://arxiv.org/abs/1606.04118}
  {arXiv:1606.04118 [cond-mat.str-el]} \BibitemShut {NoStop}%
\bibitem [{\citenamefont {{Lim}}\ \emph {et~al.}(2021)\citenamefont {{Lim}},
  \citenamefont {{Asasi}}, \citenamefont {{Teo}},\ and\ \citenamefont
  {{Mulligan}}}]{Mulligan2021Disentangling}%
  \BibitemOpen
  \bibfield  {author} {\bibinfo {author} {\bibfnamefont {P.~K.}\ \bibnamefont
  {{Lim}}}, \bibinfo {author} {\bibfnamefont {H.}~\bibnamefont {{Asasi}}},
  \bibinfo {author} {\bibfnamefont {J.~C.~Y.}\ \bibnamefont {{Teo}}},\ and\
  \bibinfo {author} {\bibfnamefont {M.}~\bibnamefont {{Mulligan}}},\ }\bibfield
   {title} {\bibinfo {title} {{Disentangling (2 +1 )D topological states of
  matter with entanglement negativity}},\ }\href
  {https://doi.org/10.1103/PhysRevB.104.115155} {\bibfield  {journal} {\bibinfo
   {journal} {\prb}\ }\textbf {\bibinfo {volume} {104}},\ \bibinfo {eid}
  {115155} (\bibinfo {year} {2021})},\ \Eprint
  {https://arxiv.org/abs/2106.07668} {arXiv:2106.07668 [cond-mat.str-el]}
  \BibitemShut {NoStop}%
\bibitem [{\citenamefont {{Liu}}\ \emph
  {et~al.}(2022{\natexlab{a}})\citenamefont {{Liu}}, \citenamefont {{Sohal}},
  \citenamefont {{Kudler-Flam}},\ and\ \citenamefont
  {{Ryu}}}]{Ryu2022VertexStates}%
  \BibitemOpen
  \bibfield  {author} {\bibinfo {author} {\bibfnamefont {Y.}~\bibnamefont
  {{Liu}}}, \bibinfo {author} {\bibfnamefont {R.}~\bibnamefont {{Sohal}}},
  \bibinfo {author} {\bibfnamefont {J.}~\bibnamefont {{Kudler-Flam}}},\ and\
  \bibinfo {author} {\bibfnamefont {S.}~\bibnamefont {{Ryu}}},\ }\bibfield
  {title} {\bibinfo {title} {{Multipartitioning topological phases by vertex
  states and quantum entanglement}},\ }\href
  {https://doi.org/10.1103/PhysRevB.105.115107} {\bibfield  {journal} {\bibinfo
   {journal} {\prb}\ }\textbf {\bibinfo {volume} {105}},\ \bibinfo {eid}
  {115107} (\bibinfo {year} {2022}{\natexlab{a}})},\ \Eprint
  {https://arxiv.org/abs/2110.11980} {arXiv:2110.11980 [cond-mat.str-el]}
  \BibitemShut {NoStop}%
\bibitem [{\citenamefont {{Liu}}\ \emph
  {et~al.}(2022{\natexlab{b}})\citenamefont {{Liu}}, \citenamefont
  {{Geoffrion}},\ and\ \citenamefont {{Witczak-Krempa}}}]{LiuCC2022IQHNeg}%
  \BibitemOpen
  \bibfield  {author} {\bibinfo {author} {\bibfnamefont {C.-C.}\ \bibnamefont
  {{Liu}}}, \bibinfo {author} {\bibfnamefont {J.}~\bibnamefont {{Geoffrion}}},\
  and\ \bibinfo {author} {\bibfnamefont {W.}~\bibnamefont {{Witczak-Krempa}}},\
  }\bibfield  {title} {\bibinfo {title} {{Entanglement negativity versus mutual
  information in the quantum Hall effect and beyond}},\ }\href@noop {}
  {\bibfield  {journal} {\bibinfo  {journal} {arXiv e-prints}\ ,\ \bibinfo
  {eid} {arXiv:2208.12819}} (\bibinfo {year} {2022}{\natexlab{b}})},\ \Eprint
  {https://arxiv.org/abs/2208.12819} {arXiv:2208.12819 [cond-mat.str-el]}
  \BibitemShut {NoStop}%
\bibitem [{\citenamefont {{Rodr{\'\i}guez}}\ and\ \citenamefont
  {{Sierra}}(2010)}]{Sierra2010IQHCornerEE}%
  \BibitemOpen
  \bibfield  {author} {\bibinfo {author} {\bibfnamefont {I.~D.}\ \bibnamefont
  {{Rodr{\'\i}guez}}}\ and\ \bibinfo {author} {\bibfnamefont {G.}~\bibnamefont
  {{Sierra}}},\ }\bibfield  {title} {\bibinfo {title} {{Entanglement entropy of
  integer quantum Hall states in polygonal domains}},\ }\href
  {https://doi.org/10.1088/1742-5468/2010/12/P12033} {\bibfield  {journal}
  {\bibinfo  {journal} {Journal of Statistical Mechanics: Theory and
  Experiment}\ }\textbf {\bibinfo {volume} {2010}},\ \bibinfo {pages} {12033}
  (\bibinfo {year} {2010})},\ \Eprint {https://arxiv.org/abs/1007.5356}
  {arXiv:1007.5356 [cond-mat.str-el]} \BibitemShut {NoStop}%
\bibitem [{\citenamefont {{Sirois}}\ \emph {et~al.}(2021)\citenamefont
  {{Sirois}}, \citenamefont {{Fournier}}, \citenamefont {{Leduc}},\ and\
  \citenamefont {{Witczak-Krempa}}}]{Sirois2021IQHCornerEE}%
  \BibitemOpen
  \bibfield  {author} {\bibinfo {author} {\bibfnamefont {B.}~\bibnamefont
  {{Sirois}}}, \bibinfo {author} {\bibfnamefont {L.~M.}\ \bibnamefont
  {{Fournier}}}, \bibinfo {author} {\bibfnamefont {J.}~\bibnamefont
  {{Leduc}}},\ and\ \bibinfo {author} {\bibfnamefont {W.}~\bibnamefont
  {{Witczak-Krempa}}},\ }\bibfield  {title} {\bibinfo {title} {{Geometric
  entanglement in integer quantum Hall states}},\ }\href
  {https://doi.org/10.1103/PhysRevB.103.115115} {\bibfield  {journal} {\bibinfo
   {journal} {\prb}\ }\textbf {\bibinfo {volume} {103}},\ \bibinfo {eid}
  {115115} (\bibinfo {year} {2021})}\BibitemShut {NoStop}%
\bibitem [{\citenamefont {Rudolph}(2005)}]{Rudolph2005CCNR}%
  \BibitemOpen
  \bibfield  {author} {\bibinfo {author} {\bibfnamefont {O.}~\bibnamefont
  {Rudolph}},\ }\bibfield  {title} {\bibinfo {title} {Further results on the
  cross norm criterion for separability},\ }\href@noop {} {\bibfield  {journal}
  {\bibinfo  {journal} {Quantum Information Processing}\ }\textbf {\bibinfo
  {volume} {4}},\ \bibinfo {pages} {219} (\bibinfo {year} {2005})},\ \Eprint
  {https://arxiv.org/abs/quant-ph/0202121} {arXiv:quant-ph/0202121 [quant-ph]}
  \BibitemShut {NoStop}%
\bibitem [{\citenamefont {Chen}\ and\ \citenamefont {Wu}(2003)}]{Chen2002CCNR}%
  \BibitemOpen
  \bibfield  {author} {\bibinfo {author} {\bibfnamefont {K.}~\bibnamefont
  {Chen}}\ and\ \bibinfo {author} {\bibfnamefont {L.-A.}\ \bibnamefont {Wu}},\
  }\bibfield  {title} {\bibinfo {title} {A matrix realignment method for
  recognizing entanglement},\ }\href@noop {} {\bibfield  {journal} {\bibinfo
  {journal} {Quantum Info. Comput.}\ }\textbf {\bibinfo {volume} {3}},\
  \bibinfo {pages} {193–202} (\bibinfo {year} {2003})},\ \Eprint
  {https://arxiv.org/abs/quant-ph/0205017} {arXiv:quant-ph/0205017 [quant-ph]}
  \BibitemShut {NoStop}%
\bibitem [{\citenamefont {{Rudolph}}(2003)}]{Rudolph2003CCNRProperties}%
  \BibitemOpen
  \bibfield  {author} {\bibinfo {author} {\bibfnamefont {O.}~\bibnamefont
  {{Rudolph}}},\ }\bibfield  {title} {\bibinfo {title} {{Some properties of the
  computable cross-norm criterion for separability}},\ }\href
  {https://doi.org/10.1103/PhysRevA.67.032312} {\bibfield  {journal} {\bibinfo
  {journal} {\pra}\ }\textbf {\bibinfo {volume} {67}},\ \bibinfo {eid} {032312}
  (\bibinfo {year} {2003})},\ \Eprint {https://arxiv.org/abs/quant-ph/0212047}
  {arXiv:quant-ph/0212047 [quant-ph]} \BibitemShut {NoStop}%
\bibitem [{\citenamefont {Horodecki}\ \emph {et~al.}(2009)\citenamefont
  {Horodecki}, \citenamefont {Horodecki}, \citenamefont {Horodecki},\ and\
  \citenamefont {Horodecki}}]{rmp_entan09}%
  \BibitemOpen
  \bibfield  {author} {\bibinfo {author} {\bibfnamefont {R.}~\bibnamefont
  {Horodecki}}, \bibinfo {author} {\bibfnamefont {P.}~\bibnamefont
  {Horodecki}}, \bibinfo {author} {\bibfnamefont {M.}~\bibnamefont
  {Horodecki}},\ and\ \bibinfo {author} {\bibfnamefont {K.}~\bibnamefont
  {Horodecki}},\ }\bibfield  {title} {\bibinfo {title} {Quantum entanglement},\
  }\href {https://doi.org/10.1103/RevModPhys.81.865} {\bibfield  {journal}
  {\bibinfo  {journal} {Rev. Mod. Phys.}\ }\textbf {\bibinfo {volume} {81}},\
  \bibinfo {pages} {865} (\bibinfo {year} {2009})}\BibitemShut {NoStop}%
\bibitem [{\citenamefont {{Aubrun}}\ and\ \citenamefont
  {{Nechita}}(2012)}]{Aubrun2012RandStateCCNR}%
  \BibitemOpen
  \bibfield  {author} {\bibinfo {author} {\bibfnamefont {G.}~\bibnamefont
  {{Aubrun}}}\ and\ \bibinfo {author} {\bibfnamefont {I.}~\bibnamefont
  {{Nechita}}},\ }\bibfield  {title} {\bibinfo {title} {{Realigning random
  states}},\ }\href {https://doi.org/10.1063/1.4759115} {\bibfield  {journal}
  {\bibinfo  {journal} {Journal of Mathematical Physics}\ }\textbf {\bibinfo
  {volume} {53}},\ \bibinfo {pages} {102210} (\bibinfo {year} {2012})},\
  \Eprint {https://arxiv.org/abs/1203.3974} {arXiv:1203.3974 [math.PR]}
  \BibitemShut {NoStop}%
\bibitem [{\citenamefont {{Collins}}\ and\ \citenamefont
  {{Nechita}}(2016)}]{Collins2016RMT}%
  \BibitemOpen
  \bibfield  {author} {\bibinfo {author} {\bibfnamefont {B.}~\bibnamefont
  {{Collins}}}\ and\ \bibinfo {author} {\bibfnamefont {I.}~\bibnamefont
  {{Nechita}}},\ }\bibfield  {title} {\bibinfo {title} {{Random matrix
  techniques in quantum information theory}},\ }\href
  {https://doi.org/10.1063/1.4936880} {\bibfield  {journal} {\bibinfo
  {journal} {Journal of Mathematical Physics}\ }\textbf {\bibinfo {volume}
  {57}},\ \bibinfo {eid} {015215} (\bibinfo {year} {2016})},\ \Eprint
  {https://arxiv.org/abs/1509.04689} {arXiv:1509.04689 [quant-ph]} \BibitemShut
  {NoStop}%
\bibitem [{\citenamefont {{Yin}}\ and\ \citenamefont
  {{Liu}}(2022)}]{Yin2022CFTCCNR}%
  \BibitemOpen
  \bibfield  {author} {\bibinfo {author} {\bibfnamefont {C.}~\bibnamefont
  {{Yin}}}\ and\ \bibinfo {author} {\bibfnamefont {Z.}~\bibnamefont {{Liu}}},\
  }\bibfield  {title} {\bibinfo {title} {{Universal entanglement and
  correlation measure in two-dimensional conformal field theory}},\ }\href@noop
  {} {\bibfield  {journal} {\bibinfo  {journal} {arXiv e-prints}\ ,\ \bibinfo
  {eid} {arXiv:2211.11952}} (\bibinfo {year} {2022})},\ \Eprint
  {https://arxiv.org/abs/2211.11952} {arXiv:2211.11952 [cond-mat.stat-mech]}
  \BibitemShut {NoStop}%
\bibitem [{\citenamefont {Liu}\ \emph {et~al.}(2022)\citenamefont {Liu},
  \citenamefont {Tang}, \citenamefont {Dai}, \citenamefont {Liu}, \citenamefont
  {Chen},\ and\ \citenamefont {Ma}}]{Liu2022CCNR}%
  \BibitemOpen
  \bibfield  {author} {\bibinfo {author} {\bibfnamefont {Z.}~\bibnamefont
  {Liu}}, \bibinfo {author} {\bibfnamefont {Y.}~\bibnamefont {Tang}}, \bibinfo
  {author} {\bibfnamefont {H.}~\bibnamefont {Dai}}, \bibinfo {author}
  {\bibfnamefont {P.}~\bibnamefont {Liu}}, \bibinfo {author} {\bibfnamefont
  {S.}~\bibnamefont {Chen}},\ and\ \bibinfo {author} {\bibfnamefont
  {X.}~\bibnamefont {Ma}},\ }\bibfield  {title} {\bibinfo {title} {Detecting
  entanglement in quantum many-body systems via permutation moments},\ }\href
  {https://doi.org/10.1103/PhysRevLett.129.260501} {\bibfield  {journal}
  {\bibinfo  {journal} {Phys. Rev. Lett.}\ }\textbf {\bibinfo {volume} {129}},\
  \bibinfo {pages} {260501} (\bibinfo {year} {2022})}\BibitemShut {NoStop}%
\bibitem [{\citenamefont {{Milekhin}}\ \emph {et~al.}(2022)\citenamefont
  {{Milekhin}}, \citenamefont {{Rath}},\ and\ \citenamefont
  {{Weng}}}]{holo2022CCNR}%
  \BibitemOpen
  \bibfield  {author} {\bibinfo {author} {\bibfnamefont {A.}~\bibnamefont
  {{Milekhin}}}, \bibinfo {author} {\bibfnamefont {P.}~\bibnamefont {{Rath}}},\
  and\ \bibinfo {author} {\bibfnamefont {W.}~\bibnamefont {{Weng}}},\
  }\bibfield  {title} {\bibinfo {title} {{Computable Cross Norm in Tensor
  Networks and Holography}},\ }\href@noop {} {\bibfield  {journal} {\bibinfo
  {journal} {arXiv e-prints}\ ,\ \bibinfo {eid} {arXiv:2212.11978}} (\bibinfo
  {year} {2022})},\ \Eprint {https://arxiv.org/abs/2212.11978}
  {arXiv:2212.11978 [hep-th]} \BibitemShut {NoStop}%
\bibitem [{\citenamefont {{Dutta}}\ and\ \citenamefont
  {{Faulkner}}(2021)}]{Dutta2021ReflectedEntropy}%
  \BibitemOpen
  \bibfield  {author} {\bibinfo {author} {\bibfnamefont {S.}~\bibnamefont
  {{Dutta}}}\ and\ \bibinfo {author} {\bibfnamefont {T.}~\bibnamefont
  {{Faulkner}}},\ }\bibfield  {title} {\bibinfo {title} {{A canonical
  purification for the entanglement wedge cross-section}},\ }\href
  {https://doi.org/10.1007/JHEP03(2021)178} {\bibfield  {journal} {\bibinfo
  {journal} {Journal of High Energy Physics}\ }\textbf {\bibinfo {volume}
  {2021}},\ \bibinfo {eid} {178} (\bibinfo {year} {2021})}\BibitemShut
  {NoStop}%
\bibitem [{\citenamefont {Witten}(1989)}]{Witten1989CSTheory}%
  \BibitemOpen
  \bibfield  {author} {\bibinfo {author} {\bibfnamefont {E.}~\bibnamefont
  {Witten}},\ }\bibfield  {title} {\bibinfo {title} {{Quantum field theory and
  the Jones polynomial}},\ }\href {https://doi.org/10.1007/BF01217730}
  {\bibfield  {journal} {\bibinfo  {journal} {Communications in Mathematical
  Physics}\ }\textbf {\bibinfo {volume} {121}},\ \bibinfo {pages} {351}
  (\bibinfo {year} {1989})}\BibitemShut {NoStop}%
\bibitem [{\citenamefont {Nielsen}\ and\ \citenamefont
  {Chuang}(2010)}]{NielsenChuangBook}%
  \BibitemOpen
  \bibfield  {author} {\bibinfo {author} {\bibfnamefont {M.~A.}\ \bibnamefont
  {Nielsen}}\ and\ \bibinfo {author} {\bibfnamefont {I.~L.}\ \bibnamefont
  {Chuang}},\ }\href {https://doi.org/10.1017/CBO9780511976667} {\emph
  {\bibinfo {title} {Quantum Computation and Quantum Information: 10th
  Anniversary Edition}}}\ (\bibinfo  {publisher} {Cambridge University Press},\
  \bibinfo {year} {2010})\BibitemShut {NoStop}%
\bibitem [{\citenamefont {{Kitaev}}(2006)}]{Kitaev2006Honeycomb}%
  \BibitemOpen
  \bibfield  {author} {\bibinfo {author} {\bibfnamefont {A.}~\bibnamefont
  {{Kitaev}}},\ }\bibfield  {title} {\bibinfo {title} {{Anyons in an exactly
  solved model and beyond}},\ }\href
  {https://doi.org/10.1016/j.aop.2005.10.005} {\bibfield  {journal} {\bibinfo
  {journal} {Annals of Physics}\ }\textbf {\bibinfo {volume} {321}},\ \bibinfo
  {pages} {2} (\bibinfo {year} {2006})},\ \Eprint
  {https://arxiv.org/abs/cond-mat/0506438} {arXiv:cond-mat/0506438
  [cond-mat.mes-hall]} \BibitemShut {NoStop}%
\bibitem [{\citenamefont {{Bonderson}}(2012)}]{BondersonThesis}%
  \BibitemOpen
  \bibfield  {author} {\bibinfo {author} {\bibfnamefont {P.~H.}\ \bibnamefont
  {{Bonderson}}},\ }\emph {\bibinfo {title} {{Non-Abelian Anyons and
  Interferometry}}},\ \href@noop {} {Ph.D. thesis},\ \bibinfo  {school}
  {California Institute of Technology} (\bibinfo {year} {2012})\BibitemShut
  {NoStop}%
\bibitem [{\citenamefont {Di~Francesco}\ \emph {et~al.}(2012)\citenamefont
  {Di~Francesco}, \citenamefont {Mathieu},\ and\ \citenamefont
  {S{\'e}n{\'e}chal}}]{CFTBook}%
  \BibitemOpen
  \bibfield  {author} {\bibinfo {author} {\bibfnamefont {P.}~\bibnamefont
  {Di~Francesco}}, \bibinfo {author} {\bibfnamefont {P.}~\bibnamefont
  {Mathieu}},\ and\ \bibinfo {author} {\bibfnamefont {D.}~\bibnamefont
  {S{\'e}n{\'e}chal}},\ }\href@noop {} {\emph {\bibinfo {title} {Conformal
  field theory}}}\ (\bibinfo  {publisher} {Springer Science \& Business
  Media},\ \bibinfo {year} {2012})\BibitemShut {NoStop}%
\bibitem [{\citenamefont {{Nayak}}\ \emph {et~al.}(2008)\citenamefont
  {{Nayak}}, \citenamefont {{Simon}}, \citenamefont {{Stern}}, \citenamefont
  {{Freedman}},\ and\ \citenamefont {{Das Sarma}}}]{Nayak2008RMP}%
  \BibitemOpen
  \bibfield  {author} {\bibinfo {author} {\bibfnamefont {C.}~\bibnamefont
  {{Nayak}}}, \bibinfo {author} {\bibfnamefont {S.~H.}\ \bibnamefont
  {{Simon}}}, \bibinfo {author} {\bibfnamefont {A.}~\bibnamefont {{Stern}}},
  \bibinfo {author} {\bibfnamefont {M.}~\bibnamefont {{Freedman}}},\ and\
  \bibinfo {author} {\bibfnamefont {S.}~\bibnamefont {{Das Sarma}}},\
  }\bibfield  {title} {\bibinfo {title} {{Non-Abelian anyons and topological
  quantum computation}},\ }\href {https://doi.org/10.1103/RevModPhys.80.1083}
  {\bibfield  {journal} {\bibinfo  {journal} {Reviews of Modern Physics}\
  }\textbf {\bibinfo {volume} {80}},\ \bibinfo {pages} {1083} (\bibinfo {year}
  {2008})},\ \Eprint {https://arxiv.org/abs/0707.1889} {arXiv:0707.1889
  [cond-mat.str-el]} \BibitemShut {NoStop}%
\bibitem [{\citenamefont {{Fattal}}\ \emph {et~al.}(2004)\citenamefont
  {{Fattal}}, \citenamefont {{Cubitt}}, \citenamefont {{Yamamoto}},
  \citenamefont {{Bravyi}},\ and\ \citenamefont
  {{Chuang}}}]{Chuang2004StabilizerEE}%
  \BibitemOpen
  \bibfield  {author} {\bibinfo {author} {\bibfnamefont {D.}~\bibnamefont
  {{Fattal}}}, \bibinfo {author} {\bibfnamefont {T.~S.}\ \bibnamefont
  {{Cubitt}}}, \bibinfo {author} {\bibfnamefont {Y.}~\bibnamefont
  {{Yamamoto}}}, \bibinfo {author} {\bibfnamefont {S.}~\bibnamefont
  {{Bravyi}}},\ and\ \bibinfo {author} {\bibfnamefont {I.~L.}\ \bibnamefont
  {{Chuang}}},\ }\bibfield  {title} {\bibinfo {title} {{Entanglement in the
  stabilizer formalism}},\ }\href@noop {} {\bibfield  {journal} {\bibinfo
  {journal} {arXiv e-prints}\ ,\ \bibinfo {eid} {quant-ph/0406168}} (\bibinfo
  {year} {2004})},\ \Eprint {https://arxiv.org/abs/quant-ph/0406168}
  {arXiv:quant-ph/0406168 [quant-ph]} \BibitemShut {NoStop}%
\bibitem [{\citenamefont {Zou}\ and\ \citenamefont
  {Haah}(2016)}]{spurious_cylinder16}%
  \BibitemOpen
  \bibfield  {author} {\bibinfo {author} {\bibfnamefont {L.}~\bibnamefont
  {Zou}}\ and\ \bibinfo {author} {\bibfnamefont {J.}~\bibnamefont {Haah}},\
  }\bibfield  {title} {\bibinfo {title} {Spurious long-range entanglement and
  replica correlation length},\ }\href
  {https://doi.org/10.1103/PhysRevB.94.075151} {\bibfield  {journal} {\bibinfo
  {journal} {Phys. Rev. B}\ }\textbf {\bibinfo {volume} {94}},\ \bibinfo
  {pages} {075151} (\bibinfo {year} {2016})}\BibitemShut {NoStop}%
\bibitem [{\citenamefont {Williamson}\ \emph {et~al.}(2019)\citenamefont
  {Williamson}, \citenamefont {Dua},\ and\ \citenamefont
  {Cheng}}]{spurious_subsystem19}%
  \BibitemOpen
  \bibfield  {author} {\bibinfo {author} {\bibfnamefont {D.~J.}\ \bibnamefont
  {Williamson}}, \bibinfo {author} {\bibfnamefont {A.}~\bibnamefont {Dua}},\
  and\ \bibinfo {author} {\bibfnamefont {M.}~\bibnamefont {Cheng}},\ }\bibfield
   {title} {\bibinfo {title} {Spurious topological entanglement entropy from
  subsystem symmetries},\ }\href
  {https://doi.org/10.1103/PhysRevLett.122.140506} {\bibfield  {journal}
  {\bibinfo  {journal} {Phys. Rev. Lett.}\ }\textbf {\bibinfo {volume} {122}},\
  \bibinfo {pages} {140506} (\bibinfo {year} {2019})}\BibitemShut {NoStop}%
\bibitem [{\citenamefont {{Kato}}\ and\ \citenamefont
  {{Brand{\~a}o}}(2020)}]{spurious_nonSPT20}%
  \BibitemOpen
  \bibfield  {author} {\bibinfo {author} {\bibfnamefont {K.}~\bibnamefont
  {{Kato}}}\ and\ \bibinfo {author} {\bibfnamefont {F.~G.~S.~L.}\ \bibnamefont
  {{Brand{\~a}o}}},\ }\bibfield  {title} {\bibinfo {title} {{Toy model of
  boundary states with spurious topological entanglement entropy}},\ }\href
  {https://doi.org/10.1103/PhysRevResearch.2.032005} {\bibfield  {journal}
  {\bibinfo  {journal} {Physical Review Research}\ }\textbf {\bibinfo {volume}
  {2}},\ \bibinfo {eid} {032005} (\bibinfo {year} {2020})},\ \Eprint
  {https://arxiv.org/abs/1911.09819} {arXiv:1911.09819 [quant-ph]} \BibitemShut
  {NoStop}%
\bibitem [{\citenamefont {Kim}\ \emph {et~al.}(2023)\citenamefont {Kim},
  \citenamefont {Levin}, \citenamefont {Lin}, \citenamefont {Ranard},\ and\
  \citenamefont {Shi}}]{spurious_bound23}%
  \BibitemOpen
  \bibfield  {author} {\bibinfo {author} {\bibfnamefont {I.~H.}\ \bibnamefont
  {Kim}}, \bibinfo {author} {\bibfnamefont {M.}~\bibnamefont {Levin}}, \bibinfo
  {author} {\bibfnamefont {T.-C.}\ \bibnamefont {Lin}}, \bibinfo {author}
  {\bibfnamefont {D.}~\bibnamefont {Ranard}},\ and\ \bibinfo {author}
  {\bibfnamefont {B.}~\bibnamefont {Shi}},\ }\bibfield  {title} {\bibinfo
  {title} {{Universal lower bound on topological entanglement entropy}},\
  }\href@noop {} {\  (\bibinfo {year} {2023})},\ \Eprint
  {https://arxiv.org/abs/2302.00689} {arXiv:2302.00689 [quant-ph]} \BibitemShut
  {NoStop}%
\bibitem [{\citenamefont {{Brown}}\ \emph {et~al.}(2013)\citenamefont
  {{Brown}}, \citenamefont {{Bartlett}}, \citenamefont {{Doherty}},\ and\
  \citenamefont {{Barrett}}}]{Brown2013TEETwist}%
  \BibitemOpen
  \bibfield  {author} {\bibinfo {author} {\bibfnamefont {B.~J.}\ \bibnamefont
  {{Brown}}}, \bibinfo {author} {\bibfnamefont {S.~D.}\ \bibnamefont
  {{Bartlett}}}, \bibinfo {author} {\bibfnamefont {A.~C.}\ \bibnamefont
  {{Doherty}}},\ and\ \bibinfo {author} {\bibfnamefont {S.~D.}\ \bibnamefont
  {{Barrett}}},\ }\bibfield  {title} {\bibinfo {title} {{Topological
  Entanglement Entropy with a Twist}},\ }\href
  {https://doi.org/10.1103/PhysRevLett.111.220402} {\bibfield  {journal}
  {\bibinfo  {journal} {\prl}\ }\textbf {\bibinfo {volume} {111}},\ \bibinfo
  {eid} {220402} (\bibinfo {year} {2013})},\ \Eprint
  {https://arxiv.org/abs/1303.4455} {arXiv:1303.4455 [quant-ph]} \BibitemShut
  {NoStop}%
\bibitem [{\citenamefont {{Kitaev}}(2003)}]{Kitaev2003ToricCode}%
  \BibitemOpen
  \bibfield  {author} {\bibinfo {author} {\bibfnamefont {A.~Y.}\ \bibnamefont
  {{Kitaev}}},\ }\bibfield  {title} {\bibinfo {title} {{Fault-tolerant quantum
  computation by anyons}},\ }\href
  {https://doi.org/10.1016/S0003-4916(02)00018-0} {\bibfield  {journal}
  {\bibinfo  {journal} {Annals of Physics}\ }\textbf {\bibinfo {volume}
  {303}},\ \bibinfo {pages} {2} (\bibinfo {year} {2003})},\ \Eprint
  {https://arxiv.org/abs/quant-ph/9707021} {arXiv:quant-ph/9707021 [quant-ph]}
  \BibitemShut {NoStop}%
\bibitem [{\citenamefont {{Berthiere}}\ \emph {et~al.}(2021)\citenamefont
  {{Berthiere}}, \citenamefont {{Chen}}, \citenamefont {{Liu}},\ and\
  \citenamefont {{Chen}}}]{Chen2021ReflectedEntropy}%
  \BibitemOpen
  \bibfield  {author} {\bibinfo {author} {\bibfnamefont {C.}~\bibnamefont
  {{Berthiere}}}, \bibinfo {author} {\bibfnamefont {H.}~\bibnamefont {{Chen}}},
  \bibinfo {author} {\bibfnamefont {Y.}~\bibnamefont {{Liu}}},\ and\ \bibinfo
  {author} {\bibfnamefont {B.}~\bibnamefont {{Chen}}},\ }\bibfield  {title}
  {\bibinfo {title} {{Topological reflected entropy in Chern-Simons
  theories}},\ }\href {https://doi.org/10.1103/PhysRevB.103.035149} {\bibfield
  {journal} {\bibinfo  {journal} {\prb}\ }\textbf {\bibinfo {volume} {103}},\
  \bibinfo {eid} {035149} (\bibinfo {year} {2021})},\ \Eprint
  {https://arxiv.org/abs/2008.07950} {arXiv:2008.07950 [hep-th]} \BibitemShut
  {NoStop}%
\bibitem [{\citenamefont {{Siva}}\ \emph {et~al.}(2022)\citenamefont {{Siva}},
  \citenamefont {{Zou}}, \citenamefont {{Soejima}}, \citenamefont {{Mong}},\
  and\ \citenamefont {{Zaletel}}}]{Zaletel2022MarkovGap}%
  \BibitemOpen
  \bibfield  {author} {\bibinfo {author} {\bibfnamefont {K.}~\bibnamefont
  {{Siva}}}, \bibinfo {author} {\bibfnamefont {Y.}~\bibnamefont {{Zou}}},
  \bibinfo {author} {\bibfnamefont {T.}~\bibnamefont {{Soejima}}}, \bibinfo
  {author} {\bibfnamefont {R.~S.~K.}\ \bibnamefont {{Mong}}},\ and\ \bibinfo
  {author} {\bibfnamefont {M.~P.}\ \bibnamefont {{Zaletel}}},\ }\bibfield
  {title} {\bibinfo {title} {{Universal tripartite entanglement signature of
  ungappable edge states}},\ }\href
  {https://doi.org/10.1103/PhysRevB.106.L041107} {\bibfield  {journal}
  {\bibinfo  {journal} {\prb}\ }\textbf {\bibinfo {volume} {106}},\ \bibinfo
  {eid} {L041107} (\bibinfo {year} {2022})},\ \Eprint
  {https://arxiv.org/abs/2110.11965} {arXiv:2110.11965 [quant-ph]} \BibitemShut
  {NoStop}%
\bibitem [{\citenamefont {{Fan}}\ \emph {et~al.}(2023)\citenamefont {{Fan}},
  \citenamefont {{Bao}}, \citenamefont {{Altman}},\ and\ \citenamefont
  {{Vishwanath}}}]{Fan2023TODecoherence}%
  \BibitemOpen
  \bibfield  {author} {\bibinfo {author} {\bibfnamefont {R.}~\bibnamefont
  {{Fan}}}, \bibinfo {author} {\bibfnamefont {Y.}~\bibnamefont {{Bao}}},
  \bibinfo {author} {\bibfnamefont {E.}~\bibnamefont {{Altman}}},\ and\
  \bibinfo {author} {\bibfnamefont {A.}~\bibnamefont {{Vishwanath}}},\
  }\bibfield  {title} {\bibinfo {title} {{Diagnostics of mixed-state
  topological order and breakdown of quantum memory}},\ }\href@noop {}
  {\bibfield  {journal} {\bibinfo  {journal} {arXiv e-prints}\ ,\ \bibinfo
  {eid} {arXiv:2301.05689}} (\bibinfo {year} {2023})},\ \Eprint
  {https://arxiv.org/abs/2301.05689} {arXiv:2301.05689 [quant-ph]} \BibitemShut
  {NoStop}%
\bibitem [{\citenamefont {{Sohal}}\ and\ \citenamefont
  {{Ryu}}(2023)}]{Ryu2023Diagrammatic}%
  \BibitemOpen
  \bibfield  {author} {\bibinfo {author} {\bibfnamefont {R.}~\bibnamefont
  {{Sohal}}}\ and\ \bibinfo {author} {\bibfnamefont {S.}~\bibnamefont
  {{Ryu}}},\ }\bibfield  {title} {\bibinfo {title} {{Entanglement in
  tripartitions of topological orders: A diagrammatic approach}},\ }\href
  {https://doi.org/10.1103/PhysRevB.108.045104} {\bibfield  {journal} {\bibinfo
   {journal} {\prb}\ }\textbf {\bibinfo {volume} {108}},\ \bibinfo {eid}
  {045104} (\bibinfo {year} {2023})},\ \Eprint
  {https://arxiv.org/abs/2301.07763} {arXiv:2301.07763 [cond-mat.str-el]}
  \BibitemShut {NoStop}%
\bibitem [{\citenamefont {Liu}\ \emph {et~al.}(2023)\citenamefont {Liu},
  \citenamefont {Kusuki}, \citenamefont {Kudler-Flam}, \citenamefont {Sohal},\
  and\ \citenamefont {Ryu}}]{Ryu2023Trisection}%
  \BibitemOpen
  \bibfield  {author} {\bibinfo {author} {\bibfnamefont {Y.}~\bibnamefont
  {Liu}}, \bibinfo {author} {\bibfnamefont {Y.}~\bibnamefont {Kusuki}},
  \bibinfo {author} {\bibfnamefont {J.}~\bibnamefont {Kudler-Flam}}, \bibinfo
  {author} {\bibfnamefont {R.}~\bibnamefont {Sohal}},\ and\ \bibinfo {author}
  {\bibfnamefont {S.}~\bibnamefont {Ryu}},\ }\bibfield  {title} {\bibinfo
  {title} {{Multipartite entanglement in two-dimensional chiral topological
  liquids}},\ }\href@noop {} {\bibfield  {journal} {\bibinfo  {journal} {arXiv
  e-prints}\ ,\ \bibinfo {pages} {arXiv:2301.07130}} (\bibinfo {year}
  {2023})}\BibitemShut {NoStop}%
\end{thebibliography}%
\end{document}